\documentclass[11pt,a4paper]{article}
\usepackage[a4paper,margin=1in]{geometry}
\usepackage{algorithm,algpseudocode,booktabs,comment}
\usepackage{mathtools}

\usepackage{amsmath,amsthm}
\usepackage[table]{xcolor}
\usepackage{array}
\usepackage{booktabs}
\usepackage{diagbox}
\usepackage{tikz}
\usetikzlibrary{arrows.meta,positioning}

\usepackage{amssymb}
\newtheorem{definition}{Definition}
\newcommand{\Description}[1]{}  

\definecolor{cSAME}{RGB}{200,230,201}
\definecolor{cDIFF}{RGB}{197,217,241}
\definecolor{cSTUCK}{RGB}{245,198,203}
\definecolor{cSPEC}{RGB}{255,224,178}

\newcolumntype{Y}{>{\centering\arraybackslash}p{0.95cm}}
\usepackage{graphicx}
\usepackage{makecell}
\usepackage{tabularx}
\usepackage{thmtools,thm-restate}

\usetikzlibrary{arrows.meta,shapes.geometric,positioning,fit,backgrounds}

\definecolor{cState}{RGB}{123,31,162}   
\definecolor{cCond}{RGB}{245,124,0}     
\definecolor{cAct}{RGB}{56,142,60}      
\definecolor{cA}{RGB}{56,142,60}        
\definecolor{cB}{RGB}{25,118,210}       
\definecolor{cX}{RGB}{211,47,47}        
\definecolor{cS}{RGB}{245,124,0}        

\usepackage{textcomp}
\usepackage{xcolor}
\usepackage{enumitem}
\usepackage{float}
\usepackage[font=scriptsize]{caption}
\usepackage{wrapfig}
\usepackage[labelformat=simple]{subcaption}

\usepackage{pgf,tikz,pgfplots}
\pgfplotsset{compat=1.15}
\usepackage{mathrsfs}
\usepackage{hyperref}
\usepackage{multicol}
\usepackage{multirow}

\usepackage{xcolor}
\usepackage[linesnumbered,ruled,vlined,algo2e]{algorithm2e}
\usepackage[export]{adjustbox}

\newcommand{\fsync}{$\mathcal {FSYNC}$}
\newcommand{\ssync}{$\mathcal {SSYNC}$}
\newcommand{\async}{$\mathcal {ASYNC}$}
\newcommand{\CH}{\ensuremath{\mathcal{CH}}}
\newcommand{\CHr}{\ensuremath{\mathcal{CH}_r}}

\SetCommentSty{mycommfont}

\usepackage{balance} 

\title{Asynchronous Indulgent Gathering of Opaque Robots}

\author{
Subhajit Pramanick\\
\small University of Wroclaw, Poland\\
\small\texttt{suvo.iitg17@gmail.com}
\and
Saswata Jana\\
\small IIT Guwahati, India\\
\small\texttt{saswatajana@iitg.ac.in}
\and 
Partha Sarathi Mandal\\
\small IIT Guwahati, India\\
\small\texttt{psm@iitg.ac.in}
\and
Gokarna Sharma\\
\small Kent State University, USA\\
\small\texttt{gsharma2@kent.edu}
}
\date{}

\begin{document}

\maketitle

\begin{abstract}
We consider the benchmarking problem of gathering in an $(N,f)$-fault system consisting of $N$ robots, of which at most $f$ may fail at any execution. 
The {\em traditional} version asks $N-f$ non-faulty robots to gather at a point not known a priori.  
The recently studied {\em indulgent} version asks $N-f$ non-faulty robots to gather at the location of a faulty one.
Indulgent gathering is impossible in a $(2,1)$-crash system in the $\mathcal{LUMI}$ model under semi-synchrony, even when infinitely many colors and coordinate axis agreement are available.   
In this paper, we establish four results for indulgent gathering under asynchrony in $\mathcal{LUMI}$, considering a {\em mobility fault} model in which a robot crash only impacts its mobility but not the operation of the light. 
 We first show that indulgent gathering is impossible in a $(2,1)$-mobility system using 2-colored lights, and then give an optimal solution w.r.t. the number of colors using 3-colored lights. We then give two deterministic algorithms for an $(N,f)$-mobility system, $f<N$, $N,f$ unknown, that exhibit a nice time-color trade-off: One with time $O(N)$ using 7-colored lights and another with time $O(\max\{\ell,f\})$ using 26-colored lights, where $\ell< N$ is the number of distinct convex layers of robot positions in the initial configuration. 
Our algorithms for an $(N,f)$-mobility system are the first to be analysed time complexity, can withstand obstructed visibility (opaque robot model) and asynchronous scheduling.
Interestingly, our algorithms also solve traditional gathering.
\end{abstract}



\section{Introduction}
\label{sec:Introduction}

\subsection{Background and Motivation}
The gathering problem is one of the benchmarking tasks (see \cite{Flocchini19Book} and the references therein).
The problem asks $N$ robots in the Euclidean plane to reach an unknown common point. 
In this context, the {\em oblivious robot} ($\mathcal{OBLOT}$) model considers robots to be point-shaped, autonomous, anonymous, disoriented (no global coordinate system), silent, oblivious, and have unlimited and unobstructed visibility (i.e., robots are transparent), executing identical Look–Compute–Move (LCM) cycles. 
Synchrony is classified into three standard levels: in fully synchronous ({\fsync}), all robots execute an LCM cycle in every round; in semi-synchronous ({\ssync}), an arbitrary non-empty subset is activated each round; and in asynchronous ({\async}), robots perform their LCM cycles of arbitrary but finite duration, starting from and ending at arbitrary times.
Suzuki and Yamashita \cite{SuzukiY99} proved that, under {\ssync}, deterministic gathering is impossible for two robots in $\mathcal{OBLOT}$ (this impossibility extends to {\async}). 

\begin{table}[t]
\centering

\resizebox{\linewidth}{!}{
\begin{tabular}{c|c|c|c|c|c|c}
\hline
{\bf Paper} & {\bf Synchrony} & {\bf System} & {\bf Gathering} & {\bf Opacity} & {\bf Model} & {\bf Time}\\
\Xhline{1pt}
\cite{SuzukiY99} & {\ssync} & $(2,0)$-fault & Impossible & -- & $\mathcal{OBLOT}$ & -- \\

\cite{doi:10.1137/050645221} & {\async} & $(3,1)$-Byzantine & Impossible & No & $\mathcal{OBLOT}$ & --\\
\hline
\cite{HeribanDT18} & {\async} & $(2,0)$-fault & Possible & -- & 2-$\mathcal{LUMI}$ & -- \\
\hline
\cite{doi:10.1137/050645221} & {\ssync} & $(N,1)$-crash & Possible & No & $\mathcal{OBLOT}$ & --\\
\cite{doi:10.1137/050645221}$^{*}$ & {\fsync} & $(N,f)$-Byzantine & Possible & No & $\mathcal{OBLOT}$ & --\\
\hline 
\end{tabular}
}

\vspace{2pt}

\caption{Results on {\bf traditional} gathering.$^{*}N\ge3f+1$.
$\mathcal{OBLOT}\equiv$ 1-$\mathcal{LUMI}$.
`--' in time means finite time.}
\label{table:comparative}

\vspace{8pt}

\resizebox{\linewidth}{!}{
\begin{tabular}{c|c|c|c|c|c|c}
\hline
{\bf Paper} & {\bf Synchrony} & {\bf System} & {\bf Gathering} & {\bf Opacity} & {\bf Model} & {\bf Time}\\
\Xhline{1pt}
\cite{BRAMAS202363} & {\ssync} & $(2,1)$-crash & Impossible & -- & $\infty$-$\mathcal{LUMI}$ & --\\
{\bf This} & {\async} & $(2,1)$-mobility & Impossible & -- & 2-$\mathcal{LUMI}$ & --\\
\hline
{\bf This} & {\async} & $(2,1)$-mobility & Possible & -- & 3-$\mathcal{LUMI}$ & --\\
{\bf This}$^{\dag}$ & {\async} & $(N,f)$-mobility & Possible & Yes & 7-$\mathcal{LUMI}$ & $O(N)$\\
{\bf This}$^{\dag}$ & {\async} & $(N,f)$-mobility & Possible & Yes & 26-$\mathcal{LUMI}$ & $O(\max\{\ell,f\})$\\
\hline
\end{tabular}
}

\caption{Our contributions and known results on {\bf indulgent} gathering.
$^{\dag}f<N$, and $N,f$ unknown. $\ell<N$ is the number of distinct convex polygon layers (defined formally in Sec.~\ref{sec:Nalgorithm}) of robots' positions in the initial configuration. `--' in time means finite time.}
\label{table:comparative2}

\end{table}

Subsequently, researchers turned their attention to ways to circumvent this impossibility. A model considered is the {\em luminous robot} ($\mathcal{LUMI}$) model, which endows a robot with an externally visible \emph{light}. 
Each robot can emit one of $\psi$ (prefixed) colors that remain visible across cycles, while all other properties remain identical to $\mathcal{OBLOT}$. 
We denote this model by $\psi$-$\mathcal{LUMI}$. Notice that $\mathcal{OBLOT}\equiv$ 1-$\mathcal{LUMI}$.  
Heriban {\it et al.} \cite{HeribanDT18} optimally (w.r.t. the number of colors) solved gathering of two robots under {\async} in 2-$\mathcal{LUMI}$, circumventing the impossibility \cite{SuzukiY99}.

In practical applications such as rescue or exploration, robots may fail during execution. 
Faults are typically \emph{transient} (state corruption), \emph{crash} (unexpected halt), or \emph{Byzantine} (arbitrary behavior). 
In an $(N,f)$-fault system, at most $f$ of the $N$ robots may fail; restricting the fault type yields $(N,f)$-crash, $(N,f)$-Byzantine, or $(N,f)$-transient systems. 
A fault-tolerant algorithm must guarantee that all non-faulty robots complete their tasks. The {\em traditional} gathering problem asks $N-f$ non-faulty robots to gather at an unknown common point. 
For traditional gathering, Agmon and Peleg~\cite{doi:10.1137/050645221} proved impossibility in $\mathcal{OBLOT}$ {\async} model for $(3,1)$-Byzantine systems, but provided algorithms for $(N,1)$-crash under {\ssync} and $(N,f)$-Byzantine under {\fsync} when $N \geq 3f+1$ (see Table \ref{table:comparative}).
The {\em indulgent} gathering problem requires one step more and asks $N-f$ non-faulty robots to gather at the location of a faulty robot.
We study indulgent gathering, defined as follows.
\begin{definition}
\label{def:problem}
In an $(N,f)$-fault system under the $\psi$-$\mathcal{LUMI}$ model, starting from an arbitrary configuration on the Euclidean plane, $N-f$ non-faulty opaque robots gather at a single point not known in advance, which, if $f>0$, coincides with the location of some faulty robot.
\end{definition}

Bramas {\it et al.} \cite{BRAMAS202363} proved the impossibility of indulgent gathering 
in a $(2,1)$-crash system under {\ssync} in the $\infty$-$\mathcal{LUMI}$ model.

In this paper, we aim to shed further light on (im)possibilities of indulgent gathering under the \emph{mobility fault model} \cite{POUDEL2021116,PramanickJM25}, where fault impacts only the mobility but not the light (unlike crash faults).
Therefore, in  $\mathcal{OBLOT}$ ($\equiv$ 1-$\mathcal{LUMI}$), an $(N,f)$-mobility system is an $(N,f)$-crash system.
Additionally, we consider \emph{opaque robots} -- two robots cannot see each other if a third lies on their connecting segment. Most of the previous algorithms \cite{doi:10.1137/050645221,BRAMAS202363,HeribanDT18,SuzukiY99} use transparent robots. 
Coordinating opaque robots is harder, as they might be unaware of the whole configuration due to obstructed visibility.

\subsection{More Related Work}\label{app-subsec:related_work}
Gathering is one of the fundamental benchmarking problems \cite{SuzukiY99} in the field of mobile robots. 
The problem has been intensively studied under different models and various assumptions \cite{doi:10.1137/050645221,6681603,DBLP:journals/dc/CiceroneSN18,doi:10.1137/100796534,FLOCCHINI2005147,10.1007/978-3-540-75142-7_24,DBLP:conf/sss/NakaiSW21}. The literature on the problem of gathering is quite voluminous. 
In this paper, we focus on the barebone models with no assumption that are relevant to three key aspects: the asynchronous scheduler, opacity (robots obstructing visibility), and fault-tolerance.
However the literature has different assumptions.
An important common thread in most of the gathering literature that we discussed below is the assumption of a transparent robot model.
Agmon and Peleg \cite{doi:10.1137/050645221} have proposed an algorithm to gather robots with at most one faulty robot under {\ssync}.
Later, Cohen and Peleg \cite{10.1007/978-3-540-30140-0_22} solved the convergence problem (a subproblem of gathering where the point of gathering is the center of gravity of the configuration) under {\async} for an $(N,f)$-crash fault model.
D{\'e}fago {\it et al.} \cite{10.1007/11864219_4} studied the deterministic gathering feasibility under crash and byzantine faults and a larger set of schedulers, such as bounded schedulers.
Bouzid {\it et al.} \cite{6681603} considered gathering in the presence of an arbitrary number of crash faults under {\ssync} model with the assumption of common orientation and strong multiplicity detection.
Later, Bramas and Tixeuil \cite{10.1007/978-3-319-25258-2_22} proposed an algorithm tolerating an arbitrary number of crash faults, which removed the assumption of common orientation. 
Under the {\ssync} model, Bhagat and Mukhopadyaya \cite{10.1145/3007748.3007781} solved the gathering problem with crash-faulty robots by revoking the assumption of orientation and strong multiplicity detection. 

Bhagat and Mukhopadhyaya \cite{10.1007/978-3-319-53007-9_4} solved gathering for $N \geq 5$ robots having $O(1)$ bit of persistent memory with an additional requirement that the maximum distance traversed by the robots needs to be minimized. 
Cicerone {\it et al.} \cite{DBLP:journals/dc/CiceroneSN18} proposed two algorithms for the gathering problem under {\async} w.r.t. optimizing two objectives, the first one to minimize the overall travelled distance performed by all robots and the second to minimize the maximum traveled distance performed by a single robot. 
Pattanayak {\it et al.} \cite{PATTANAYAK2019145} focused on the problem of gathering transparent robots on the Euclidean plane in the presence of crash faults. They studied the problem under {\async} but with the assumption of instantaneous computation. 
Alcantara {\it et al.} \cite{DBLP:journals/dc/AlcantaraCFR19} establishes a formal connection between Look–Compute–Move robot models with lights and classical wait-free shared-memory computing, showing that gathering with hard termination is impossible by reducing it to the impossibility of binary consensus. It then characterizes the solvability of weakened gathering tasks (edge-gathering and 1-gathering) on graphs, providing matching algorithms and impossibility results rooted in combinatorial topology.
Bramas {\it et al.} \cite{10.1007/978-3-030-64348-5_4} introduced another line of research in the gathering problem, called stand-up indulgent rendezvous, which corresponds to $(2,1)$-crash fault model, where the non-faulty robot needs to gather at the location of the faulty one. They proved impossibility under {\ssync} and proposed an algorithm that works under {\fsync}. This work is later extended to the stand-up indulgent gathering by Bramas {\it et al.} \cite{BRAMAS202363}, where multiple robots may crash at the same location, and the remaining non-faulty robots must gather at that crash site. The proposed algorithm works under {\fsync} model. 
Nakai {\it et al.} \cite{DBLP:conf/sss/NakaiSW21} considered gathering under {\async} and solved in the 3-$\mathcal{LUMI}$ model.
Recently, Terai {\it et al.} \cite{TERAI2023241} studied the power of $\mathcal{LUMI}$ model in gathering under {\ssync} with different computational capabilities of the robots such as set-view (where robots can recognize the sets of colors robots have in the same locations) vs the arbitrary-view (where robots can recognize arbitrary colors robots have in the same locations). However, the robot model considered in this paper is still transparent.

The problem of gathering is also examined under the obstructed visibility model \cite{10.1007/978-3-319-11764-5_11} (i.e., the robots are opaque). Bhagat {\it et al.} \cite{BHAGAT201650} presented a deterministic algorithm for gathering opaque robots in the presence of crash faults, however, with an extra assumption of one-axis agreement. 
Okumura {\it et al.} \cite{OKUMURA2023114198} solved the rendezvous problem in the 3-$\mathcal{LUMI}$ model under {\async}.
Another closely related line of literature is the mobility fault in the $\mathcal{LUMI}$ model, which is primarily considered to circumvent the difficulties in solving problems in the $\mathcal{OBLOT}$ model.
A widely popular fundamental problem, known as the \emph{mutual visibility problem}, is well-inspected under {\async} setting in $(N, f)$-mobility system with the assumption of one axis agreement \cite{POUDEL2021116,PRAMANICK2024114723} and later without any axis agreement \cite{PramanickJM25}.

\subsection{Contributions, Challenges and Techniques}
Table~\ref{table:comparative2} presents and compares our contributions. 
Interestingly, our algorithm can solve both indulgent and traditional gatherings, and, unlike previous works such as \cite{doi:10.1137/050645221}, our algorithms need no global knowledge of $N$. 
Our techniques also highlight the role of mobility fault identification, and, to the best of our knowledge, provide the first time-complexity analysis.

To establish the results, we exploit a key property of the mobility fault model. When a robot, activated with color $C_1$, identifies a characteristic $V_1$ (e.g., an interior robot on some convex hull) and switches to the color $C_2$ to move toward a target, a successful move should change its characteristic to $V_2$ (e.g., a corner robot). If the robot remains in $V_1$ with color $C_2$, the fault is revealed. Since its light still works, it can signal failure by switching to a designated color $C_3$, possibly using intermediate color transitions involving other visible robots. 


\vspace{1mm}
\noindent {\bf Challenges.}
We run into four key challenges: opacity, disorientation, asynchrony, and mobility faults. 
Opaque robots block each other's view, making coordination difficult.  
For example, selecting an invariant gathering point is difficult, especially if robots fail en route and obstruct others’ views. 
This issue does not arise in the $\mathcal{OBLOT}$ model as transparency provides a complete view of the system. 
With disorientation, small (possibly symmetric) movements may fail to break collinearities or even create new ones. 
With asynchrony, robots may rely on an outdated view of the system, leading to incorrect destination computations. 
Distinguishing faulty robots remains challenging even with lights.

\vspace{1mm}
\noindent {\bf Our Techniques.}
The impossibility of indulgent gathering in a $(2,1)$-mobility system is proved using the following technique. At a high level, if $r_1$ stays idle, the adversary may either mark $r_2$ as faulty or activate it so that its look phase overlaps with $r_1$’s move phase. If $r_1$ moves to a point $P \neq r_2$, the adversary alternates or synchronizes their activations based on $r_2$’s target. When $r_1$ moves toward $r_2$, the adversary can either make $r_1$ faulty or trigger two consecutive activations of $r_2$ during $r_1$’s (one) movement, thus preventing gathering.
The possibility result for a $(2,1)$-mobility system uses a technique based on a third color, allowing a non-faulty robot to detect whether the other (faulty or not) has already moved under {\async}.

For the $(N,f)$-mobility system, $f<N$, we   
overcome the major challenges by employing the concept of  {\em convex polygonal layers} (see Fig. \ref{fig:layers} and formally defined later in Section \ref{sec:Nalgorithm}). 
At a high level, the outermost polygonal layer consists of robots lying on the perimeter of the convex hull of the entire configuration. 
The next layer is computed by disregarding the robots on the outermost layer and then considering the boundary of the convex hull of the remaining robots. 
Subsequent layers are formed recursively, resulting in a sequence of disjoint convex (global) layers $\mathcal{CH}^0, \cdots, \mathcal{CH}^{l-1}$ for the initial configuration, where $\mathcal{CH}^0$ is the outermost layer, and all others are inner layers.
Among the inner layers, $\mathcal{CH}^{l-1}$ is the innermost layer.
The global layers may differ from those perceived locally by a robot. 
Intuitively, this hierarchical structure enables us to incrementally bring all robots from the inner layers to the outermost layer, starting with the second outermost layer and proceeding layer by layer. 

The idea of layers may appear straightforward w.r.t. global view, but there are two major challenges.
First, robots can not see the global layers completely due to opacity. Second, the layers themselves are dynamic (the initial layers will not be intact), i.e., under {\async}, some robots on a layer may already be moving while others from the same layer have not yet been activated, or become activated during the movement of other robots.
Coordinating robots so that layer-by-layer merging succeeds under {\async} is non-trivial, particularly when ensuring a time bound in the presence of mobility faults and opacity.

We devise two novel techniques (that exhibit a nice time-color trade-off) that move all non-faulty robots from an inner layer to the outermost layer, by the first technique (in 7-$\mathcal{LUMI}$ algorithm) in $O(n_i)$ epochs and by the other technique (in 26-$\mathcal{LUMI}$ algorithm) in just $O(1)$ epochs, where $n_i$ is the number of robots in $i$-th inner layer.
Consequently, the layer-by-layer merging takes $O(N)$ epochs in the first technique, and $O(\ell)$ epochs in the second technique, where $\ell$ is the number of layers in the initial configuration.

We then reach a stage at which all non-faulty interior robots are gathered at the corners of the outermost convex layer, whereas the faulty robots remain in the interior.
Afterwards, each non-faulty robot at a corner of the outermost layer is guided to a region, called the \emph{visible area}, where, after the movement, it remains a corner of the (local) convex hull and can see all distinct positions of the faulty robots.
Using additional colors and the properties of faulty robots, each non-faulty robot then reaches the center of gravity (CG) of the faulty (if any) robots’ convex hull or a corner of the outermost layer after visiting the CG within $O(f)$ epochs in both techniques.
Moreover, we handle the non-trivial case of an initial linear configuration by transforming it into a non-linear configuration through sequentially moving the non-faulty robots in $O(f)$ epochs.
The overall approach achieves a time complexity of $O(N)$ in the first technique with only 7 colors and $O(\max\{\ell,f\})$ epochs in the second technique with 26 colors. The second technique is optimal when both the number of layers ($l$) and the number of faulty robots ($f$) are constants.

\section{Model and Preliminaries}
\label{sec:model}

\noindent {\bf Robot Model.} 
We consider a system of $N$ robots under the $\psi$-$\mathcal{LUMI}$ model, which extends $\mathcal{OBLOT}$ by equipping each robot with an externally visible light emitting $\psi$ colors (one at a time). 
Colors enable robots to communicate and coordinate through observation. 
Robots are points, autonomous, anonymous, silent (except via lights), have unlimited visibility (if not obstructed), and run the same algorithm. 
They are initially placed arbitrarily in the Euclidean plane.
A robot is denoted by $r_i$ or simply $r$ when clear; we use $r$ to refer to both the robot and its position. 
Each robot has its own local coordinate system with the origin at its current position. 
No global coordinate system is shared, so robots are disoriented and may perceive the same point differently.
We assume that robots are \emph{opaque}: two robots see each other if no third robot lies on the line segment between them. 
This prevents any robot from knowing $N$ or the number of distinct positions. Up to $f < N$ robots (with $f$ unknown) may experience \emph{mobility faults} at any time during execution, after which they become permanently immobile but retain fully functional lights.

\noindent {\bf Configuration and View of the Robots.}
At any time $t \geq 0$, the \emph{configuration} of the robots means the set 
$\{ (r_i^t, r_i^t.color) | 1\leq i\leq N \}$ of tuples, where each tuple consists of a robot's position $r_i^t$ and the color of its light $r_i^t.color$.
We denote the (global) convex hull of all robots at any time $t$ by $\mathcal{CH}$.
For a robot $r$, we define the \emph{local configuration} as the configuration of all visible robots to $r$ at $t$, and we denote the convex hull of these visible robots by $\mathcal{CH}_r$.
When the context is clear, we omit the time superscript $t$.
A $2$D point occupied by two or more robots is termed a \emph{multiplicity point}.
When a robot $r$ observes such a point (including its own position), it can identify all the distinct colors but not the exact number of robots present at that location. 
If two or more robots are collocated, they act as a single entity: they are activated simultaneously, share the same view, and compute the same destination. This is not a limiting assumption and we prove in Theorem \ref{thm:single-entity} that without this, the problem becomes unsolvable.
If some of them experience a mobility fault, the remaining non-faulty robots proceed toward the computed destination.
This assumption is not merely a convenience; it is necessary for the problem, as we argue in Section~\ref{sec:discussion}.

\noindent {\bf Activation Cycles.}
After activation, a robot operates in \emph{Look-Compute-Move} (LCM) cycles, as follows. \emph{Look:} It captures a snapshot of the positions (w.r.t. its local coordinate system) and colors of all visible robots, including itself. 
\emph{Compute:} It then runs the algorithm using the captured data to determine the destination point and possibly a new color. 
\emph{Move:} It updates the color and moves straight towards the destination or stays put. 
It loses all memory except the current color after each LCM cycle.

\noindent \textbf{Activation Scheduler.} Robots operate under an adversarial asynchronous ({\async}) scheduler, where activations occur at arbitrary times with no global clock or synchrony. 
Each LCM phase may be delayed unpredictably, and robots can remain inactive for arbitrarily long, though fairness ensures that each robot is activated infinitely often. Time is measured in \emph{epochs}, the minimum interval in which all robots are activated and complete at least one LCM cycle. Formally, if $t_i$ is the time a robot starts its LCM cycle and $t_j$ the time by which every robot completes at least one LCM cycle, then $t_j - t_i$ defines an epoch.

\section{Impossibilities and 3-\texorpdfstring{$\mathcal{LUMI}$}{LUMI} Algorithm in \texorpdfstring{$(2,1)$}{(2,1)}-Mobility System under \texorpdfstring{$\mathcal{ASYNC}$}{ASYNC}}
\label{sec:impossibility}

In this section, we present two impossibilities: one proving the necessity of the assumption that the collocated robots are treated as a single body and another proving that 2-$\mathcal{LUMI}$ is insufficient for indulgent gathering in a $(2,1)$-mobility system under {\async}. We then present an optimal 3-$\mathcal{LUMI}$ algorithm for an {\async} $(2,1)$-mobility system.

Our model treats collocated non-faulty robots as a single entity
(Section~\ref{sec:model}): on activation they take one snapshot and compute one
destination. We now show that this cannot be dropped. If collocated robots were
instead allowed to keep separate local coordinate systems and to act on their
own, then indulgent gathering would no longer be solvable, for any number of
colors and even under the fully synchronous scheduler.

For two points $A$ and $X$, write $\sigma_A(X)=2A-X$ for the point reflection of X through $A$, that is, the rotation by $\pi$ about $A$.

\begin{restatable}{theo}{singlebody}\label{thm:single-entity}
If collocated robots are not treated as a single entity: each has its own local coordinate system and runs its own LCM cycle, then, in the $\psi$-$\mathcal{LUMI}$ model, no deterministic algorithm solves indulgent gathering even under {\fsync}.
\end{restatable}

\begin{proof}
Fix a deterministic algorithm $\mathcal{A}$ and a number of colors $\psi$. Take
two distinct points $B$ and $C$, let $A$ be their midpoint, and let
$\sigma=\sigma_A$ be the point reflection through $A$; then $\sigma(A)=A$,
$\sigma(B)=C$ and $\sigma(C)=B$. Put two non-faulty robots $r_1,r_2$ at $A$ and
one faulty robot at each of $B$ and $C$, and give all four robots the same
initial color. This is a valid initial configuration of a $(4,2)$-mobility
system. A faulty robot never moves, so $B$ and $C$ remain the positions of the
faulty robots for the whole execution; hence, by Definition~\ref{def:problem},
the instance is solved only if $r_1$ and $r_2$ end up together at $B$ or together
at $C$.

\noindent\textit{The adversary.}
The adversary fixes the local coordinate systems, the activation schedule, and
the faults. It gives $r_1$ an arbitrary coordinate system $\varphi$ and gives
$r_2$ the system $\varphi\circ\sigma$; it gives the faulty robot at $B$ an
arbitrary system $\chi$ and the one at $C$ the system $\chi\circ\sigma$. As
$\sigma$ is a rotation, $\varphi\circ\sigma$ is a legal coordinate system, so
this assignment is permitted. In every round the adversary activates all four
robots together, which is a fair, fully synchronous execution.

\noindent\textit{The configuration stays symmetric.}
Pair the robots as $r_1\leftrightarrow r_2$ and as (faulty at $B$)
$\leftrightarrow$ (faulty at $C$). Call a configuration \emph{symmetric} if
paired robots have equal colors and their positions are images of each other
under $\sigma$; equivalently, $\sigma$ maps the configuration onto itself with
the two pairs exchanged. We show by induction that every configuration in the
execution is symmetric.

The initial configuration is symmetric: $r_1,r_2$ are both at $A=\sigma(A)$, the
faulty robots are at $B$ and $C=\sigma(B)$, and all colors are equal.

Assume a configuration $\Gamma$ is symmetric, so $\sigma(\Gamma)=\Gamma$. In the
round, $r_1$ and $r_2$ look at the same instant. The snapshot of $r_1$ is
$\varphi(\Gamma)$, and the snapshot of $r_2$ is
$(\varphi\circ\sigma)(\Gamma)=\varphi(\sigma(\Gamma))=\varphi(\Gamma)$. So the
two robots record the same snapshot, and by symmetry they also carry the same
color. Running the same algorithm with no memory beyond the color, they compute
the same new color and the same destination $d$ in local coordinates. In global
coordinates the destination of $r_1$ is $\varphi^{-1}(d)$ and that of $r_2$ is
$(\varphi\circ\sigma)^{-1}(d)=\sigma\big(\varphi^{-1}(d)\big)$; thus the two
robots move to points that are images under $\sigma$, and they take the same new
color. The two faulty robots read identical snapshots in the same way and adopt
one common color; being immobile, they stay at $B$ and $C$, which are already
images under $\sigma$. Opacity changes nothing here: $\sigma$ is an isometry, so
it preserves collinearity, and a robot hides a point from $r_1$ exactly when its
pair hides the reflected point from $r_2$, so the two snapshots stay identical.
Hence the configuration after the round is again symmetric, and the induction is
complete.

\noindent\textit{Conclusion.}
At every time, $r_1$ is at some point $P$ and $r_2$ is at $\sigma(P)$. These two
points coincide only when $P=\sigma(P)$, that is, only at the unique fixed point
$A$ of the reflection. So $r_1$ and $r_2$ are together only when both are at $A$,
and $A\neq B$ and $A\neq C$. They are therefore never together at $B$ or at $C$.
Their colors are equal at every step, so the lights offer no way to tell the two
robots apart, whatever the value of $\psi$. The fair {\fsync} execution above
thus stops $\mathcal{A}$ from gathering the non-faulty robots at the position of
a faulty robot. Since every {\fsync} execution is also a legal {\ssync} and
{\async} execution, the same holds under every scheduler.

The one obstruction is the freedom of the two collocated robots to hold mirrored
coordinate systems. Under the single-entity assumption they share one coordinate
system and one snapshot, so they order the two faulty positions in the same way
and move together to the same one of $B$ and $C$; the symmetry breaks in a single
round and gathering follows. This is why the assumption is necessary.
\end{proof}

\begin{restatable}{theo}{twolumiimpossibility}\label{thm:2lumiimpossibility}
In the $2$-$\mathcal{LUMI}$ model, no deterministic algorithm solves gathering in a $(2,1)$-mobility system under {\async}.
\end{restatable}

\begin{proof}
    We first define some useful notation exclusive to this proof. All the notations used in the diagrams and in the proof are listed in Table \ref{tab:lb2-notation}.

    \begin{table*}[t]
\centering
\caption{List of notation used in the 2-$\mathcal{LUMI}$ impossibility}
\label{tab:lb2-notation}

\begin{tabular}{@{}l>{\raggedright\arraybackslash}p{13.2cm}@{}}
\toprule
\textbf{Notation} & \textbf{Meaning}\\
\midrule
\multicolumn{2}{@{}l}{\textit{Move action (first component of an action)}}\\
$St$        & \emph{stay}: the robot does not move\\
$M$         & \emph{meet}: the robot moves onto the other robot's current position\\
$Sh$        & \emph{shift}: the robot moves to some other point, i.e.\ neither robot's current position\\
$Sh=m$      & a shift onto the midpoint $m$\\
$Sh\neq m$  & a shift onto a point other than the midpoint $m$\\
$m$         & midpoint of the segment joining the current positions of the two robots\\
\midrule
\multicolumn{2}{@{}l}{\textit{Color action (second component of an action)}}\\
$K$ & \emph{keep}: the color is left unchanged\\
$F$ & \emph{flip}: the color is switched to the other color\\
\midrule
\multicolumn{2}{@{}l}{\textit{Configuration classes}}\\
$P_{ww}$ & the two robots are at distinct points and both are \texttt{WHITE}\\
$P_{bb}$ & the two robots are at distinct points and both are \texttt{BLACK}\\
$P_{bw}$ & the two robots are at distinct points and have different colors\\
\midrule
\multicolumn{2}{@{}l}{\textit{Action notation} (e.g.)}\\
$(St,K)$ & the common action taken by both robots: move action $St$ and color action $K$; written as a single pair because, when the two robots share a color, they compute the same action\\
$r_1(St,K)\;r_2(M,K)$ & robot $r_1$ takes action $(St,K)$ and robot $r_2$ takes action $(M,K)$; used in $P_{bw}$, where the two robots may act differently\\
$A\in\{K,F\}$ & $A$ stands for a common color action that is either $K$ or $F$\\
``$/$'' & reads ``or'': e.g.\ $St/Sh$ means $St$ or $Sh$, $K/F$ means $K$ or $F$, and $Sh/M$ means $Sh$ or $M$\\
\midrule
\multicolumn{2}{@{}l}{\textit{Shapes in the diagram}}\\
double oval      & a configuration class ($P_{ww}$, $P_{bb}$ or $P_{bw}$)\\
oval             & a configuration or an action pair reached during the analysis\\
diamond          & a condition checked by the adversary (a branch)\\
rectangular box  & a scheduling action of the adversary, i.e.\ which robot(s) it activates\\
\bottomrule
\end{tabular}
\end{table*}

\paragraph{\bf Notations.}
The two colors are \texttt{BLACK} ($B$) and \texttt{WHITE} ($W$). The two robots are $r_1$ and $r_2$, placed at distinct points $p_1$ and $p_2$, and $m$ denotes the midpoint of the segment $\overline{p_1p_2}$.

We describe what a robot does in one cycle, called its \emph{action}, by an ordered pair $(\mathit{move}, \mathit{color})$: the first component $\mathit{move}$ is its \emph{move action} and the second $\mathit{color}$ is its \emph{color action}.

The \emph{move action} (the first component) records how the robot moves and takes one of three values:
\begin{itemize}
    \item \emph{stay}, written $St$: the robot does not move;
    \item \emph{meet}, written $M$: the robot moves onto the current position of the other robot;
    \item \emph{shift}, written $Sh$: the robot moves to some other point, that is, neither robot's current position.
\end{itemize}
When the landing point of a shift matters, we refine it by its position relative to the midpoint: $Sh=m$ is a shift onto $m$, and $Sh \ne m$ is a shift onto a point other than $m$.

The \emph{color action} (the second component) records how the robot sets its light and takes one of two values:
\begin{itemize}
    \item \emph{keep}, written $K$: the color is left unchanged;
    \item \emph{flip}, written $F$: the color is switched to the other color (the only change possible with two colors).
\end{itemize}

These three move actions and two color actions are exhaustive: on a given snapshot, every rule of any deterministic algorithm prescribes exactly one such pair. For instance, $(St, F)$ is the rule ``do not move, but switch color,'' and $(M, K)$ is ``move onto the other robot, keeping the current color.''

We write $\sim$ in place of a component to mean that the statement holds whatever value that component takes. For example, $(M, \sim)$ stands for both $(M, K)$ and $(M, F)$: the robot moves onto the other one and its color update does not matter.

A \emph{snapshot} gives a robot the positions and colors of both robots, expressed in its own local coordinate system.

\paragraph{\bf The adversary.}
The adversary controls the robots' local coordinate systems, the timing of their activations, and the occurrence of the mobility fault. Two consequences of these powers are used repeatedly. First, whenever the two robots have the same color, the adversary fixes their coordinate systems as mirror images across the perpendicular bisector of the segment $\overline{p_1 p_2}$, which passes through the midpoint $m$. The two robots then take identical snapshots, run the same algorithm $\mathcal{A}$ with no memory beyond their color, and compute the same action; a common move therefore sends them to two points symmetric about $m$, and these coincide only when the common target is $m$ itself. Second, since activations may overlap under \async, the adversary may let one robot take its snapshot while the other is still executing its move, so that the observer computes on the mover's intermediate, out-of-date position; in particular, the move by which $r_1$ has just entered the current configuration may be left unfinished while $r_2$ acts.

We track a configuration by its \emph{class}:
\begin{itemize}
    \item $P_{ww}$: the two robots are at distinct points and both have color \texttt{WHITE};
    \item $P_{bb}$: the two robots are at distinct points and both have color \texttt{BLACK};
    \item $P_{bw}$: the two robots are at distinct points and have different colors.
\end{itemize}
Whenever the two robots share a color, that is, in $P_{ww}$ or $P_{bb}$, they take identical snapshots and hence compute the same action; we write $P_{ww}.\textsf{action}$ and $P_{bb}.\textsf{action}$ for this common action. In $P_{bw}$ the two robots may receive different snapshots and act differently, so there we speak of each robot's action separately.

We now prove the impossibility by splitting into the following cases according to the robots' action in $P_{ww}$ or $P_{bb}$.
\begin{itemize}
    \item \textbf{Case 1}: $P_{ww}.\textsf{action}, P_{bb}.\textsf{action} \in \{ (\mathit{move}, \mathit{color}): \mathit{move} \in \{St, Sh{=}m, Sh{\ne}m, M\},~ \mathit{color} \in \{K, F\}\} \setminus \{(Sh{=}m, F)\}$.

    \item \textbf{Case 2}: $P_{ww}.\textsf{action} = (Sh{=}m, F)$ and 
    
    $P_{bb}.\textsf{action} \in \{ (\mathit{move}, \mathit{color}): \mathit{move} \in \{Sh{=}m, Sh{\ne}m, M\},~ \mathit{color} \in \{K, F\}\}$.

    \item \textbf{Case 3}: $P_{ww}.\textsf{action} = (Sh{=}m, F)$ and $P_{bb}.\textsf{action} = (St, K)$.

    \item \textbf{Case 4}: $P_{ww}.\textsf{action} = (Sh{=}m, F)$ and $P_{bb}.\textsf{action} = (St, F)$.

    \item \textbf{Case 5}: $P_{ww}.\textsf{action} \in \{ (\mathit{move}, \mathit{color}): \mathit{move} \in \{St, Sh{=}m,$ $Sh{\ne}m, M\},~ \mathit{color} \in \{K, F\}\}$ and $P_{bb}.\textsf{action} = (Sh{=}m, F)$.
\end{itemize}

These five cases are exhaustive for the robots' action when they have the same color. For each one we give adversarial activations under which the robots never gather.

\begin{figure}[t]
\centering
\resizebox{\linewidth}{!}{%
\begin{tikzpicture}[x=1cm,y=1cm,
  >={Stealth[length=2.6mm,width=2.3mm]},
  state/.style={draw=cState!80!black,line width=1pt,ellipse,align=center,inner sep=2pt,
                fill=cState!22,font=\large,text=black,minimum width=4.9cm,minimum height=1.2cm},
  dstate/.style={draw=cState!80!black,line width=1pt,ellipse,double,double distance=1.5pt,align=center,
                 inner sep=1pt,fill=cState!30,font=\large\bfseries,text=black,minimum width=3.0cm,minimum height=1.75cm},
  cond/.style={draw=cCond!85!black,line width=1pt,diamond,aspect=1.5,align=center,inner sep=2pt,
               fill=cCond!28,font=\large,text=black},
  act/.style={draw=cAct!65!black,line width=1pt,rounded corners=2pt,align=center,inner sep=5pt,
              fill=cAct!26,font=\large,text=black},
  e/.style={->,line width=1pt,black!88,shorten >=1pt},
  lab/.style={font=\large,inner sep=1.5pt,fill=white},
]

\node[dstate] (WC) at (-4.0,0) {$P_{ww}$};
\node[dstate] (BC) at ( 4.0,0) {$P_{bb}$};

\node[state] (W1) at (-5.6, 3.0) {$St/Sh{\neq}m/M,\,K$};
\node[state] (B1) at ( 5.6, 3.0) {$St/Sh{\neq}m/M,\,K$};

\node[state] (W2) at (-9.0, 0.3) {$Sh{=}m,\,K$};
\node[state] (B2) at ( 9.0, 0.3) {$Sh{=}m,\,K$};

\node[state] (W3) at (-5.6,-3.0) {$St/Sh{\neq}m/M,\,F$};
\node[state] (B3) at ( 5.6,-3.0) {$St/Sh{\neq}m/M,\,F$};

\node[act] (AT) at (0,3.0) {Activate $r_1,r_2$};
\node[act] (AC) at (0,0.0) {Activate $r_1,r_2$};

\node[cond] (DW) at (-4.0,-5.0)
    {Previously\\$r_2$ activated?};

\node[cond] (DB) at ( 4.0,-5.0)
    {Previously\\$r_2$ activated?};

\node[act] (WW1) at (-6.5,-7.7) {Activate $r_1$};
\node[act] (WW2) at (-1.5,-7.7) {Activate $r_2$};

\node[act] (BB1) at ( 1.5,-7.7) {Activate $r_1$};
\node[act] (BB2) at ( 6.5,-7.7) {Activate $r_2$};

\draw[e] (WC) to[out=120,in=-35] (W1);
\draw[e] (WC) to[out=178,in=2]   (W2);
\draw[e] (WC) to[out=-120,in=35] (W3);

\draw[e] (BC) to[out=60,in=215]  (B1);
\draw[e] (BC) to[out=2,in=178]   (B2);
\draw[e] (BC) to[out=-60,in=145] (B3);

\draw[e] (W1) to[out=-5,in=172]   (AT);
\draw[e] (B3) to[out=118,in=-32] (AT);
\draw[e] (AT) to[out=212,in=58]  (WC);

\draw[e] (B1) to[out=185,in=32]  (AC);
\draw[e] (W3) to[out=25,in=212]  (AC);
\draw[e] (AC) to[out=-28,in=200] (BC);

\draw[e] (W2) to[out=-90,in=160] (DW);
\draw[e] (B2) to[out=-90,in=20]  (DB);

\draw[e] (DW) to[out=230,in=90]
    node[lab,pos=0.48]{yes} (WW1);

\draw[e] (DW) to[out=310,in=90]
    node[lab,pos=0.48]{no} (WW2);

\draw[e] (WW1) to[out=135,in=-135] (WC);
\draw[e] (WW2) to[out=45,in=-45]   (WC);

\draw[e] (DB) to[out=230,in=90]
    node[lab,pos=0.48]{yes} (BB1);

\draw[e] (DB) to[out=310,in=90]
    node[lab,pos=0.48]{no} (BB2);

\draw[e] (BB1) to[out=135,in=-45] (BC);
\draw[e] (BB2) to[out=45,in=-135] (BC);

\end{tikzpicture}%
}
\caption{Case 1 of the $2$-$\mathcal{LUMI}$ impossibility.}
\Description{case1}
\label{fig:LB-1}
\end{figure}

\paragraph{\textbf{Case 1:}}
Suppose first that $P_{ww}.\textsf{action} = P_{bb}.\textsf{action} = (Sh{=}m, K)$. Under a synchronous activation both robots would move onto $m$ and gather, so instead the adversary activates them one at a time. Then only one robot moves onto $m$, so they stay at distinct points, and since the color action is $K$ their colors are unchanged. For every other action in Case 1 the adversary activates both robots synchronously: if the move action is $St$ they stay in place; if it is $M$ the synchronous activation swaps their positions; and if it is $Sh{\ne}m$ they move to two distinct points symmetric about $m$. In each of these the robots remain at distinct points. Their color action is common, so $K$ keeps the two colors equal and $F$ flips both to the other color; either way the robots keep a shared color. Hence, for every action in Case 1, the robots stay separated and in $P_{ww}$ or $P_{bb}$, so they never gather. This is illustrated in Figure~\ref{fig:LB-1}.

\paragraph{\textbf{Case 2:}}
If $P_{ww}.\textsf{action} = (Sh{=}m, F)$, the adversary activates a single robot, namely $r_1$ if $r_2$ was the last to be activated, and $r_2$ otherwise. Because the color action is $F$, this one activation makes the two colors differ and the configuration moves to $P_{bw}$.
If $P_{bb}. \textsf{action} = (Sh{=}m, K~{or}~F$, then also the adversary activates only one robot, say $r_1$ (or $r_2$), if $r_2$ (or $r_1$) was the last robot to be activated. 
For both such cases, since one robot remains idle (or on its move phase), the robots remain separated.
Such adversarial action allow the robot reach either to $P_{bw}$ or $P_{bb}$, depending on the color action is $K$ or $F$.
If $P_{bb}. \textsf{action}$ is other than $(Sh{=}m ,K~{or}~F)$, the adversary allows both robots to activate (if one is already activated at that time, then only the current ideal robot), so that they reach to $P_{ww}$ or remains in $P_{bb}$ without gathered.
We now handle $P_{bw}$, first splitting on whether the two color actions agree and then on the move actions (also refer to Fig. \ref{fig:diag-case2}).
\begin{itemize}
    \item The two color actions agree, and the move actions are any combination other than (i)~both $Sh{=}m$, or (ii)~one $St$ and the other $M$. The adversary activates both robots synchronously. A synchronous activation gathers them only in the two excluded combinations, so in every remaining one the robots stay at distinct points. As the color actions agree, the two colors stay different and the configuration remains in $P_{bw}$.

    \item Both robots have the action $(Sh{=}m, K)$. The adversary activates them one at a time, which keeps the schedule fair. Since only one robot runs its cycle at a time, the robots stay apart, and with color action $K$ the configuration remains in $P_{bw}$.

    \item Both robots have the action $(Sh{=}m, F)$. Again the adversary activates them one at a time, so the robots stay apart; and since the color action is $F$, the configuration moves to $P_{ww}$ or $P_{bb}$.

    \item The two color actions are both $K$, but one robot's move action is $St$ and the other's is $M$. The adversary activates both robots synchronously and makes the robot with move action $M$ faulty, so its move does not complete. The color action being $K$, the configuration stays in $P_{bw}$, and because the $M$-move is left unfinished the robots do not meet.

    \item The two color actions are both $F$, but one robot's move action is $St$ and the other's is $M$. The adversary activates only the robot whose move action is $St$. That robot flips its color, so the configuration becomes $P_{ww}$ or $P_{bb}$ (depending on the color last set). Fairness is restored once the configuration reaches $P_{ww}$ or $P_{bb}$.

    \item The two color actions differ, and the move action of the robot with color action $F$ is anything other than $M$. The adversary activates only the robot having color action as $F$. It flips its color, so the configuration becomes $P_{ww}$ or $P_{bb}$; and since its move action is not $M$, the robots do not meet.

    \item One robot's action is $(St, K)$ and the other's is $(M, F)$. Let $r_1$ have $(St, K)$ and $r_2$ have $(M, F)$. The adversary activates $r_2$, lets it flip its color, and does not let it finish its move until $r_1$ is activated; this turns the configuration into $P_{ww}$ or $P_{bb}$. When $r_1$ is then activated in this same-color configuration, it takes the action of that class. Since $P_{ww}.\textsf{action} = (Sh{=}m, F)$ and $P_{bb}.\textsf{action} \in \{ (\mathit{move}, \mathit{color}): \mathit{move} \in \{Sh{=}m, Sh{\ne}m, M\},~ $ $\mathit{color} \in \{K, F\}\}$, the robot $r_1$ must leave its current position. By the time both cycles end, $r_2$ reaches the point $r_1$ had occupied, but $r_1$ has already moved away, so the robots stay apart.

    \item One robot's action is $(Sh, K)$ and the other's is $(M, F)$. The adversary simply activates both robots synchronously, which takes the configuration to $P_{ww}$ or $P_{bb}$.

    \item One robot's action is $(M, K)$ and the other's is $(M, F)$. The adversary again activates both robots synchronously, which takes the configuration to $P_{ww}$ or $P_{bb}$.
\end{itemize}

\begin{figure*}[t]
\centering
\resizebox{\textwidth}{!}{%
\begin{tikzpicture}[x=1cm,y=1cm,
  >={Stealth[length=2.8mm,width=2.4mm]},
  state/.style={draw=cState!80!black,line width=1pt,ellipse,align=center,inner sep=1pt,
                fill=cState!16,font=\normalsize,text=black,minimum width=2.2cm,minimum height=1.15cm},
  dstate/.style={state,double,double distance=1.5pt,font=\large\bfseries,fill=cState!28,minimum width=1.7cm},
  cond/.style={draw=cCond!85!black,line width=1pt,diamond,aspect=1.5,align=center,inner sep=1pt,
               fill=cCond!20,font=\normalsize,text=black},
  act/.style={draw=cAct!65!black,line width=1pt,rounded corners=2pt,align=center,inner sep=4pt,
              fill=cAct!16,font=\normalsize,text=black},
  big/.style={act,align=center,text width=3.1cm},
  e/.style={->,line width=0.9pt,black!88,shorten >=1pt},
  en/.style={line width=0.9pt,black!88},
  lab/.style={font=\normalsize,inner sep=1.5pt,fill=white},
]
\node[cond]   (N5)  at (3.0,-2.6)   {Previously\\$r_2$ actiavted?};
\node[act]    (N6)  at (6.6,-2.0)   {Activate $r_1$};
\node[act]    (N7)  at (6.6,-3.1)   {Activate $r_2$};
\node[act]    (N1)  at (9.2,-0.5)   {Activate $r_1,r_2$};
\node[state]  (N2)  at (18.2,-0.9)  {$Sh\neq m/M, F$};
\node[state]  (N3)  at (15.2,-2.4)  {$Sh=m, K$};
\node[state]  (N4)  at (20.8,-2.6)  {$Sh\neq m/M, K$};
\node[dstate] (N8)  at (9.4,-4.4)   {$P_{ww}$};
\node[dstate] (N9)  at (16.6,-4.4)  {$P_{bb}$};
\node[state]  (N10) at (6.7,-5.7)   {$Sh=m, F$};
\node[cond]   (N11) at (3.0,-7.0)   {Previously\\$r_2$ actiavted?};
\node[act]    (N12) at (6.6,-6.4)   {Activate $r_1$};
\node[act]    (N13) at (6.6,-7.5)   {Activate $r_2$};
\node[dstate] (N14) at (9.4,-7.6)   {$P_{bw}$};
\node[cond]   (N15) at (12.9,-7.8)  {Is view of $r_1$\\$=(B,W)$?};
\node[act]    (N16) at (18.2,-7.3)  {Activate $r_2$};
\node[state]  (N17) at (19.7,-9.5)  {$r_1(St/Sh/M,K)$\\$r_2(St/Sh,F)$};
\node[act]    (N18) at (0.9,-8.8)   {Activate $r_1$};
\node[state]  (N19) at (4.0,-9.9)   {$r_1(St,F)$\\$r_2(M,F)$};
\node[state]  (N20) at (7.6,-9.8)   {$r_1(M,F)$\\$r_2(St,F)$};
\node[state]  (N21) at (1.8,-11.3)  {$r_1(Sh=m,F)$\\$r_2(Sh=m,F)$};
\node[state]  (N22) at (6.0,-12.0)  {Color action of both\\the robots are identical\\say $A\in\{K,F\}$};
\node[state]  (N28) at (2.4,-13.9)  {$r_1(St/Sh,A)$\\$r_2(St/Sh\neq m,A)$};
\node[state]  (N29) at (6.5,-14.1)  {$r_1(St/Sh\neq m,A)$\\$r_2(Sh=m,A)$};
\node[state]  (N30) at (10.4,-13.3) {$r_1(Sh\neq m,A)$\\$r_2(M,A)$};
\node[state]  (N32) at (4.2,-15.8)  {$r_1(Sh\neq m,A)$\\$r_2(Sh\neq m,A)$};
\node[state]  (N33) at (11.5,-15.4) {$r_1(M,A)$\\$r_2(Sh\neq m,A)$};
\node[state]  (N35) at (8.1,-17.0)  {$r_1(Sh=m,A)$\\$r_2(M,A)$};
\node[state]  (N38) at (12.2,-17.2) {$r_1(M,A)$\\$r_2(M,A)$};
\node[state]  (N41) at (8.1,-18.6)  {$r_1(M,A)$\\$r_1(Sh=m,A)$};
\node[act]    (N37) at (2.0,-17.0)  {Activate $r_1,r_2$};
\node[act]    (N23) at (11.6,-9.7)  {$r_1=r_1, r_2=r_2$};
\node[act]    (N24) at (14.9,-9.7)  {$r_1=r_2, r_2=r_1$};
\node[state]  (N25) at (13.1,-11.9) {$r_1(St/Sh,F)$\\$r_2(St/Sh/M,K)$};
\node[state]  (N26) at (17.0,-11.7) {Color action of both\\the robots are different};
\node[state]  (N27) at (20.3,-13.0) {$r_1(Sh,K)$\\$r_2(M,F)$};
\node[state]  (N31) at (19.8,-14.6) {$r_1(M,K)$\\$r_2(M,F)$};
\node[state]  (N34) at (14.9,-15.6) {$r_1(St,K)$\\$r_2(M,F)$};
\node[state]  (N36) at (17.7,-16.3) {$r_1(M,F)$\\$r_2(Sh/M,K)$};
\node[state]  (N39) at (15.6,-17.3) {$r_1(M,F)$\\$r_2(St,K)$};
\node[act]    (N40) at (23.0,-14.6) {Activate $r_1,r_2$};
\node[act]    (N40b)at (22.4,-16.9) {Activate $r_1,r_2$};
\node[act]    (N42) at (2.0,-20.6)  {Activate $r_2$};
\node[cond]   (N43) at (4.3,-21.2)  {Previously\\$r_1$ actiavted?};
\node[act]    (N44) at (2.0,-21.8)  {Activate $r_1$};
\node[state]  (N45) at (8.4,-20.9)  {$r_1(Sh,K)$\\$r_2(Sh,K)$};
\node[big]    (N46) at (15.2,-20.9) {Activate $r_2$, let it flips its color and don't finish its move until $r_1$ is activated};
\node[big]    (N47) at (19.0,-20.9) {Activate $r_1$, let it flips its color and don't finish its move until $r_2$ is activated};
\node[state]  (N48) at (9.4,-22.6)  {$r_1(St,K)$\\$r_2(M,K)$};
\node[state]  (N49) at (11.9,-23.0) {$r_1(M,K)$\\$r_2(St,K)$};
\node[act]    (N50) at (7.1,-24.0)  {Activate $r_1,r_2$\\$r_2$ faulty};
\node[act]    (N51) at (9.6,-25.0)  {Activate $r_1,r_2$\\$r_1$ faulty};
\coordinate (JH) at (13.25,-10.7);     
\coordinate (JP) at (7.7,-8.35);       
\fill (JH) circle (2.4pt);
\fill (JP) circle (2.4pt);
\draw[e] (N5) -- node[lab,above]{yes} (N6);
\draw[e] (N5) -- node[lab,below]{no} (N7);
\draw[e] (N6) to[out=0,in=150] (N9);
\draw[e] (N7) to[out=0,in=165] (N9);
\draw[e] (N3) to[out=175,in=15] (N5);
\draw[e] (N9) to[out=120,in=-5] (N3);
\draw[e] (N9) to[out=72,in=-110] (N2);
\draw[e] (N9) to[out=45,in=200] (N4);
\draw[e] (N9) to[out=-150,in=-20] (N10);
\draw[e] (N2) -- (N1);
\draw[e] (N1) -- (N8);
\draw[e,rounded corners=4pt] (N4.east) -- (23.0,-2.6) -- (23.0,-13.9) -- (N40.north); 
\draw[e] (N8) to[out=-120,in=90] (N10);
\draw[e] (N10) to[out=-150,in=55] (N11);
\draw[e] (N11) -- node[lab,above]{yes} (N12);
\draw[e] (N11) -- node[lab,below]{no} (N13);
\draw[e] (N12) to[out=0,in=140] (N14);
\draw[e] (N13) to[out=0,in=185] (N14);
\draw[e] (N14) -- (N15);
\draw[e] (N15) to[out=-110,in=90] node[lab,left=1pt]{yes} (N23);
\draw[e] (N15) to[out=-30,in=120] node[lab,right=1pt]{no} (N24);
\draw[en] (N23) to[out=-55,in=170] (JH);
\draw[en] (N24) to[out=-125,in=10] (JH);
\draw[e]  (JH) to[out=180,in=30]  (N22);
\draw[e]  (JH) to[out=-15,in=140] (N26);
\draw[e] (N22) to[out=125,in=-60] (N19);
\draw[e] (N22) to[out=80,in=-115] (N20);
\draw[e] (N22) to[out=155,in=-25] (N21);
\draw[e] (N22) to[out=-150,in=55] (N28);
\draw[e] (N22) to[out=-120,in=75] (N32);
\draw[e] (N22) to[out=-95,in=85] (N29);
\draw[e] (N22) to[out=-35,in=155] (N30);
\draw[e] (N22) to[out=20,in=165] (N33);
\draw[e] (N22) to[out=-80,in=120] (N35);
\draw[e] (N22) to[out=-20,in=155] (N38);
\draw[e] (N22) to[out=-70,in=95] (N41);
\draw[e] (N22) to[out=-55,in=155] (N45);
\draw[e] (N22) to[out=-50,in=150] (N48);
\draw[e] (N22) to[out=-45,in=140] (N49);
\draw[e] (N28) to[out=-115,in=120] (N37);
\draw[e] (N29) to[out=-150,in=45] (N37);
\draw[e] (N30) to[out=-160,in=25] (N37);
\draw[e] (N32) to[out=-150,in=80] (N37);
\draw[e] (N33) to[out=-165,in=8] (N37);
\draw[e] (N35) to[out=-170,in=-12] (N37);
\draw[e] (N38) to[out=-170,in=4] (N37);
\draw[e] (N41) to[out=-170,in=-22] (N37);
\draw[e] (N19) to[out=160,in=65] (N18);
\draw[e] (N21) to[out=165,in=-55] (N18);
\draw[e] (N25) to[out=175,in=-6] (N18);
\draw[e,rounded corners=4pt] (N18.north) -- (0.5,-4.4) -- (N8.west);   
\draw[e] (N20) to[out=10,in=180] (N16);
\draw[e] (N17) to[out=120,in=-55] (N16);
\draw[e] (N16) to[out=150,in=12] (N9);
\draw[en,rounded corners=4pt] (N37.west) -- (-0.3,-17.0) -- (-0.3,-8.35) -- (JP);
\draw[en,rounded corners=4pt] (N42.west) -- (-0.7,-20.6) -- (-0.7,-8.35) -- (JP);
\draw[en,rounded corners=4pt] (N44.west) -- (-1.1,-21.8) -- (-1.1,-8.35) -- (JP);
\draw[en,rounded corners=4pt] (N50.west) -- (-1.5,-24.0) -- (-1.5,-8.35) -- (JP);
\draw[en,rounded corners=4pt] (N51.west) -- (-1.9,-25.0) -- (-1.9,-8.35) -- (JP);
\draw[e] (JP) -- (N14.south);
\draw[e] (N45) -- (N43);
\draw[e] (N43) -- node[lab,above]{yes} (N42);
\draw[e] (N43) -- node[lab,below]{no} (N44);
\draw[e] (N48) to[out=-120,in=55] (N50);
\draw[e] (N49) to[out=-120,in=35] (N51);
\draw[e] (N26) to[out=165,in=25] (N25);
\draw[e] (N26) to[out=80,in=-120] (N17);
\draw[e] (N26) to[out=35,in=150] (N27);
\draw[e] (N26) to[out=-25,in=130] (N31);
\draw[e] (N26) to[out=-120,in=70] (N34);
\draw[e] (N26) to[out=-80,in=110] (N36);
\draw[e] (N26) to[out=-100,in=95] (N39);
\draw[e] (N27) to[out=-55,in=80] (N40);
\draw[e] (N31) to[out=-15,in=120] (N40);
\draw[e] (N34) to[out=-75,in=130] (N46);
\draw[e] (N36) to[out=-35,in=120] (N40b);
\draw[e] (N39) to[out=-40,in=150] (N47);
\draw[e,rounded corners=5pt] (N40.east)  -- (24.6,-14.6) -- (24.6,-4.7) -- (N9.east);   
\draw[e,rounded corners=5pt] (N46.south) -- (15.2,-22.6) -- (25.4,-22.6) -- (25.4,-4.3) -- (N9.east);   
\draw[e,rounded corners=5pt] (N40b.east) -- (26.2,-16.9) -- (26.2,-1.5) -- (9.6,-1.5) -- (N8.north);  
\draw[e,rounded corners=5pt] (N47.east)  -- (27.0,-20.9) -- (27.0,-1.9) -- (9.2,-1.9) -- (N8.north); 
\end{tikzpicture}%
}
\caption{Case 2 of the $2$-$\mathcal{LUMI}$ impossibility.}
\label{fig:diag-case2}
\Description{case2 diagram}
\end{figure*}

\paragraph{\textbf{Case 3:}} When the robots are in $P_{ww}$ or $P_{bb}$ with $P_{ww}.action = (Sh{=}m, F)$ and $P_{bb}.action = (St,K)$, the adversary acts as follows: it first activates only one robot initially when they are in $P_{ww}$. Such an activation makes the two robots of different colors and the configuration becomes $P_{bw}$.
The adversary follows similar activation scheduling as Case 2, except when the action of one of the robots is $(St,K)$ and the other is $(M,F)$.
Otherwise, let us assume that $r_1.action = (St,K)$ and $r_2.action = (M,F)$. The adversary activates $r_2$ and makes it faulty, which makes the configuration $P_{bb}$.
After that the adversary activates both the robots in sync.
Since $P_{bb}.action = (St,K)$, the configuration remains in $P_{bb}$ thereafter.
This is shown in transition diagram in Fig. \ref{fig:diag-lb3}.
\begin{figure*}[t]
\centering
\newcommand{\LCMcycle}[6]{%
  \node[phase,minimum width=#4cm,anchor=west]
    (#1L) at (#2,#3) {$L$};
  \node[phase,minimum width=#5cm,anchor=west]
    (#1C) at (#1L.east) {$C$};
  \node[phase,minimum width=#6cm,anchor=west]
    (#1M) at (#1C.east) {$M$};
}
\resizebox{\textwidth}{!}{%
\begin{tikzpicture}[
    x=1cm,y=1cm,
    phase/.style={
        draw=black,
        line width=0.7pt,
        minimum height=0.68cm,
        inner sep=0pt,
        outer sep=0pt,
        align=center,
        font=\small
    }
]

\node[anchor=east,font=\small] at (-0.18,0.82) {$r_1$};

\LCMcycle{A}{0}{0.82}{1.40}{1.30}{2.80}

\LCMcycle{B}{5.68}{0.82}{0.45}{0.45}{0.42}

\LCMcycle{C}{7.18}{0.82}{0.55}{0.55}{2.80}

\LCMcycle{D}{11.26}{0.82}{0.42}{0.42}{0.42}

\LCMcycle{E}{12.70}{0.82}{0.45}{0.45}{2.60}

\node[anchor=east,font=\small] at (-0.18,-0.32) {$r_2$};

\LCMcycle{F}{3.00}{-0.32}{0.45}{0.45}{0.42}

\LCMcycle{G}{4.50}{-0.32}{0.45}{0.45}{3.05}

\LCMcycle{H}{8.63}{-0.32}{0.42}{0.42}{0.42}

\LCMcycle{I}{10.07}{-0.32}{0.45}{0.45}{2.85}

\LCMcycle{J}{13.90}{-0.32}{0.42}{0.42}{0.42}

\LCMcycle{K}{15.24}{-0.32}{0.42}{0.42}{1.55}

\end{tikzpicture}%
}
\caption{An asynchronous execution in which the adversary delays
the completion of an LCM cycle while the other robot executes
multiple cycles.}
\Description{An important figure on LCM}
\label{fig:lcm-overlap-case4}
\end{figure*}

\paragraph{\textbf{Case 4:}}
When the robots are in $P_{ww}$ or $P_{bb}$ with $P_{ww}.action = (Sh{=}m, F)$ and $P_{bb}.action = (St,F)$, the adversary acts as follows. When the robots are in $P_{ww}$ or $P_{bb}$, it activates only $r_1$ (or $r_2$) if $r_2$ (or $r_1$) is the last robot activated. Such an activation makes the two robots of different colors and the configuration becomes $P_{bw}$.
The adversary follows similar activation scheduling as Case 2, except when the action of one of the robots is $(St,K)$ and the other is $(M,F)$.
Otherwise, without loss of generality, let us assume that $r_1.view = (B, W)$, and $r_2.view = (W, B)$. Moreover, let us also assume that $r_1.action = (St,K)$ and $r_2.action = (M,F)$.
The adversary activates $r_2$, let it flip its color and do not finish its move until $r_1$ gets activated twice, as follows. After change of the color of $r_2$, the two robots have the same color, and so the configuration is either $P_{bb}$ or $P_{bw}$. As $r_2$ is already an activated robot, as set earlier by the adversary, $r_1$ again gets activated and changes its color to reach the class in $P_{bw}$. Now the view of $r_1$ becomes $(W, B)$, so by the assumption $r_1.action$ must be $(M, F)$. Therefore $r_1$ must initiate its movement before $r_2$ settles in $r_1$'s previous position. Again, the adversary makes the LCM cycle of $r_1$ (instead of $r_2$) sufficiently large that $r_2$ activates twice, and so on, as shown in Fig. \ref{fig:lcm-overlap-case4}. 
Thus the adversary can schedule the activation of the robots and manipulate their LCM cycle phases such that both of them will target to move to the other's position, and hence they never gather. This is demonstrated in Fig. \ref{fig:diag-lb-case4}.

\begin{figure}[ht]
\centering
\resizebox{0.5\linewidth}{!}{
\begin{tikzpicture}[
  >={Stealth[length=2.6mm]},
  state/.style={draw=cState!80!black,semithick,ellipse,align=center,inner sep=2pt,fill=cState!22,font=\large,text=black,minimum height=1cm},
  dstate/.style={state,double,double distance=1.2pt,font=\large\bfseries,fill=cState!30},
  cond/.style={draw=cCond!85!black,semithick,diamond,aspect=2.2,align=center,inner sep=1pt,fill=cCond!22,font=\normalsize,text=black},
  act/.style={draw=cAct!70!black,semithick,rounded corners=1pt,align=center,inner sep=3pt,fill=cAct!24,font=\normalsize,text=black},
  e/.style={->,semithick,black},
  lab/.style={font=\normalsize,inner sep=1.5pt},
]
\node[dstate] (PAW) at (0,0) {$P_{ww}$};
\node[state]  (ShmF) at (0,-1.7) {$Sh=m,F$};
\node[dstate] (PB) at (0,-3.4) {$P_{bw}$};
\node[cond]   (D2) at (0,-5.1) {Is view of $r_1$\\$=(B,W)$?};
\node[act] (Byy) at (-2.7,-6.8) {$r_1{=}r_1,\,r_2{=}r_2$};
\node[act] (Bnn) at (2.7,-6.8)  {$r_1{=}r_2,\,r_2{=}r_1$};
\node[state] (o6) at (-2.7,-8.9) {$r_1(St,K)$\\$r_2(M,F)$};
\node[state] (o7) at (2.7,-8.9)  {$r_1(M,F)$\\$r_2(St,K)$};
\node[act] (F2) at (-2.6,-10.9) {Activate $r_2$\\$r_2$ faulty};
\node[act] (F1) at (2.6,-10.9)  {Activate $r_1$\\$r_1$ faulty};
\node[dstate] (PAB) at (0,-13.0) {$P_{bb}$};
\node[state]  (StK) at (-2.9,-12.5) {$St,K$};
\node[act]    (Ar12) at (-2.9,-14.2) {Activate $r_1,r_2$};

\draw[e] (PAW) -- (ShmF);
\draw[e] (ShmF) -- (PB);
\draw[e] (PB) -- (D2);
\draw[e] (D2) to[out=-122,in=90] node[lab,pos=0.4,left=-1pt]{yes} (Byy);
\draw[e] (D2) to[out=-58,in=90]  node[lab,pos=0.4,right=-1pt]{no}  (Bnn);
\coordinate (jc) at (0,-7.9);
\draw[semithick] (Byy.south) to[out=-90,in=180] (jc);
\draw[semithick] (Bnn.south) to[out=-90,in=0]   (jc);
\draw[semithick] (-2.7,-7.9) -- (2.7,-7.9);
\draw[e] (-2.7,-7.9) -- (o6.north);
\draw[e] (2.7,-7.9)  -- (o7.north);
\draw[e] (o6) -- (F2);
\draw[e] (o7) -- (F1);
\coordinate (mF) at (0,-11.9);
\draw[semithick] (F2.south) to[out=-90,in=180] (mF);
\draw[semithick] (F1.south) to[out=-90,in=0]   (mF);
\draw[e] (mF) -- (PAB);
\draw[e] (PAB) to[out=165,in=-5] (StK);
\draw[e] (StK) -- (Ar12);
\draw[e] (Ar12) to[out=0,in=-150] (PAB);
\end{tikzpicture}%
}
\caption{Case 3 of the $2$-$\mathcal{LUMI}$ impossibility.}
\Description{A figure}
\label{fig:diag-lb3}
\Description{case3 diagram}
\end{figure}

\begin{figure}[ht]
\centering
\resizebox{0.5\linewidth}{!}{
\begin{tikzpicture}[
  >={Stealth[length=2.6mm]},
  state/.style={draw=cState!80!black,semithick,ellipse,align=center,inner sep=2pt,fill=cState!22,font=\large,text=black,minimum height=1cm},
  dstate/.style={state,double,double distance=1.2pt,font=\large\bfseries,fill=cState!30},
  cond/.style={draw=cCond!85!black,semithick,diamond,aspect=1.9,align=center,inner sep=1pt,fill=cCond!22,font=\normalsize,text=black},
  act/.style={draw=cAct!70!black,semithick,rounded corners=1pt,align=center,inner sep=3pt,fill=cAct!24,font=\normalsize,text=black},
  bigact/.style={act,font=\small,text width=2.9cm},
  e/.style={->,semithick,black},
  lab/.style={font=\normalsize,inner sep=1.5pt},
]
\node[dstate] (PAB) at (-3.4,0) {$P_{bb}$};
\node[dstate] (PAW) at (3.4,0)  {$P_{ww}$};
\node[state]  (StF) at (-3.4,-1.8) {$St,F$};
\node[state]  (ShmF) at (3.4,-1.8) {$Sh=m,F$};
\node[cond] (D1) at (0,-3.6) {Previously\\$r_2$ activated?};
\node[act] (Act1) at (-2.0,-5.2) {Activate $r_1$};
\node[act] (Act2) at (2.0,-5.2)  {Activate $r_2$};
\node[dstate] (PB) at (0,-6.8) {$P_{bw}$};
\node[cond] (D2) at (0,-8.7) {Is view of $r_1$\\$=(B,W)$?};
\node[act] (Byy) at (-2.7,-10.5) {$r_1{=}r_1,\,r_2{=}r_2$};
\node[act] (Bnn) at (2.7,-10.5)  {$r_1{=}r_2,\,r_2{=}r_1$};
\node[state] (o6) at (-2.7,-12.8) {$r_1(St,K)$\\$r_2(M,F)$};
\node[state] (o7) at (2.7,-12.8)  {$r_1(M,F)$\\$r_2(St,K)$};
\node[bigact] (B14) at (-3.0,-15.3) {Activate $r_2$, let it flip its color and don't finish its move until $r_1$ gets activated twice, until the second move of $r_2$ starts};
\node[bigact] (B15) at (3.0,-15.3)  {Activate $r_1$, let it flip its color and don't finish its move until $r_2$ is activated};

\draw[e] (PAB) -- (StF);
\draw[e] (PAW) -- (ShmF);
\draw[e] (StF)  to[out=-90,in=150] (D1);
\draw[e] (ShmF) to[out=-90,in=30]  (D1);
\draw[e] (D1) to[out=-122,in=90] node[lab,pos=0.4,left=-1pt]{yes} (Act1);
\draw[e] (D1) to[out=-58,in=90]  node[lab,pos=0.4,right=-1pt]{no}  (Act2);
\draw[e] (Act1) to[out=-80,in=160] (PB);
\draw[e] (Act2) to[out=-80,in=20]  (PB);
\draw[e] (PB) -- (D2);
\draw[e] (D2) to[out=-122,in=90] node[lab,pos=0.4,left=-1pt]{yes} (Byy);
\draw[e] (D2) to[out=-58,in=90]  node[lab,pos=0.4,right=-1pt]{no}  (Bnn);
\coordinate (jc) at (0,-11.6);
\draw[semithick] (Byy.south) to[out=-90,in=180] (jc);
\draw[semithick] (Bnn.south) to[out=-90,in=0]   (jc);
\draw[semithick] (-2.7,-11.6) -- (2.7,-11.6);
\draw[e] (-2.7,-11.6) -- (o6.north);
\draw[e] (2.7,-11.6)  -- (o7.north);
\draw[e] (o6) -- (B14);
\draw[e] (o7) -- (B15);
\draw[e] (B14.west) .. controls (-5.4,-13.0) and (-5.4,-2.5) .. (PAB.west);
\draw[e] (B15.east) .. controls (5.4,-13.0) and (5.4,-2.5) .. (PAW.east);
\end{tikzpicture}%
}
\caption{Case 4 of the $2$-$\mathcal{LUMI}$ impossibility.}
\label{fig:diag-lb-case4}
\Description{case4 diagram}
\end{figure}

\paragraph{\textbf{Case 5:}} We do not need to analyze this case as a standalone case. The adversarial action for this case is similar to Case 2, 3 and 4, but the adversarial action corresponding to $P_{ww}$ and $P_{bb}$ swapped. 

As a conclusion, the configuration always remains in one of the three classes $P_{ww}$, $P_{bb}$ and $P_{bw}$. Hence, gathering is never achieved in this setting. 
The transition
diagrams in Figures~\ref{fig:LB-1}, \ref{fig:diag-case2}, \ref{fig:diag-lb3}, and \ref{fig:diag-lb-case4}
trace, case by case, how the adversary drives the configuration among these three classes. Read together, they describe a single infinite execution in which the two robots are never collocated. 
It remains to argue that this execution is \emph{fair}, that is, that each robot is activated infinitely often.
We prove this by following two claims.

\smallskip
\textbf{Claim:} \emph{The schedule is fair if the configuration reaches or is in $P_{ww}$ or $P_{bb}$ infinitely often.} 

\begin{proof}
    At each step in which the configuration is in a same-color class (i.e., $P_{ww}$ or $P_{bb}$), the adversary either activates both robots at once (Case~1, when the common action is not $(Sh{=}m,F)$), or activates a single robot (one at a time in the $(Sh{=}m,K)$
sub-case of Case~1, or a single robot to break the symmetry in Cases~2--5). A
simultaneous activation activates both robots. For single activations, the
alternation rule makes consecutive ones use $r_1$ and $r_2$ in turn. By hypothesis
there are infinitely many same-color steps, hence infinitely many such activations.
If infinitely many are simultaneous, both robots are activated infinitely often.
Otherwise all but finitely many are single activations, and by the alternation rule
each of $r_1$ and $r_2$ is used in infinitely many of them. In either case both
robots are activated infinitely often, so the schedule is fair.
\end{proof}

\textbf{Claim:} \emph{The scheduling is fair if the configuration reaches $P_{bw}$, and remains in it thereafter.}

\begin{proof}
After that round every step uses a $P_{bw}$ sub-case that keeps the configuration
in $P_{bw}$, and each such sub-case uses one of three activation patterns:
(i)~a simultaneous activation of both robots (the sub-cases in which the color
actions agree and the moves avoid the two gathering combinations, and the sub-case
with equal color action $K$ and moves $St$/$M$, where one robot is additionally made
faulty); (ii)~a one-at-a-time activation alternating between the robots (the
$(Sh{=}m,K)$ sub-case); or (iii)~the overlapping schedule of Case~4, in which the
adversary repeatedly lengthens one robot's LCM cycle so that the other is activated
in the meantime, the two robots exchanging this role from one iteration to the next.
In pattern~(i) both robots are activated in every round (a robot hit by the mobility
fault is still activated; only its move fails to complete); in pattern~(ii) they are
activated alternately; in pattern~(iii) each iteration activates both robots and the
construction repeats without end. In all three patterns each of $r_1$ and $r_2$ is
activated infinitely often, so the schedule is fair.
\end{proof}

All the above argument prove Theorem \ref{thm:2lumiimpossibility}.
\end{proof}

\subsection{3-$\mathbf{\mathcal{LUMI}}$ Optimal Algorithm} 
We have an {\async} color-optimal algorithm for a $(2,1)$-mobility system.
Initially, both $r_1$, $r_2$ are colored \texttt{OFF}, and $\overline{r_1r_2}$ is the line segment joining them. We describe the algorithm from $r_1$'s perspective.
\begin{itemize}[leftmargin=*]
    \item $\mathbf{r_1.color =}$ \textbf{\texttt{OFF}}: If $r_1$ finds $r_2.color =$ \texttt{OFF}, it sets $r_1.color=$ \texttt{MOVE} and moves to the midpoint of $\overline{r_1r_2}$.
If $r_2.color =$ \texttt{MOVE}, it waits, as $r_2$ might be in its \texttt{Move} phase due to {\async}.
Otherwise, if $r_2.color =$ \texttt{END} and they are at different locations, it sets $r_1.color =$ \texttt{MOVE} and moves to the position of $r_2$.

\item $\mathbf{r_1.color =}$ \textbf{\texttt{MOVE}}: It changes its color to \texttt{END} without moving, so a robot is colored \texttt{END} only after it has reached its destination.

\item $\mathbf{r_1.color =}$ \textbf{\texttt{END}}: If $r_2.color = $ \texttt{OFF}$/$\texttt{MOVE}, $r_1$ waits, allowing $r_2$ to complete its movement and change to \texttt{END}.
    If $r_1$ finds itself collocated with \texttt{END}-colored $r_2$, it terminates.
    Otherwise, if $r_2.color =$ \texttt{END} and they are at different locations, $r_1$ moves to the position of $r_2$ with the current color. 
\end{itemize}

\begin{restatable}{theo}{theoremtwogather}
\label{thm:gathering2robots}
The above deterministic algorithm solves gathering in a $(2,1)$-mobility system under {\async} in the 3-$\mathcal{LUMI}$ model.
\end{restatable}

\begin{proof}
Let $p_1$ and $p_2$ be the initial distinct positions of the two robots $r_1$ and $r_2$. We denote the midpoint of $\overline{p_1 p_2}$ by $m$. For the convenience of the proof, we call an algorithmic step $\texttt{color}_1/\texttt{color}_2$-rule to mean that one of the robots has $\texttt{color}_1$ and the other has $\texttt{color}_2$. 

We first prove that, in any \textit{fault-free execution}, every destination computed by either robot is $m$.
Note that there is no target computation when the algorithm applies \texttt{MOVE}/\texttt{color}$_2$-rule.
If a robot applies the \texttt{OFF}/\texttt{OFF}-rule, both of them are at their initial (distinct) positions, so the computed midpoint is $m$.
If it applies the \texttt{OFF}/\texttt{END}-rule, the observed \texttt{END}-colored robot is stationary and sitting at $m$, as the only way to have the color \texttt{END}-colored is to have completed a move to $m$. Therefore, the target of the \texttt{OFF}-colored robot in this case is also  $m$.
Moreover, since every \texttt{END}-colored robot is at $m$, \texttt{END}/\texttt{END}-rule never applies to compute any further target in a fault-free execution.
Thus, all the computed destinations are $m$.

In a fault-free execution, each robot only moves forward through $\texttt{OFF} \rightarrow \texttt{MOVE} \rightarrow \texttt{END}$, and any wait ends within one epoch (a \texttt{MOVE} colored robot turns \texttt{END} the next time it gets activated). 
The first robot to see a configuration with both of them being \texttt{OFF} moves to $m$ and turns \texttt{END} there. 
The other then sees \texttt{END}, moves to $m$ and turns \texttt{END} too.
Now both are \texttt{END} at $m$, so they are gathered and terminate. This takes $O(1)$ epochs.

When there is a faulty robot in an execution, it can stop at some point, say $q$.
However, its light still works, so it changes its color to \texttt{END} at $q$.
The other robot also turns \texttt{END} following the same steps. If the two robots are collocated, they terminate. 
If not, there are two \texttt{END} robots at different locations.
It means that one of them must be faulty, and a stationary robot. 
The non-faulty robot then applies the \texttt{END}/\texttt{END}-rule and moves to $q$. Since the robot at $q$ cannot move anymore,
they both terminate together at $q$ -- the indulgent gathering.   
\end{proof}

\section{Algorithm for \texorpdfstring{$(N,f)$}{(N,f)}-Mobility System}
\label{sec:Nalgorithm}
We first  present the 7-$\mathcal{LUMI}$ algorithm that solves the indulgent gathering problem in an $(N, f)$-mobility system under {\async} and it takes $O(N)$ epochs in total.

\subsection{Notations and Preliminaries} 
We introduce some terminology first. 

\smallskip
\noindent $\blacktriangleright$~$\mathcal{V}_r$ is the set of all visible robots to $r$.

\smallskip
    \noindent $\blacktriangleright$ \textbf{Convex Layers:} The boundary of the convex hull $\mathcal{CH}$ of all the robots, denoted by $\mathcal{CH}^0$, is the \emph{outer-most layer} of the robots. 
    For $j \geq 1$, the $j$-th layer, $\mathcal{CH}^j$, is defined as the boundary of the convex hull obtained after removing all robots on $\bigcup_{k=0}^{j-1} \mathcal{CH}^k$.
    $\mathcal{CH}^{l-1}$ denotes the \emph{inner-most layer} such that $\mathcal{CH}^{l}$ does not exist (see Fig. \ref{fig:layers}).
    \begin{figure}[h]
    \centering
    \begin{minipage}[b]{0.33\linewidth}
    \centering
        \includegraphics[width=0.86\linewidth]{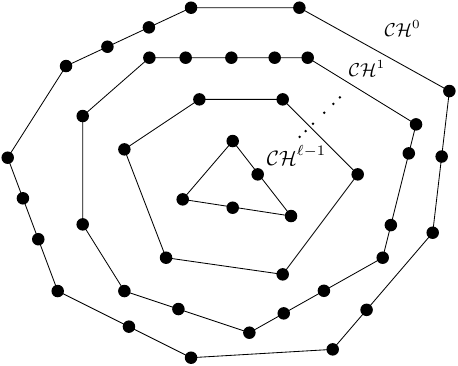}
        \caption{The disjoint global convex layers}
        \label{fig:layers}
    \end{minipage}\hfill
    \begin{minipage}[b]{0.62\linewidth}
    \centering
      \includegraphics[width=\linewidth]{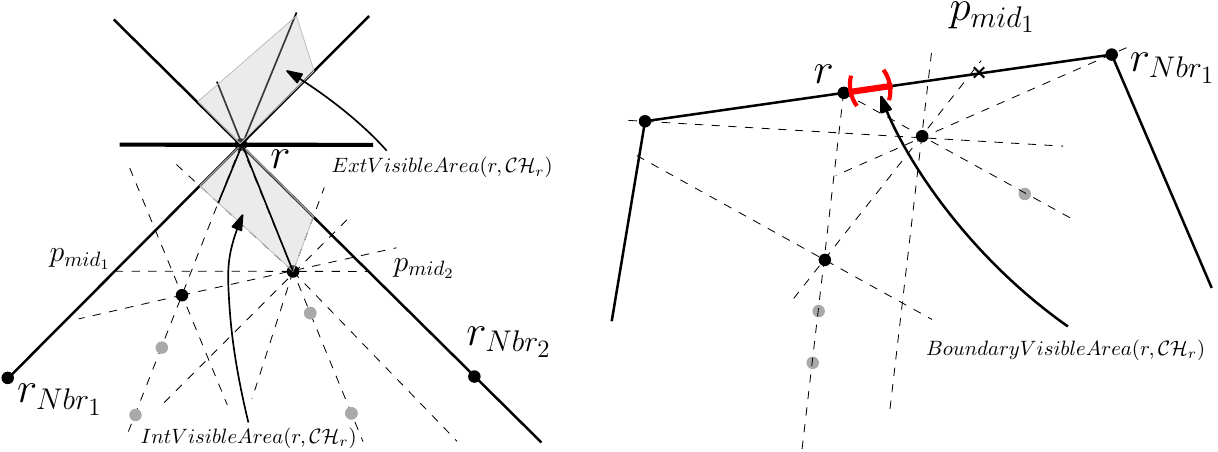}  
    \caption{Interior, Exterior, and Boundary visible Area of a robot $r$ w.r.t. the convex hull $\mathcal{CH}_r$}
    \label{fig:visiblearea}
    \end{minipage}
    \Description{A technical figure}
\end{figure}
    Although robots have no knowledge of the global layers, each robot $r$ can compute its local layers using the set of visible robots $\mathcal{V}_r$.
    We denote the $j$-th locally computed layer of $r$ by $\mathcal{CH}_r^j$.

\smallskip
    \noindent $\blacktriangleright$~$l$ denotes the number of global layers in the initial configuration. 

\smallskip
    \noindent $\blacktriangleright$~For two points $a$ and $b$, the line segment joining them is denoted by $\overline{ab}$, and the line passing through them is denoted by $\overleftrightarrow{ab}$.

\smallskip
    \noindent $\blacktriangleright$~$d(a, b)$ is the usual Euclidean distance between points $a$ and $b$.

\smallskip
    \noindent $\blacktriangleright$~For a line $L$, $\mathcal{H}_{L}^1$ and $\mathcal{H}_{L}^2$ (resp. $\overline{\mathcal{H}}_{L}^1$ and $\overline{\mathcal{H}}_{L}^2$) represent the two open (resp. closed) half-planes delimited by $L$.

\smallskip
    \noindent $\blacktriangleright$~On $\mathcal{CH}_r$, two robots $r$ and $r'$ are \emph{neighbours} if the line segment $\overline{rr'}$ lies on the boundary of $\mathcal{CH}_r$ with no other robot on it.

\smallskip
    \noindent $\blacktriangleright$ \textbf{Corner, Boundary \& Interior Robot:} A robot $r$ is called a \textbf{corner robot} if there exist two distinct lines, $L^1_r$ and $L^2_r$, such that $r$ lies at their point of intersection, and all other robots visible to $r$ are located in the region $\overline{\mathcal{H}}_{L^1_r}^i \cap \overline{\mathcal{H}}_{L^2_r}^j$ for some $i, j \in \{1, 2\}$. 
    Otherwise, if there exists a line $L_r$ passing through $r$ such that all robots visible to $r$ lie entirely within one of the two half-planes $\overline{\mathcal{H}}_{L_r}^1$ or $\overline{\mathcal{H}}_{L_r}^2$, the robot $r$ is called a \textbf{boundary robot}.
    Else, $r$ is an \textbf{interior robot}.

We now highlight a key subroutine, the \emph{visible area}, introduced by Sharma \emph{et al.}~\cite{SharmaVTBR16} and incorporated into our algorithm. It repositions a robot so it can see all the stationary robots. In~\cite{SharmaVTBR16}, the visible area is computed only for corner robots of a convex hull. 
Here, we extend this definition to a boundary robot $r$ as well, enabling it to see all stationary robots on the half-plane $\mathcal{H}_{L_r}$ (described below). Each computation takes two inputs: the current position of $r$ and a given convex hull.

\smallskip
 \noindent $\blacktriangleright$~\textbf{Interior Visible Area} ($IntVisibleArea(r, \mathcal{CH}_r)$): The detailed computation of visible area for a corner robot $r$ is presented in \cite{SharmaVTBR16}. 
    As shown in Fig. \ref{fig:visiblearea}, this is a polygonal subregion contained within the triangle $\Delta r p_{mid_1} p_{mid_2}$, where $p_{mid_1}$ and $p_{mid_2}$ are the two midpoints of $\overline{rr_{Nbr_1}}$ and $\overline{rr_{Nbr_2}}$ resp. with $r_{Nbr_1}$ and $r_{Nbr_2}$ being the two neighbours of $r$ on the convex hull $\mathcal{CH}_r$. An important property is that $r$ remains a corner robot even after moving within it, and becomes visible to any stationary robot.

\smallskip
\noindent $\blacktriangleright$~\textbf{Boundary Visible Area} ($BoundaryVisibleArea(r, \mathcal{CH}_r)$): This type of visible area is computed when $r$ is a boundary robot on $\mathcal{CH}_r$. 
    Let $p_{mid_1}$ and $p_{mid_2}$ be the midpoints of the line segments $\overline{rr_{Nbr_1}}$ and $\overline{rr_{Nbr_2}}$. Let $r$ choose the midpoint $p_{mid_1}$ arbitrarily, as shown in Fig. \ref{fig:visiblearea}. 
    It computes a target point $t_r$ on the line segment $\overline{rp_{mid_1}}$ such that after moving to $t_r$, the robot $r$ can see all the stationary robots on the open half plane $\mathcal{H}_{L_r}$ that does not contain any robot of $\mathcal{CH}_r$, where $L_r = \overleftrightarrow{rr_{Nbr_1}}$. 
    To compute such a point $t_r$, $r$ considers the following sets of lines $S^1_r, S^2_r$ 
    and $S^3_r$ defined as $S^1_r = \{\overleftrightarrow{rr'} | r' \in \mathcal{V}_r \}$, $S^2_r = \{\overleftrightarrow{r'r''} | r', r'' \in \mathcal{V}_r \}$  and $S^3_r =  \bigcup_{r' \in \mathcal{V}_r}\{ L ~|~  L \text{~passes through $r'$ and parallel to }L', \forall L' \in S_r^1 \cup S_r^2  \}$.
    Finally, a point $t_r$ on the line segment $\overline{rp_{mid_1}}$ belongs to the boundary visible area of $r$ if $d(r, t_r) = \frac{1}{2} \min \{ d(r,L) | ~ L \in S_r^2 \cup S_r^3\}$. 
    When $\mathcal{CH}_r$ is a point, $r$ has no neighbours. In that case, it skips the computation of the half plane and directly computes the point $t_r$ satisfying the above condition. 

\smallskip
\noindent $\blacktriangleright$~\textbf{Exterior Visible Area} ($ExtVisibleArea(r, \mathcal{CH}_r)$): It is a reflection of the interior visible area with the axis of reflection passing through the current position of $r$ and parallel to the line $\overleftrightarrow{p_{mid_1} p_{mid_2}}$ (see Fig. \ref{fig:visiblearea}).

\subsection{Description of the Algorithm in 7-\texorpdfstring{$\mathcal{LUMI}$}{LUMI}}
\label{app-sec:description_7color_algo}
Initially, all robots are colored \texttt{OFF}. We describe the algorithm from the perspective of an individual robot $r$, depending on its current color and observations. 
For clarity and ease of analysis, we divide the algorithm into a series of \emph{stages}, each with a descriptive name (e.g., \textsc{Robot-Classification}).
The algorithm is divided into two main cases based on the nature of the initial configuration: Case 1 handles nonlinear configurations, and Case 2 addresses linear configurations.

\noindent \textbf{Case 1 ($\mathcal{CH}$ is non-linear):}
In the following stage, we begin by classifying the robots into two categories based on their positions relative to their local convex hull: those lying on the perimeter and those situated in the interior.

\noindent \underline{\textbf{Stage }\textsc{Robot-Classification:}} Upon activation with color \texttt{OFF}, if $r$ finds itself as an interior robot on $\mathcal{CH}_r$, it changes its color to \texttt{INNER}, otherwise \texttt{OUTER}.
Any robot that observes another robot still colored \texttt{OFF} will wait until that robot updates its color to either \texttt{OUTER} or \texttt{INNER}.

\noindent \underline{\textbf{Stage }\textsc{Inner-To-Outer:}} In this stage, we now bring the non-faulty \texttt{INNER}-colored robots to the nearest \texttt{OUTER}-colored robots through the following sequence of steps.
   \begin{itemize}
        \item \textbf{$r.color =$ \texttt{INNER}:}
        $r$ waits if it sees any \texttt{MOVE1}, or \texttt{MOVE2}-colored robot. Otherwise, it checks eligibility by ignoring visible \texttt{OUTER} and \texttt{FAULT}-colored robots and computing the convex hull $\mathcal{CH}_r^*$ of the remaining robots. If $r$ is a corner of $\mathcal{CH}_r^*$, it is considered \emph{eligible}. An eligible $r$ computes its interior visible area $IntVisibleArea(r, \mathcal{CH}_r^*)$, changes its color to \texttt{MOVE1}, and moves to a point within this area. If non-faulty, this move enables $r$ to see the nearest \texttt{OUTER}-colored robots.

        \item \textbf{$r.color =$ \texttt{MOVE1}:} If $r$ sees any robot with color \texttt{MOVE2}, it remains idle. 
        Otherwise, it identifies its neighbors $r_{Nbr_1}$ and $r_{Nbr_2}$ on $\mathcal{CH}_r^*$, and computes the line $L_r$ passing through its position and parallel to $\overleftrightarrow{r_{Nbr_1}r_{Nbr_2}}$. 
        Then $r$ changes its color to \texttt{MOVE2} and moves to the nearest visible \texttt{OUTER}-colored robot in $H_{L_r}$, the open half plane that excludes both neighbours of $r$.
        In case of only one neighbour (i.e., $\mathcal{CH}_r^*$ is linear), $L_r$ is defined as the line perpendicular to $\overleftrightarrow{rr_{Nbr_1}}$.
        When $r$ has no neighbours, it skips computing $H_{L_r}$ and directly moves to the nearest \texttt{OUTER}-colored robot.
        Moreover, $r$ may not see any \texttt{OUTER} or \texttt{MOVE2}- colored robots. In such a case, $r$ updates its color to \texttt{FAULT}, as it may happen that the layers are concentric and, being a faulty robot, $r$ is unable to observe any robot on the outermost layer.

        \item \textbf{$r.color =$ \texttt{MOVE2}:} 
        If $r$ is collocated with \texttt{OUTER}-colored robots upon activation, it changes its color to \texttt{OUTER}; otherwise, being faulty, it switches from \texttt{MOVE2} to \texttt{FAULT}.

    \end{itemize}

    Faulty robots (colored \texttt{FAULT}) execute the stage \textsc{Confirmation-Signal-To-Outers} to change their color to \texttt{FAULT-FINISH} and signal \texttt{OUTER}-colored robots to begin the next stage. An \texttt{OUTER}-colored robot waits until all robots are either \texttt{FAULT-FINISH} or \texttt{OUTER}. 

    \begin{figure}[h]
    \centering
    \begin{minipage}[b]{0.48\linewidth}
    \centering
        \includegraphics[width=0.86\linewidth]{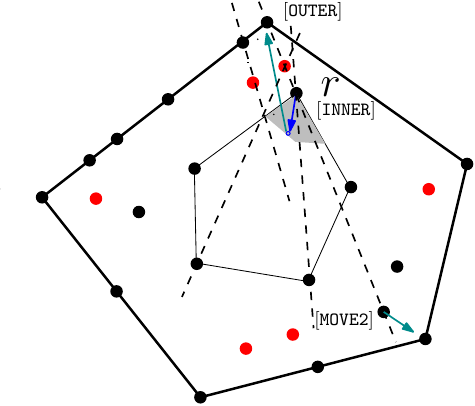}
        \caption{Movement of an eligible robot $r$ to its visible area (blue arrow) and then to \texttt{OUTER} robot (cyan arrow). Red points are the \texttt{FAULT}-colored robots}
        \label{fig:inner-to-outer}
    \end{minipage}\hfill
    \begin{minipage}[b]{0.48\linewidth}
    \centering
      \includegraphics[width=0.86\linewidth]{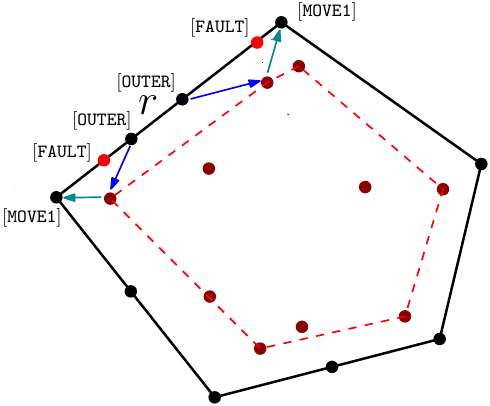}  
    \caption{Movement of boundary robot to nearest \texttt{FAULT-FINISH}-colored robots (blue arrow) and then to \texttt{MOVE1}-colored corners (cyan arrow). Brown points are \texttt{FAULT-FINISH} robots}
    \label{fig:boundary-to-corner}
    \end{minipage}
    \Description{A technical figure}
\end{figure}

    \noindent \underline{\textbf{Stage }\textsc{Confirmation-Signal-To-Outers}:}
    If $r$ (colored  \texttt{FAULT}) observes any robot with a color from the set $\{\texttt{INNER, MOVE1, MOVE2}\}$, it does nothing. 
    Otherwise, it calculates the local layers (i.e., in terms of the robots visible to it) $\mathcal{CH}_r^0, \mathcal{CH}_r^1,\cdots, \mathcal{CH}_r^{l'}$ using the visible robots in $\mathcal{V}_r$ and determines the layer $\mathcal{CH}_r^i$ to which it belongs. 
    If the interior of the convex hull formed by the robots on $\mathcal{CH}_r^i$ is either empty or contains robots only with the color \texttt{FAULT-FINISH}, $r$ changes its color to \texttt{FAULT-FINISH}. 
    The color \texttt{FAULT-FINISH} acts as a signal for the \texttt{OUTER}-colored robots to initiate their movement in the next stage.

    \noindent \underline{\textbf{Stage }\textsc{Boundary-To-Corner:}} The objective of this stage is to relocate the boundary robots colored \texttt{OUTER} to the positions of the corner robots, while keeping the positions of the corners remain unchanged. To minimise the color usage, we reuse certain colors introduced in earlier stages.
    This stage follows a strategy similar to that of the \textsc{Inner-To-Outer} stage, where we first move the boundary robots adjacent to corner robots, followed by their neighbouring boundary robots and so on. 
    Finally, the faulty boundary robots use designated color transitions to signal the corner robots to initiate the next stage.
    
    \begin{itemize}
        \item \textbf{$r.color =$ \texttt{OUTER}:} 
        If $r$ is a corner robot on $\mathcal{CH}_r$, it transitions its color to \texttt{MOVE1} to signal its adjacent boundary robots (currently colored \texttt{OUTER}) to move to its location. 
        If $r$ is a boundary robot and finds a neighbour on $\mathcal{CH}_r$ with color \texttt{MOVE1}, it switches to \texttt{MOVE2} and moves to that neighbour's position. 
        If instead it detects a \texttt{FAULT}-colored neighbor, it moves with color \texttt{OFF}—either to the nearest corner of the convex hull formed by visible \texttt{FAULT-FINISH} robots (if any exist), or otherwise to a point $t_r$ located in the interior of $\mathcal{CH}_r$.
        To compute $t_r$, the robot $r$ considers two nearest robots $r_1, r_2$, each lying on an adjacent side of the boundary segment containing $r$ in $\mathcal{CH}_r$. The point $t_r$ satisfies $\overline{r t_r} \perp \overline{r_1 r_2}$ and $d(r, t_r) = \frac{1}{2} d(r, \overleftrightarrow{r_1 r_2})$. 
        The computation of $t_r$ is constrained in such a way that another robot moving simultaneously with $r$ cannot block the view of $r$ to see the designated corner robots after the movement.

        \item \textbf{$r.color =$ \texttt{MOVE2}:} If $r$ is a corner robot  on $\mathcal{CH}_r$ and is collocated with a \texttt{MOVE1}-colored robot, it adopts \texttt{MOVE1}, indicating that $r_1$ has successfully moved to the corner. 
        If $r$ is a boundary robot whose both neighbours are colored from the set $\{\texttt{MOVE1}, \texttt{MOVE2}, \texttt{FAULT}\}$ and theres is no visible \texttt{OUTER}-colored robots, this indicates that all the boundary robots on $\mathcal{CH}$ have already commenced the current stage. 
        Consequently, it changes its color to \texttt{FAULT-FINISH}, signaling the corner robots to initiate the next stage. 
        If $r$ observes any neighbour with color \texttt{OFF}, it waits. Otherwise, either being a boundary robot with one \texttt{OUTER}-colored neighbour or being an interior robot on $\mathcal{CH}_r$, $r$ changes its color to \texttt{FAULT}.  

        \item \textbf{$r.color =$ \texttt{FAULT}:} $r$ switches to \texttt{FAULT-FINISH} if it is a boundary robot on $\mathcal{CH}_r$ and has at least one neighbour colored \texttt{FAULT-FINISH}.

        \item \textbf{$r.color= $ \texttt{OFF}:}
        If $r$ is a boundary robot on $\mathcal{CH}_r$ and notices a neighbouring robot with color \texttt{FAULT}, it updates its color to \texttt{MOVE2}.
        When $r$ is an interior robot and detects a \texttt{MOVE1}-colored corner robot on $\mathcal{CH}_r$ with a \texttt{FAULT}-colored neighbour, $r$ moves to that corner's location while changing its color to \texttt{MOVE2}.

    \end{itemize}

    At this point, every boundary robot from the initial configuration either reaches a corner position or switches its color to \texttt{FAULT-FINISH}, and we move to the next stage.

    \begin{figure}[h]
    \centering
    \begin{minipage}[b]{0.32\linewidth}
    \centering
        \includegraphics[width=0.85\linewidth]{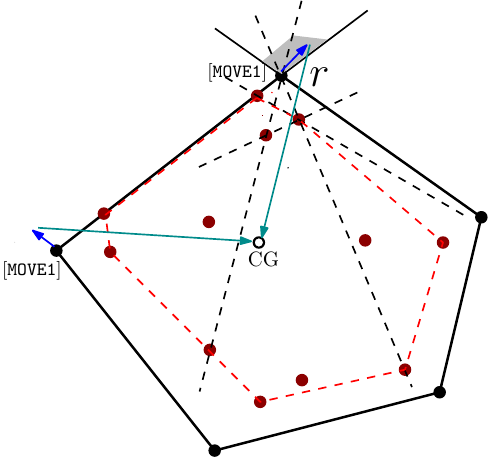}
        \caption{Movement of corner robots to visible area (blue arrow) and then to CG (cyan arrow)}
        \label{fig:corner-to-CG}
    \end{minipage}\hfill
    \begin{minipage}[b]{0.32\linewidth}
    \centering
      \includegraphics[width=0.85\linewidth]{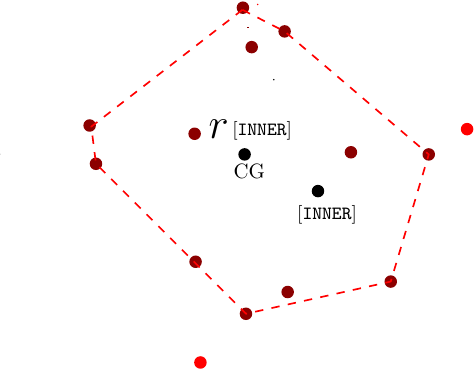}  
    \caption{All non-faulty robots move to CG, while the faulty ones cannot reach CG}
    \label{fig:movement_finish_to-CG}
    \end{minipage}\hfill
    \begin{minipage}[b]{0.32\linewidth}
    \centering
      \includegraphics[width=0.85\linewidth]{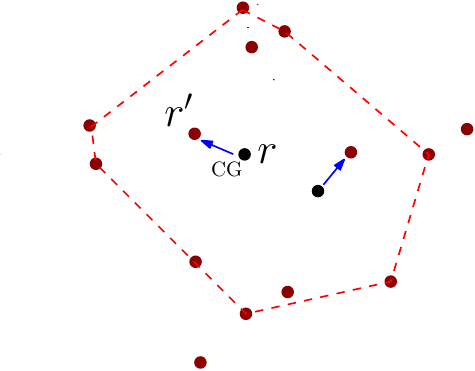}  
    \caption{Non-faulty robots move to nearest \texttt{FAULT-FINISH} robots (brown points) from the CG}
    \label{fig:CG-to-Terminate}
    \end{minipage}
    
\end{figure}

    \noindent \underline{\textbf{Stage }\textsc{Gathering:}} At this stage, two scenarios may arise. If there exists a \texttt{FAULT-FINISH}-colored robot in the configuration, the corner robots first move to their respective exterior visible areas and subsequently proceed to the CG of the convex hull formed by the \texttt{FAULT-FINISH} robots. If, instead, there is no \texttt{FAULT-FINISH} robot, the corner robots move directly to the CG of the convex hull of all visible robots.
    Due to asynchronous activation, some of the corner robots may commence their movement earlier than the others. In such cases, by reusing the colors from previous stages, those robots are redirected back to appropriate prior stages to maintain coordination.

    \begin{itemize}
        \item \textbf{$r.color=$ \texttt{
        MOVE1}:} If a \texttt{FAULT-FINISH} robot is visible to $r$ and $r$ is a corner robot having both of its neighbours with the color from $\{$\texttt{FAULT-FINISH}, \texttt{MOVE1}$\}$, it moves to a point on a point in $ExtVisibleArea(r, \mathcal{CH}_r)$ with the color \texttt{MOVE2}. 
        If no \texttt{FAULT-FINISH}-colored robot is visible to $r$ and $r$ is a corner robot with all robots on $\mathcal{CH}_r$ having the color \texttt{MOVE1}, it computes the CG of $\mathcal{CH}_r$ and moves to it with the color \texttt{INNER}. 
        Due to asynchronous activation, it is possible that some other robots change their color to \texttt{INNER} before $r$. In that case, $r$ changes its color to \texttt{OUTER} to avoid any potential miscalculation of the CG. 

        \item \textbf{$r.color=$ \texttt{MOVE2} and $r$ is a corner robot with one visible \texttt{FAULT-} \texttt{FINISH} robot:} 
        If there is an \texttt{INNER}-colored robot visible to $r$, potentially blocking its view, it switches its color to \texttt{OUTER}. 
        Otherwise, $r$ computes the CG of the convex hull made by the \texttt{FAULT-FINISH}-colored robots.
        If the CG is visible to $r$, it changes its color to \texttt{OFF} without any movement. 
        If instead, CG is not visible, $r$ moves to a point within the region  $ExtVisibleArea(r, \mathcal{CH}_r)$ with the color \texttt{OFF}.

        \item \textbf{$r.color =$ \texttt{OFF} and $r$ is a corner robot with one visible \texttt{FAULT-} \texttt{FINISH} robot:} 
        If there is an \texttt{INNER}-colored robot visible to $r$, potentially blocking its view, it switches its color to \texttt{OUTER}. 
        If a \texttt{MOVE1}-colored robot is visible to $r$ within the interior of $\mathcal{CH}_r$, it still updates its color to \texttt{OUTER}.
        $r$ again computes the CG of the convex hull made by the \texttt{FAULT-FINISH} robots and moves to it with the color  \texttt{INNER}. 

\begin{figure}[h]
    \centering
    \begin{minipage}[b]{0.48\linewidth}
    \centering
        \includegraphics[width=0.86\linewidth]{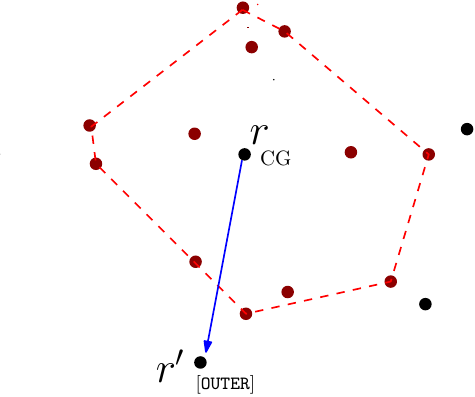}
        \caption{The non-faulty \texttt{INNER} robots moves to the \texttt{OUTER}-colored robot again from the CG.}
        \label{fig:CG-to-outer}
    \end{minipage}\hfill
    \begin{minipage}[b]{0.48\linewidth}
    \centering
      \includegraphics[width=0.86\linewidth]{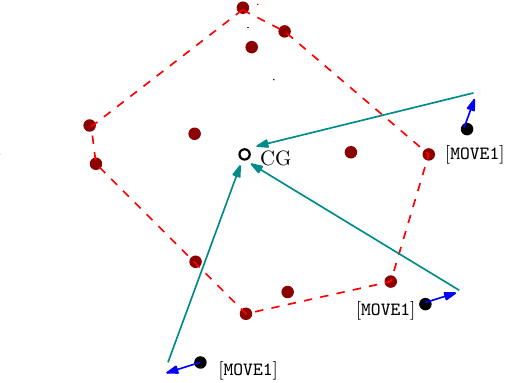}  
    \caption{\texttt{MOVE1}-colored corner robots again moves to the CG.}
    \label{fig:outer-to-CG-again}
    \end{minipage}
    \Description{A technical figure}
\end{figure}

        \item \textbf{$r.color=$ \texttt{MOVE2} and $r$ is a corner robot with no visible \texttt{FAULT-} \texttt{FINISH} robot:} In this case, if $r$ finds itself collocated with \texttt{OUTER} or \texttt{MOVE1}, it adopts the color accordingly. 
        Otherwise, it changes its color to \texttt{FAULT} or \texttt{FAULT-FINISH} by following the steps in the previous stages.

        \item \textbf{$r.color =$ \texttt{INNER} and $r$ finds a \texttt{FAULT-FINISH} robot:} If $r$ is a corner or a boundary robot on $\mathcal{CH}_r$, it changes its color to  \texttt{FAULT}, indicating that it is unable to move to the CG of the \texttt{FAULT-FINISH} robots.
        If $r$ is an interior robot in $\mathcal{CH}_r$, it follows the steps in the stage \textsc{Inner-To-Outer}, where it either moves to the position of an \texttt{OUTER}-colored robot or changes its color to \texttt{FAULT}.

        \item \textbf{$r.color =$ \texttt{INNER} and $r$ finds no \texttt{FAULT-FINISH} robot:} With its current configuration, $r$ is either in the \textsc{Inner-To-Outer} stage or in the current stage \textsc{Gathering}.
        In this case, if $r$ is a corner or a boundary robot on $\mathcal{CH}_r$, it changes its color to \texttt{FAULT}. If $r$ is an interior robot on $\mathcal{CH}_r$, it follows the steps mentioned in the stage \textsc{Inner-To-Outer}.

        \item \textbf{$r.color =$ \texttt{FAULT}:} $r$ moves to the position of the nearest \texttt{FAULT-FINISH} or \texttt{FAULT}-colored robot (if exists) with the color \texttt{FAULT-FINISH} and terminates. In this case, there could be one point where the non-faulty robots are gathered with the color \texttt{FAULT}, while other locations containing \texttt{FAULT}-colored robots have only faulty robots. Such a step ensures that the non-faulty robots gather at the position of the faulty robots and terminate.  
    
    \end{itemize}

\noindent \textbf{Case 2 ($\mathcal{CH}$ is linear):}
We use the same set of colors as in Case 1 (nonlinear configuration).
In this case, the robots move in a manner that preserves the linearity of the configuration throughout.
Movement begins with the robots located at the endpoints of $\mathcal{CH}$, each moving inward toward the position of its neighbour.
We define the notion of terminal robot as follows. 

For a robot $r$, $\mathcal{CH}_r^*$ is the (linear) convex hull of the robots excluding the \texttt{FAULT}-colored robots. $r$ classifies itself as \emph{terminal}, if it has only one neighbour on $\mathcal{CH}_r^*$, otherwise \emph{non-terminal}.

Upon activation, if $r$ finds itself as a terminal robot on $\mathcal{CH}_r^*$, it changes its color to \texttt{OUTER}, and to \texttt{INNER}, otherwise.
Furthermore, any robot that observes a robot colored \texttt{OFF} maintains the status quo.
From now on, the algorithm proceeds according to the following rules, written based on the current color of $r$.

\begin{itemize}
    \item \textbf{$r.color =$ \texttt{OUTER}}: Depending on the current color of the neighbouring robot $r_{Nbr_1}$ on $\mathcal{CH}_r^*$, the terminal robot $r$ does the following. Here, by the neighbouring robot of $r$, we mean the one that is without the color \texttt{FAULT}.
    
    \begin{itemize}
        \item $r_{Nbr_1}.color = $ \texttt{INNER}: $r$ moves to the position of $r_{Nbr_1}$ after updating its color to \texttt{MOVE1} from \texttt{OUTER}.
        This movement, if successfully executed by $r$, shrinks the (linear) convex hull $\mathcal{CH}_r^*$.
        
        \item $r_{Nbr_1}.color =$ \texttt{OUTER}: This case is analogous to the $N=2$ case. 
        It occurs when the non-faulty robots are distributed only between two points -- possibly because $\mathcal{CH}_r^*$ has now become a two-point configuration.
        In this case, 
        $r$ first switches its color to \texttt{MOVE2} and moves to the midpoint of the line segment $\overline{r r_{Nbr_1}}$. 
        This corresponds to the \texttt{OFF}/\texttt{OFF}-rule of the $3$-$\mathcal{LUMI}$ algorithm.
    
        \item $r_{Nbr_1}.color = $ \texttt{MOVE2}: In this case, $r$ maintains the status quo until $r_{Nbr_1}$ updates its color to \texttt{FAULT-FINISH}.
        This corresponds to the \texttt{OFF}/\texttt{MOVE}-rule of the $3$-$\mathcal{LUMI}$ algorithm. 

        \item $r_{Nbr_1}.color = $ \texttt{FAULT-FINISH}: If the location of $r$ and $r_{Nbr_1}$ is not the same, $r$ first changes its color to \texttt{MOVE2} from \texttt{OUTER} and moves to the position of $r_{Nbr_1}$ with the color \texttt{FAULT-FINISH}.
        This corresponds to the \texttt{OFF}/\texttt{END} rule of the $3$-$\mathcal{LUMI}$ algorithm.
    \end{itemize}
    
    \item \textbf{$r.color =$ \texttt{MOVE1}:} In case of $r$ finding itself collocated with an \texttt{INNER}-colored robot, it changes its color to \texttt{INNER}.
    Otherwise, it switches to \texttt{FAULT} to signal that $r$ failed to reach the location of its neighbour, so that any other robot $r'$ exclude it from the computation of $\mathcal{CH}_{r'}^*$.
    
    \item $r.color = $ \texttt{INNER}: If $r$ becomes a terminal robot on $\mathcal{CH}_r^*$ and $r_{Nbr_1}.color = \texttt{INNER}$, it changes its color to \texttt{OUTER}. On the other hand, if $r$ is terminal and $r_{Nbr_1}.color = $ \texttt{OUTER}, it does not change its color. In case of both neighbours of $r$ have the color \texttt{FAULT} -- it means there is no more non-faulty robot left to gather -- $r$ updates its color to \texttt{FAULT-FINISH} and moves to the location of one of the neighbours, and terminates.

    \item $r.color = $ \texttt{MOVE2}: When $r$ is activated with the color \texttt{MOVE2}, it means that $r$ has finished its previous move (to the midpoint of the two-point configuration, or the position of the other robot). Without any movement, $r$ changes its color to \texttt{FAULT-FINISH}.

    \item $r.color = $ \texttt{FAULT-FINISH}: If $r$ finds a \texttt{FAULT-FINISH} neighbouring robot, $r$ moves to the position of it with the color \texttt{MOVE2}. 
    If $r$ sees an \texttt{OUTER}- or \texttt{MOVE2}-colored robot, it waits for that robot to finish its move and turn \texttt{FAULT-FINISH} (corresponding to \texttt{END}/\texttt{OFF} and \texttt{END}/\texttt{MOVE}-rules of 3-$\mathcal{LUMI}$).
    However, if $r$ cannot find any \texttt{FAULT-FINISH}, \texttt{MOVE2} or an \texttt{OUTER}-colored robot, but a \texttt{FAULT}-colored robot, it moves to the position of such a robot, maintaining the current color and terminates.

\end{itemize}

\subsection{Analysis of the Algorithm in 7-\texorpdfstring{$\mathcal{LUMI}$}{LUMI}}
\label{app-subsec:analysis7lumi}
Here we analyse the correctness and the time complexity of the above algorithm.

\begin{restatable}{lem}{lemmaVisiblearea}\label{lemma:visiblearea}
If a non-faulty robot $r$ is a corner of $\mathcal{CH}_r$ and moves to a point either in $IntVisibleArea(r, \mathcal{CH}_r)$ or in $ExtVisibleArea(r, \mathcal{CH}_r)$, it remains a corner on $\mathcal{CH}_r$ and sees all the stationary robots on the plane.
Similarly, for a boundary robot $r$ on $\mathcal{CH}_r$ moving to a point in $BoundaryVisibleArea(r, \mathcal{CH}_r)$, it sees all the stationary robots on the half plane $\mathcal{H}_{L_r}$ that does not contain any robot of $\mathcal{CH}_r$, where $L_r = \overleftrightarrow{rr_{Nbr_1}}$.
\end{restatable}

\begin{proof}
    The proof directly follows from the proof of visible area given by Sharma {\it et al.} \cite{SharmaVTBR16}.
\end{proof}

\begin{restatable}{lem}{LemmaROBOTCategorisation}
\label{lemma:robot-categorize-N-gather}
    In the stage \textsc{Robot-Classification}, all the robots correctly color themselves as \texttt{OUTER} or \texttt{INNER} in $1$ epoch.
\end{restatable}
\begin{proof}
    By construction, a corner, boundary or interior robot $r$ on its local convex hull $\mathcal{CH}_r$ is also a corner, boundary or interior on $\mathcal{CH}$.  Upon detection, a robot updates its
    color accordingly, which requires only 1 epoch. 
\end{proof}

\begin{restatable}{lem}{LemmaStageInnerToGather}
\label{lemma:stage-Inner-to-outer}
    In the stage \textsc{Inner-To-Outer}, all the non-faulty \texttt{INNER}-colored robots move to $\mathcal{CH}^0$ in $O(N_1)$ epochs, where $N_1$ is the total number of interior robots on $\mathcal{CH}$.
\end{restatable}
\begin{proof}
    After the stage \textsc{Robot-Classification}, an \texttt{INNER}-colored robot $r$ correctly detects itself as the corner of $\mathcal{CH}_r^*$, the local convex hull of all the visible robots except the ones with \texttt{OUTER} or \texttt{FAULT}. 
    Such a robot $r$ classifies itself as eligible if it does not find any \texttt{MOVE1} or \texttt{MOVE2}-colored robots.
    In that case, $r$ changes its color to \texttt{MOVE1} and moves to a point in $IntVisibleArea(r, \mathcal{CH}_r^*)$ so that it can see the nearest robot colored \texttt{OUTER} on the outermost layer, ensured by Lemma \ref{lemma:visiblearea}. 
    Thereafter, $r$ moves to the nearest \texttt{OUTER}-colored robot with the color \texttt{MOVE2}, only if there are no \texttt{MOVE2} robots within its visibility. 
    The transition of an eligible $r$ from \texttt{INNER} to \texttt{OUTER} or \texttt{FAULT} (through \texttt{MOVE1} and \texttt{MOVE2}) requires $3$ epochs in total.
    If, instead, there is a \texttt{MOVE2} robot visible to $r$, it waits for that robot to complete its movement and change the color to \texttt{OUTER} or \texttt{FAULT}. The \texttt{MOVE2} robot requires at most $1$ epochs to update its color.
    Therefore, when an \texttt{INNER}-colored eligible robot turns its color either to \texttt{OUTER} or \texttt{FAULT} within $O(1)$ epochs, at least a new eligible robot emerges.
    This implies that an \texttt{INNER}-colored robot eventually becomes eligible. 
    In the {\async} model, the scheduler may activate robots in an order such that some eligible robots execute their part of the algorithm earlier, while others remain inactive or wait until the former complete their actions in $O(1)$ epochs. Hence, the stage \textsc{Inner-To-Outer} can be very sequential in the worst case, leading to $O(N_1)$ runtime.
\end{proof}

\begin{restatable}{lem}{LemmaStageCOnfirmationToCorners}\label{lemma:stage_confirmation_to_outers}
    In the stage \textsc{Confirmation-Signal-To-Outers}, all  \texttt{FAULT}-colored robots change their color to \texttt{FAULT-FINISH} in $O(f_1)$ epochs, where $f_1 (\leq f)$ is the total number of \texttt{FAULT}-colored robots on $\mathcal{CH}$ in this current stage. 
\end{restatable}
\begin{proof}
    There could be at most $f_1$ many \texttt{FAULT}-colored robots located within the interior region of $\mathcal{CH}$. 
    In the worst case, these robots can form $O(f_1)$ many layers. 
    Among these layers, the robots on the innermost layer first change their color to \texttt{FAULT-FINISH} in just 1 epoch. 
    Subsequently, the robots in the next layer update their color to \texttt{FAULT-FINISH} in the following epoch.
    This layer-by-layer propagation continues outward until all \texttt{FAULT}-colored robots have transitioned to \texttt{FAULT-FINISH}, resulting in $O(f_1)$ epochs in total.
\end{proof}

\begin{restatable}{rem}{RemarkConfirmationSignalToOuters}
   The stage \textsc{Confirmation-Signal-To-Outers} ensures that if an \texttt{OUTER}-colored robot $r$ finds the interior of $\mathcal{CH}_r^*$ empty or occupied solely by \texttt{FAULT-FINISH} robots, it correctly infers that all the non-faulty interior robots of $\mathcal{CH}$ have already relocated to $\mathcal{CH}^0$.
\end{restatable}

\begin{restatable}{lem}{LemmaStageBoundaryToCorner}
    \label{lemma:stage-boundary-to-corner}
    In the stage \textsc{Boundary-To-Corner}, all the non-faulty \texttt{OUTER}-colored boundary robots on $\mathcal{CH}^0$ move to the corners of $\mathcal{CH}^0$ in $O(N_2)$ epochs, where $N_2$ is the total number of boundary robots on $\mathcal{CH}$.
\end{restatable}

\begin{proof}
    At the beginning of this stage, a corner robot $r$ changes its color to \texttt{MOVE1} to signal the adjacent boundary robots, let one of them be $r_{Nbr_1}$.
    Thereafter, $r_{Nbr_1}$ moves toward $r$ with \texttt{MOVE2}.
    When non-faulty, $r_{Nbr_1}$ successfully executes this movement and adopts \texttt{MOVE1}.
    Alternatively, if it is faulty, $r_{Nbr_1}$ adopts \texttt{FAULT} if its other neighbour $r' \neq r$ has the color \texttt{OUTER}. 
    When $r'$ observes $r_{Nbr_1}.color =$ \texttt{FAULT}, it changes its color to \texttt{OFF} and moves to the nearest \texttt{FAULT-FINISH} robot or to the interior of $\mathcal{CH}_r$ depending of the existence of \texttt{FAULT-FINISH} robots in the interior of $\mathcal{CH}_r$.
    The movement of the non-faulty robot $r'$ enables it to see the adjacent corner robots colored \texttt{MOVE1} and hence it executes another movement to such a corner with the color \texttt{MOVE2}.
    If instead, $r'$ is faulty and cannot move, $r'$ remains stationary on the boundary of $\mathcal{CH}_r$ and changes its color to \texttt{FAULT} from \texttt{OFF}.
    The similar strategy is sequentially adopted by other boundary robots.
    Thus, all non-faulty boundary robots on a specific side of $\mathcal{CH}^0$ move to their adjacent corners after one by one, starting from the ones adjacent to the corner robots, which requires $O(N_2)$ epochs in the worst case as one of the sides of $\mathcal{CH}^0$ may contain $O(N_2)$ many robots.
\end{proof}

\begin{restatable}{rem}{RemarkMOVE1ConfirmationFromFault}
    A \texttt{MOVE1}-colored corner robot concludes that the non-faulty boundary robots from its adjacent sides have all relocated to the corner positions if it finds its neighbouring robots with the color \texttt{FAULT-FINISH}. This follows from the fact that a robot turns its color to \texttt{FAULT-FINISH} only when both of its neighbours have the color in $\{\texttt{MOVE1}, \texttt{MOVE2}, \texttt{FAULT}\}$ for the first time. Subsequently, the color \texttt{FAULT-FINISH} propagates along the boundary towards the adjacent corners.
\end{restatable}

\begin{restatable}{lem}
{LemmastageGathering}
    \label{lemma:stage-gathering}
    In the stage \textsc{Gathering}, all the non-faulty robots eventually terminate and gather at a single point in $O(N)$ epochs.
\end{restatable}

\begin{proof}
    We divide the proof into two cases based on whether there exists some \texttt{FAULT-FINISH}-colored robots.
    We first deal with the case when there is at least one \texttt{FAULT-FINISH}-colored robot.
    A non-faulty \texttt{MOVE1}-colored corner robot $r$ moves to a point in $ExtVisibleArea(r, \mathcal{CH}_r)$ with the color \texttt{MOVE2}, after observing both its neighbours with the color \texttt{FAULT-FINISH} or \texttt{MOVE1}. 
    After the successful movement, $r$ can see all the \texttt{FAULT-FINISH} robots by Lemma \ref{lemma:visiblearea}, implying that the robot $r$ can compute the CG of the convex hull made by the \texttt{FAULT-FINISH}-colored robots. 
    However, the CG may not be visible to $r$. 
    In that case, $r$ again executes a movement to $ExtVisibleArea(r, \mathcal{CH}_r)$ with the color \texttt{OFF} so that the CG becomes visible.
    Finally, it moves to the CG after switching the current color to \texttt{INNER}.
    Any other non-faulty robots, if simultaneously activated with $r$, compute the same CG as $r$, because at this stage, the interior of $\mathcal{CH}_r$ does not contain any robot except the \texttt{FAULT-FINISH} ones, which are always stationary.
    If there are no other corner robots yet to execute this move to CG, the robot $r$ colored \texttt{INNER} terminates within $4$ epochs using the sequence of color transitions \texttt{INNER} $\rightarrow$ \texttt{MOVE1} $\rightarrow$ \texttt{MOVE2} $ \rightarrow$ \texttt{FAULT} $\rightarrow$ \texttt{FAULT-FINISH}. 
    In the worst case, let $k_1$ be the number of robots that move to the CG simultaneously with $r$. 
    By the previous arguments, this movement is done in $O(1)$ epochs.
    After this movement to CG, the $k_1$ robots follow the stage \textsc{Inner-To-Outer}. 
    By Lemmas \ref{lemma:stage-Inner-to-outer} and \ref{lemma:stage_confirmation_to_outers}, all those $k_1$ robots move to the corners of the current $\mathcal{CH}$ or change their color to \texttt{FAULT-FINISH} in $O(k_1)$ epochs. 
    At this point, the total number of corner robots of $\mathcal{CH}$ is $N-k_1-f_1$, where $f_1$ is the number of faulty robots before the time when $r$ moves to the CG. 
    Next, we consider that $k_2$ ($k_2 \leq N-k_1-f_1$) robots move to the CG in sync.
    By a similar argument, those robots again move to the location of corners in $O(k_2)$ epochs. 
    The number of corner robots at this time is $N-k_1-k_2-f_1$. 
    The process of movement from corner to CG and again CG to corners continues till we find a $k_i$ such that $N - k_1-k_2- \cdots k_{i} -f_1 = 0$, i.e., there are no corner robots left to gather.
    Hence, the termination requires $O(k_1+k_2+\cdots+k_i) \approx O(N)$ epochs in the worst case.
    
     We now prove the case when there are no \texttt{FAULT-FINISH} robots visible to $r$. In this case, all the robots are corner robots and visible to each other.
     The robot $r$ correctly computes the CG of all the \texttt{MOVE1}-colored robots, and so are the other corner robots.
     By a similar argument as above, the termination requires $O(N)$ epochs in the worst case due to a sequential activation under {\async}. 
\end{proof}

\begin{restatable}{lem}
 {LemmalinearNgather}
    \label{lemma:linear-N-gather}
    If the initial configuration $\mathcal{CH}$ is linear, all the non-faulty robots eventually terminate and gather at a single point in $O(N)$ epochs.
\end{restatable}
\begin{proof}
Throughout the linear case, every move keeps the robots collinear: a robot moves only onto the position of a neighbour, to the midpoint of the two remaining points, or onto a faulty robot, and each of these lies on the line $\mathcal{CH}$. Hence the configuration stays linear at all times, and $\mathcal{CH}_r^*$, the linear convex hull obtained after discarding the \texttt{FAULT}-colored robots, is well defined at every step for a robot $r$.

    When the configuration is linear, the terminal robots change their color to \texttt{OUTER}, while the non-terminals to \texttt{INNER}.
    We divide the proof into two cases: (Case i) $N=2$, and (Case ii) $N>2$.

    When $N=2$, both robots have the color \texttt{OUTER}, allowing any robot to correctly detect that it is a part of Case i. The algorithm for this case is analogous to the 3-$\mathcal{LUMI}$ algorithm in Section~\ref{sec:impossibility},with the only difference being the name of colors. In 7-$\mathcal{LUMI}$ algorithm, we use \texttt{OUTER}, \texttt{MOVE2} and \texttt{FAULT-FINISH} instead of \texttt{OFF}, \texttt{MOVE} and \texttt{END} in 3-$\mathcal{LUMI}$. 
    Therefore, the lemma follows directly from Theorem \ref{thm:gathering2robots}.

    When $N>2$, a terminal robot $r$ moves to the position of its neighbour $r_{Nbr_1}$, when it finds $r_{Nbr_1}.color= $ \texttt{INNER}. If $r$ is non-faulty, $r$ collocates with the \texttt{INNER}-colored robot and changes its color to \texttt{INNER} within $2$ epochs. On the other hand, a faulty robot $r$ switches to \texttt{FAULT} in the same number of epochs, as it fails to reach $r_{Nbr_1}$. In the next epoch, $r_{Nbr_1}$ becomes terminal on its local $\mathcal{CH}_{r_{Nbr_1}}^*$ and performs the same process as $r$. These steps will continue until the configuration is left with only one or two positions where the non-faulty robots are gathered. For the latter case, it will redirect to the previous Case (i). In the earlier case, where all the positions in the current configuration are occupied by the faulty robots with the color \texttt{FAULT} or \texttt{FAULT-FINISH} except one at $r$, the robot $r$ moves to one of the faulty robots' positions and terminates. 
    In the worst-case scenario, each terminal robot moves to its neighbouring position in exactly two epochs, resulting in an overall time complexity of $O(N)$ epochs for the algorithm.
\end{proof}

All the above lemmas conclude the following theorem.

\begin{restatable}{theo}{theoremNgatherOrderN}\label{thm:N-gather-a}
    There is a deterministic algorithm that solves gathering in an $(N,f)$-mobility system, $f<N$, $f, N$ not known, under {\async} in the 7-$\mathcal{LUMI}$ model. The algorithm has a runtime of $O(N)$ epochs and works under obstructed visibility. 
\end{restatable}

\section{26-\texorpdfstring{$\mathcal{LUMI}$}{LUMI} Algorithm under \texorpdfstring{$\mathcal{ASYNC}$}{ASYNC}}
\label{sec:lfalgorithm} 

We now  present the 26-$\mathcal{LUMI}$ algorithm that takes $O(\max\{\ell, f\})$ epochs in total. 
Also an execution example is provided in Figs. \ref{fig:boundary-to-interior-2}-\ref{fig:final-faulterminate-case2-2} for better presentation.
Together, these two algorithms highlight a time–color trade-off.

Initially, all the robots have the color \texttt{OFF}. 
Based on the initial configuration, we identify three cases: Case 3 handles linearity, while Cases 1 and 2 handle non-linear configuration, which is further divided into 8 major logical stages.


\paragraph{\underline{Stage \textmd{\textsc{Robot-Categorisation:}}}} After activation with \texttt{OFF}, $r$ classifies itself as a corner, boundary, or interior robot on $\mathcal{CH}_r$ and updates its color to \texttt{CORNER}, \texttt{BOUNDARY}, or \texttt{INTERIOR} accordingly (Fig.~\ref{fig:boundary-to-interior-2}). If it then detects any \texttt{OFF}-colored robot in its view, it simply waits without further action.

\begin{figure}[t]
\centering
\begin{minipage}[b]{0.32\linewidth}
\centering
\includegraphics[width=\linewidth]{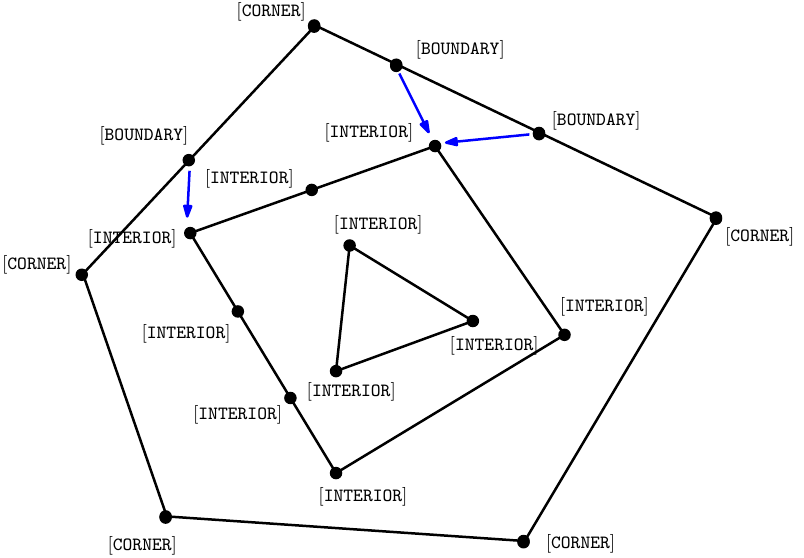}
\caption{The boundary robots first move to the nearest interior robots after Stage \texttt{Robot-Categorisation}}
\label{fig:boundary-to-interior-2}
\end{minipage}\hfill
\begin{minipage}[b]{0.32\linewidth}
\centering
\includegraphics[width=\linewidth]{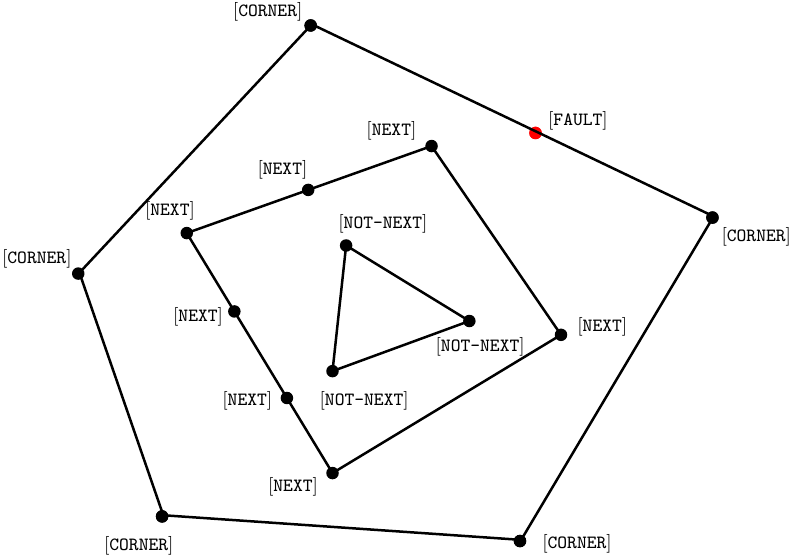}
\caption{Detecting the eligible layer of robots, in which robots on the second outermost layer with \texttt{NEXT}}
\label{fig:next-notnext-2}
\end{minipage}\hfill
\begin{minipage}[b]{0.32\linewidth}
\centering
\includegraphics[width=\linewidth]{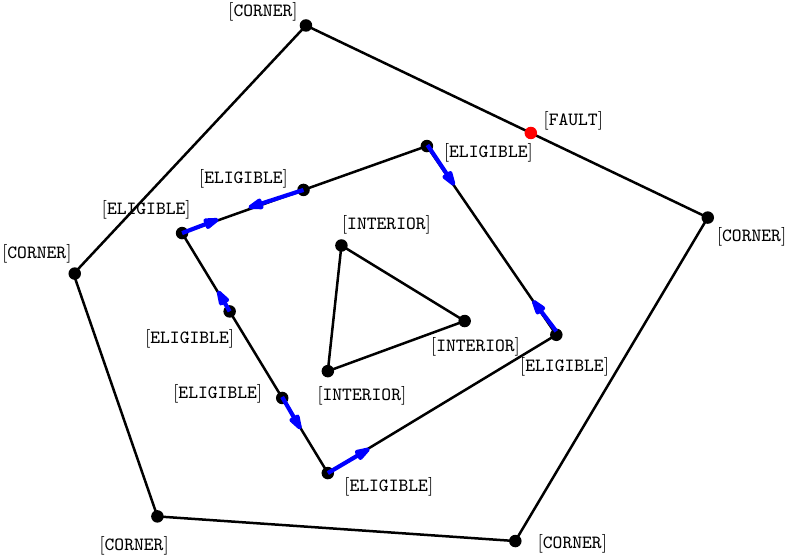}
\caption{Robots with color \texttt{ELIGIBLE} move to their respective $BoundaryVisibleArea()$}
\label{fig:eligible-layer-movement-visiblearea-2}
\end{minipage}
\Description{A technical figure}
\end{figure}

A configuration, including the initial one, may have no interior robots. In such cases, after the first epoch, only \texttt{CORNER}- and \texttt{BOUNDARY}-colored robots are visible, and we proceed through the following cases.

\paragraph{Case 1 (At least one \texttt{INTERIOR}-colored robot exists):} In this case, we begin with the following stage.

\paragraph{\underline{Stage \textmd{\textsc{Boundary-To-Interior}:}}}
At this stage, \texttt{BOUNDARY}-colored robots move toward nearby \texttt{INTERIOR}-colored robots with the color \texttt{MoveTo\_INTERIOR} to become an interior robot. 
\begin{itemize}[leftmargin=*]
    \item \textbf{$r.color =$ \texttt{BOUNDARY}:} In this case, $r$ moves to the position of the nearest \texttt{INTERIOR}-colored robot after changing its current color to \texttt{MoveTo\_INTERIOR}.

    \item \textbf{$r.color = $ \texttt{MoveTo\_INTERIOR}:} If $r$ finds itself collocated with another \texttt{INTERIOR}-colored robot, it changes its current color to \texttt{INTERIOR}. Otherwise, it updates its color to \texttt{FAULT}, (Fig. \ref{fig:boundary-to-interior-2}) as it may have experienced a mobility fault before reaching its intended destination. 

\end{itemize}


 So far, an \texttt{INTERIOR}-colored robot remains idle if it detects any robot with color \texttt{BOUNDARY} or \texttt{MoveTo\_INTERIOR}. 
 Once no such robots are visible, $r$ being an \texttt{INTERIOR}-colored robot, proceeds to check whether it qualifies as an \emph{eligible robot} for the next stage. This eligibility check is to identify \texttt{INTERIOR}-colored robots that can move to the positions of \texttt{CORNER}-colored robots currently located on the outermost layer $\mathcal{CH}^0$.
 
 An \texttt{INTERIOR}-colored robot $r$ is called an \textbf{eligible robot}, if it lies on the boundary (including corners) of $\mathcal{CH}_r^*$, where $\mathcal{CH}_r^*$ is the convex hull of all visible robots to it, excluding those with colors \texttt{CORNER} or \texttt{FAULT}.
Notice that once non-faulty boundary robots on $\mathcal{CH}^0$ move to interior positions, the robots on the next layer $\mathcal{CH}^1$ are the first to become eligible.

\paragraph{\underline{Stage \textmd{\textsc{Detecting-Eligible-Layer}:}}}
 An \texttt{INTERIOR}-colored robot $r$ stays idle if it sees any robot with the color in $\{$\texttt{ELIGIBLE}, \texttt{MOVE1}, \texttt{MOVE2}$\}$.
Else, it checks eligibility: if eligible, it switches to \texttt{NEXT}; if not, to \texttt{NOT-NEXT}. 
A \texttt{NEXT}-colored robot waits while any \texttt{INTERIOR}-colored robot is visible. 
Once all \texttt{INTERIOR} robots are recolored to \texttt{NEXT}, \texttt{NOT-NEXT}, or \texttt{ELIGIBLE}, each \texttt{NEXT}-colored robot switches to \texttt{ELIGIBLE}, and a \texttt{NOT-NEXT} one reverts to \texttt{INTERIOR} (refer to Figs.~\ref{fig:next-notnext-2},  and \ref{fig:eligible-layer-movement-visiblearea-2}).

Robots with the color \texttt{ELIGIBLE} form the eligible layer. The color transitions are designed to prevent path crossings.
For instance, if an eligible robot $r$ starts moving toward a \texttt{CORNER}-colored robot, and another robot $r'$, not in the eligible layer but visible to $r$, also becomes eligible, a conflict may arise. By marking $r$ as \texttt{ELIGIBLE} and keeping $r'$ as \texttt{INTERIOR}, we ensure that $r'$ must wait upon seeing $r$ in \texttt{ELIGIBLE}, \texttt{MOVE1}, or \texttt{MOVE2}.

\paragraph{\underline{Stage \textmd{\textsc{Interior-To-Corners:}}}}
$r$ waits till it becomes eligible. 
The goal is to move these eligible robots (colored \texttt{ELIGIBLE}) to the corner positions of $\mathcal{CH}^0$, and continue with those on $\mathcal{CH}^2$, and so on, proceeding layer by layer up to $\mathcal{CH}^{l-1}$.  
An eligible robot $r$ goes through the following steps.

\begin{itemize}[leftmargin=*]
    \item \textbf{$r.color = $ \texttt{ELIGIBLE}:} 
    In this case, $r$ can be a boundary or a corner robot on $\mathcal{CH}_r^*$, the convex hull of all robots after disregarding the visible \texttt{CORNER, FAULT} and \texttt{MOVE2}-colored robots. 
    At this step, we want $r$ to move to a point in such a way that it can see the nearest \texttt{CORNER}-colored robots. 
    So, $r$ computes a half-plane $\mathcal{H}_{L_r}$ as follows, by identifying a line $L_r$ passing through it.
        \begin{itemize}
            \item \textbf{$r$ is a boundary robot on $\mathcal{CH}_r^*$:} Let $r_{Nbr_1}$ and $r_{Nbr_2}$ be the two neighbours of $r$ on $\mathcal{CH}_r^*$. The half-plane $\mathcal{H}_{L_r}$ is the open half-plane delimited by the line $L_r = \overleftrightarrow{r_{Nbr_1}r_{Nbr_2}}$ that contains no visible \texttt{ELIGIBLE, MOVE1} or \texttt{INTERIOR}-colored robots.
        
            \item \textbf{$r$ is a corner robot on $\mathcal{CH}_r^*$:} In this case, $L_r$ is the line passing through the current location of $r$ and parallel to $\overleftrightarrow{r_{Nbr_1} r_{Nbr_2}}$ .
            So, $\mathcal{H}_{L_r}$ is the open half-plane, delimited by the line $L_r$ that does not have any \texttt{ELIGIBLE, MOVE1} or \texttt{INTERIOR}-colored robots.

            \item \textbf{$r$ is lying on a linear $\mathcal{CH}_r^*$:} In this case, if $r$ is a extreme robot on the line segment $\mathcal{CH}_r^*$, it cannot have two neighbours. So, $r$ simply chooses $\mathcal{CH}_r^*$ as the line $L_r$.
        \end{itemize}        
        
         Let $p_{mid_1}$ and $p_{mid_2}$ be the midpoints of the line segments $\overline{rr_{Nbr_1}}$ and $\overline{rr_{Nbr_2}}$ respectively. 
         $r$ arbitrarily selects one of the midpoints (say $p_{mid_1}$) and computes $BoundaryVisibleArea(r, \mathcal{CH}_r^*)$. 
         It then determines a target point $t_r$ within this area, changes its color to \texttt{MOVE1} and moves to $t_r$.
         The purpose of this movement is to let $r$ see all the \texttt{CORNER}-colored robots on the half-plane $\mathcal{H}_{L_r}$.

    \item  \textbf{$r.color = $ \texttt{MOVE1}:} When $r$ gets activated with \texttt{MOVE1}, it indicates that $r$ being an \texttt{ELIGIBLE}-colored robot has moved to its visible area.
    So it again finds $\mathcal{CH}_r^*$ and the half-plane $\mathcal{H}_{L_r}$ in the same way it is described in the earlier case (as in $r.color =$ \texttt{ELIGIBLE}).
    However, $r$ may not see any \texttt{CORNER}-colored robots on $\mathcal{H}_{L_r}$ due to the fault or a \texttt{MOVE2}-colored robot positioned between $r$ and the \texttt{CORNER}-colored robot. If a \texttt{CORNER}-colored robot is visible to $r$, we differentiate the following sub-cases.
    \begin{itemize}
        \item \textbf{$r$ is a boundary robot on $\mathcal{CH}_r^*$:} The robot $r$ sets the position of the closest \texttt{CORNER}-colored robots on $\mathcal{H}_{L_r}$ as its target $t_r$.

        \item \textbf{$r$ is a corner robot on $\mathcal{CH}_r^*$:} Let us denote the open half-plane delimited by the line $\overleftrightarrow{rr_{Nbr_1}}$ (resp. $\overleftrightarrow{rr_{Nbr_2}}$) by $\mathcal{H}'_{L_r}$ (resp. $\mathcal{H}''_{L_r}$) that does not have any \texttt{ELIGIBLE, MOVE1} or \texttt{INTERIOR}-colored robots, as shown in Fig. \ref{fig:eligible-to-corner-2} . 
        If $\mathcal{H}'_{L_r}$ contains any visible \texttt{CORNER}-colored robot, then $r$ sets $t_r$ as the position of the nearest such robot. Otherwise, $t_r$ is set to the position of the nearest visible \texttt{CORNER}-colored robot on $\mathcal{H}''_{L_r}$. 
    \end{itemize}

    Finally, it moves to the point $t_r$ after changing its color to \texttt{MOVE2}. In case of $r$ not seeing any \texttt{CORNER}-colored robot, $r$ changes the color to \texttt{FAULT} if it finds no \texttt{MOVE2} robots on $\mathcal{H}_{L_r}$, and to \texttt{ELIGIBLE} otherwise.

    \item \textbf{$r.color =$ \texttt{MOVE2}:} If $r$ finds itself collocated with a \texttt{CORNER}-colored robot, it changes its color to \texttt{CORNER}. 
    Otherwise, it must be a faulty robot that could not reach its destination. So, it changes its color to \texttt{FAULT}.

\end{itemize}

By following the steps outlined above, every non-faulty robot with the color \texttt{INTERIOR} eventually reaches the positions of \texttt{CORNER}-colored robots on the outermost layer $\mathcal{CH}^0$. 
The faulty robots, meanwhile, remain inside the convex hull $\mathcal{CH}$ with the color \texttt{FAULT}.
Now the \texttt{CORNER}-colored robots need to be informed when to begin computing the gathering point. 

\begin{figure}[h]
\centering

\begin{minipage}[b]{0.32\linewidth}
\centering
\includegraphics[width=\linewidth]{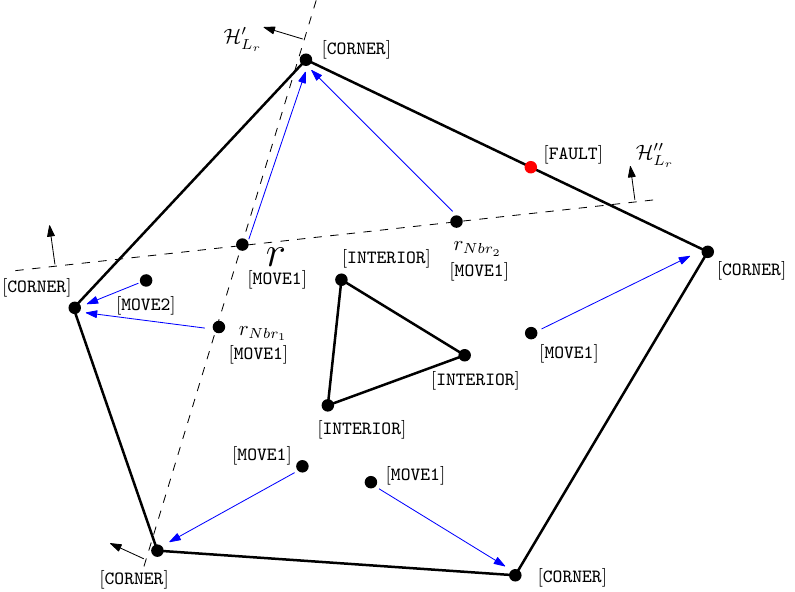}
\caption{A \texttt{MOVE1} robot $r$ computes $\mathcal{H}_{L_r}'$ and $\mathcal{H}_{L_r}''$ to find a \texttt{CORNER} robot in one of these regions to move its position}
\label{fig:eligible-to-corner-2}
\end{minipage}\hfill
\begin{minipage}[b]{0.32\linewidth}
\centering
\includegraphics[width=\linewidth]{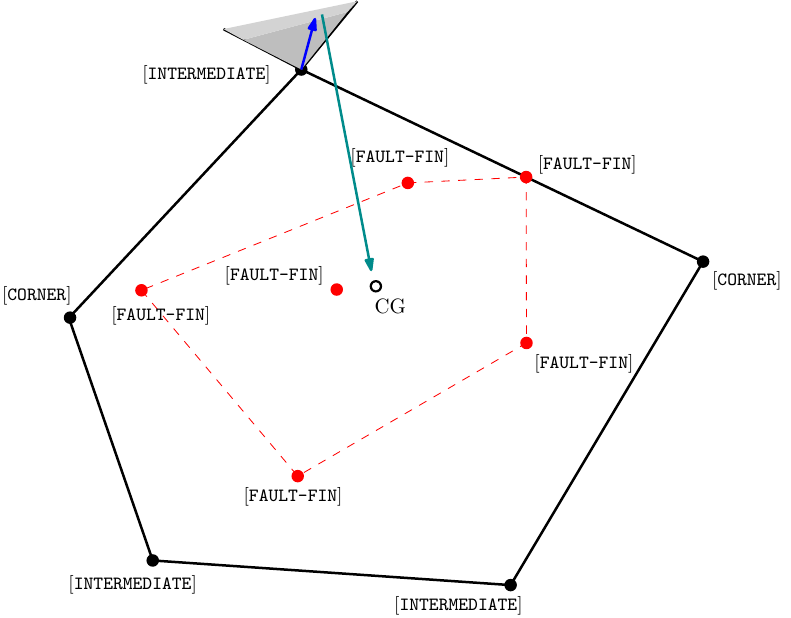}
\caption{A corner robot moves to $ExtVisibleArea()$ and then to CG of the convex hull of \texttt{FAULT-FINISH} robots}
\label{fig:corner-to-CG-2}
\end{minipage}\hfill
\begin{minipage}[b]{0.32\linewidth}
\centering
\includegraphics[width=\linewidth]{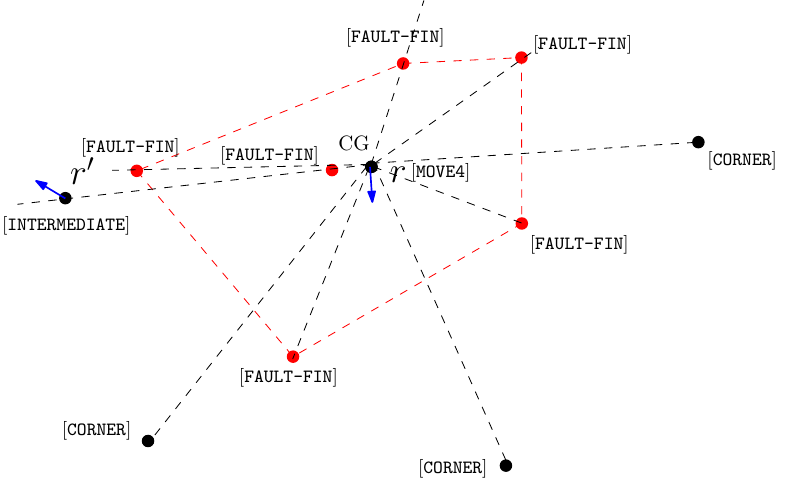}
\caption{A \texttt{MOVE4} robot $r$ moves within its boundary visible area from CG, while $r'$ moves to $ExtVisibleArea()$}
\label{fig:CG-to-visible-area-2}
\end{minipage}
\Description{A technical figure}
\end{figure}

\paragraph{\underline{Stage \textmd{\textsc{Confirmation-Signal-To-Corners}:}}}
We proceed to move the robots colored \texttt{CORNER} only after ensuring that no non-faulty interior robot is left, including any potentially obscured by \texttt{FAULT}-colored robots.
So, we make the \texttt{FAULT}-colored robots update their colors to serve as confirmation signals (using a color \texttt{FAULT-FINISH}) for the \texttt{CORNER}-colored robots.

 When a \texttt{FAULT}-colored robot $r$ finds a robot with a color in the set $\{\texttt{INTERIOR},$ $ \texttt{MOVE1}, \texttt{MOVE2}\}$, it does nothing. 
    Otherwise, it calculates the local layers $\mathcal{CH}_r^0, \mathcal{CH}_r^1,$ $\cdots, \mathcal{CH}_r^{l'}$ using the visible robots in $\mathcal{V}_r$ and identifies the layer $\mathcal{CH}_r^i$ to which it belongs. 
    If the interior of the convex hull formed by the robots on $\mathcal{CH}_r^i$ is either empty or contains robots only with the color \texttt{FAULT-FINISH}, $r$ changes its color to \texttt{FAULT-FINISH}. 
    The color \texttt{FAULT-FINISH} acts as a signal for the \texttt{CORNER}-colored robots to initiate their movement.

Observe that, at this point, if we reconstruct the (global) layers based on the current positions of the robots on the plane, the \texttt{FAULT}-colored robots on the current innermost layer are the first to change their color to \texttt{FAULT-FINISH}, followed by those on the second innermost layer, and so on outward.

Now we move the robots with the color \texttt{CORNER} lying on $\mathcal{CH}^0$ to the center of gravity (CG) of the convex hull made by the robots with the color \texttt{FAULT-FINISH} in the following stage.
However, there may not be any \texttt{FAULT-FINISH}-colored robots lying inside $\mathcal{CH}$, i.e., all the boundary and interior robots of $\mathcal{CH}$ are non-faulty. 
In this case, a \texttt{CORNER}-colored robot $r$ finds the interior of $\mathcal{CH}_r$ empty and proceeds as outlined in Case 2.
We now continue with the scenario in which at least one \texttt{FAULT-FINISH}-colored robot is present within the interior.

\paragraph{\underline{\textbf{Stage} \textmd{\textsc{Corners-To-CG}}:}}
To compute the CG, all \texttt{FAULT-FINISH}-colored robots (all of them are stationary) must be visible to a \texttt{CORNER}-colored robot $r$. $r$ first switches to \texttt{INTERMEDIATE}, then moves to $ExtVisibleArea(r,\mathcal{CH}_r)$ with color \texttt{MOVE3} to ensure visibility of all \texttt{FAULT-FINISH} robots. Upon activation with \texttt{MOVE3}, it computes the CG of the convex hull formed by robots with the color \texttt{FAULT-FINISH} and \texttt{FAULT-TERMINATE}. If the CG is visible, $r$ moves to it with the color \texttt{MOVE4} (see Fig.~\ref{fig:corner-to-CG-2}); otherwise, it moves again within $ExtVisibleArea(r,\mathcal{CH}_r)$ with \texttt{OFF} until the CG becomes visible, then moves to it upon reactivation. After \texttt{MOVE4}, $r$ switches to \texttt{TERMINATE} if collocated with a \texttt{FAULT-FINISH}-colored robot; otherwise to \texttt{MOVE-ENDED}. Finally, at this stage, if $r$ detects a robot with color in $\{$\texttt{MOVE4, MOVE-ENDED, MOVE5, MoveTo-CORNER}$\}$, it switches back to \texttt{CORNER}.

Up to this point, the goal is to gather the \texttt{CORNER}-colored robots at the CG of the convex hull formed by the faulty robots.
So, we compute the point of gathering as follows. 

\begin{figure}[t]
\centering
\begin{minipage}[b]{0.32\linewidth}
\centering
\includegraphics[width=\linewidth]{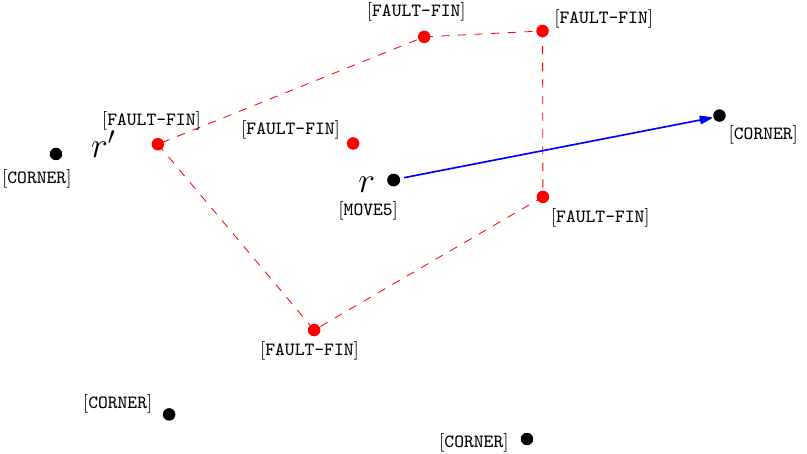}
\caption{\texttt{MOVE5} robot $r$ moves to the position of \texttt{CORNER} robot, which will be the gathering point for other non-faulty ones}
\label{fig:move5-to-corner-2}
\end{minipage}\hfill
\begin{minipage}[b]{0.32\linewidth}
\centering
\includegraphics[width=\linewidth]{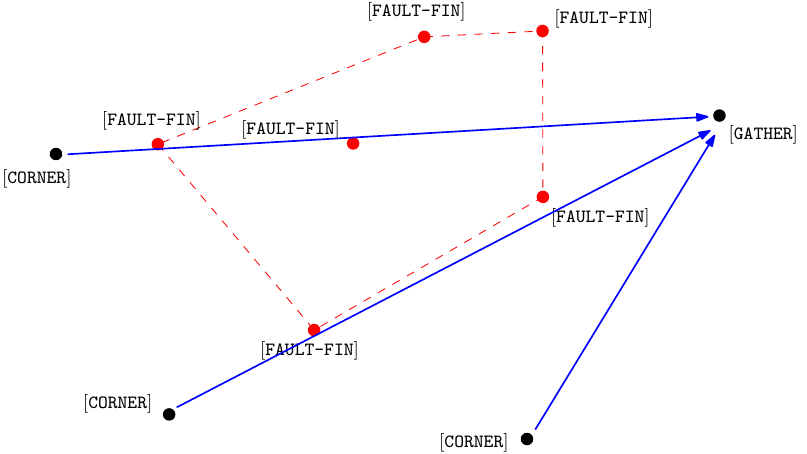}
\caption{After seeing \texttt{GATHER}-colored robot, all other non-faulty corner robots move to that position for gathering}
\label{fig:gather-at-corner-2}
\end{minipage}\hfill
\begin{minipage}[b]{0.32\linewidth}
\centering
\includegraphics[width=\linewidth]{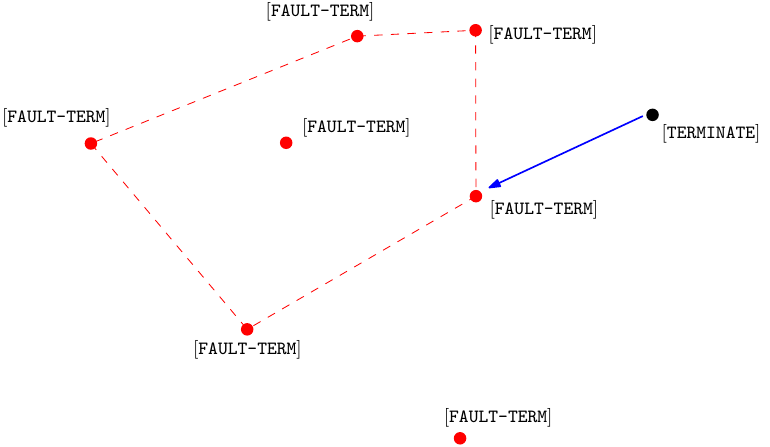}
\caption{All gathered non-faulty robots (with \texttt{TERMINATE}), together move to the position of a \texttt{FAULT-TERMINATE} robot}
\label{fig:terminate-to-fault-2}
\end{minipage}
\Description{A technical figure}
\end{figure}

\paragraph{\underline{\textbf{Stage} \textmd{\textsc{Finding-Gathering-Point}}:}} 
The intended \textit{gathering point} is the CG, with all non-faulty robots moving towards it. However, due to asynchronous activation, some of the \texttt{CORNER}-colored robots may remain idle while others move towards the CG.
It is also possible that some \texttt{CORNER}-colored robots may become faulty during this movement, potentially resulting in multiple locations on the plane occupied by \texttt{MOVE-ENDED}-colored robots.
Among these, it is evident that at most one location contains non-faulty robots; however, the robots cannot determine which one it is.
Therefore, each such robot moves within its visible area to locate a \texttt{CORNER}-colored robot and then proceeds to its position as follows. 
In such cases, our goal is to redirect the robots that have already reached the CG back to the position of an idle \texttt{CORNER}-colored robot, which will then serve as the final \textit{gathering point}. 

\begin{itemize}[leftmargin=*]
    \item \textbf{$r.color =$ \texttt{MOVE-ENDED}:} If $r$ sees any \texttt{MOVE4}-colored robots, it waits to confirm their movement has been completed. Otherwise, it calculates its target point $t_r$ inside the region $BoundaryVisibleArea(r, \mathcal{CH}_r^*)$, where $\mathcal{CH}_r^*$ is the convex hull formed solely by $r$ itself.  
    It then moves to the point $t_r$ with the color \texttt{MOVE5} (see Fig. \ref{fig:CG-to-visible-area-2}).
    This movement allows $r$ to gain visibility of any \texttt{CORNER}-colored robot present on the outermost layer.

    \item \textbf{$r.color =$ \texttt{MOVE5}:} 
    If $r$ observes any \texttt{MOVE3} or \texttt{INTERMEDIATE}-colored robots, it waits for them to turn into \texttt{CORNER}.
    When $r$ finds a \texttt{CORNER}-colored robot, it moves to its position with the color \texttt{MoveTo-CORNER}, as shown in Fig. \ref{fig:move5-to-corner-2}. 
    If this move is successfully executed without fault, this position becomes the gathering point for other non-faulty robots.
    It is possible that all \texttt{CORNER}-colored robots simultaneously move to the CG of the convex hull made by the \texttt{FAULT-FINISH}-colored robot from their respective exterior visible areas, as shown in Fig. \ref{fig:corner-toCG-case2-2}. 
    We want to move all the non-faulty robots, currently located at a single position on the plane, to the position of a \texttt{FAULT-FINISH} robot.
    To do this, when $r$ does not find any \texttt{CORNER}-colored robot but a \texttt{FAULT-FINISH} colored robot, it moves to its position with the color \texttt{MoveTo-CORNER}, as shown in Fig. \ref{fig:move5-to-faultfinish-case2-2}.
    The robot changes its color to \texttt{FAULT-FINISH} if it does not find any visible \texttt{CORNER} or \texttt{FAULT-FINISH}.

    \item \textbf{$r.color =$ \texttt{MoveTo-CORNER}:} Upon reaching its destination, $r$ checks whether it is collocated with a \texttt{CORNER}-colored robot. If so, it recognises this as a valid gathering point and updates its color to \texttt{GATHER}. On the other hand, if $r$ finds itself collocated with a \texttt{FAULT-FINISH}-colored robot, it switches to \texttt{TERMINATE}.
    If no such robots are found at its location, $r$ assumes the movement was unsuccessful, due to a fault during execution, and switches its color to \texttt{FAULT-FINISH}.

\end{itemize}

    \begin{figure}[t]
\centering
\begin{minipage}[b]{0.32\linewidth}
\centering
\includegraphics[width=\linewidth]{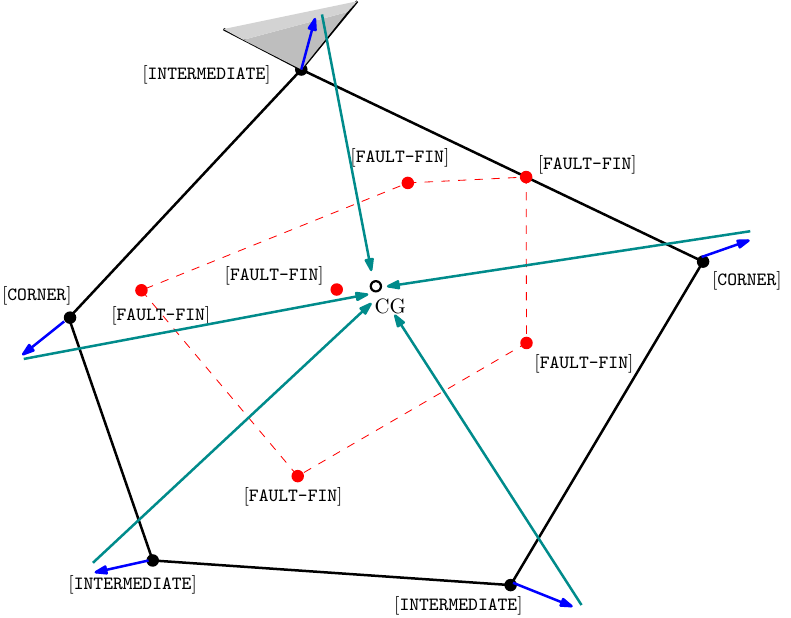}
\caption{It is also possible that all \texttt{CORNER} robots move to CG simultaneously.}
\label{fig:corner-toCG-case2-2}
\end{minipage}\hfill
\begin{minipage}[b]{0.32\linewidth}
\centering
\includegraphics[width=\linewidth]{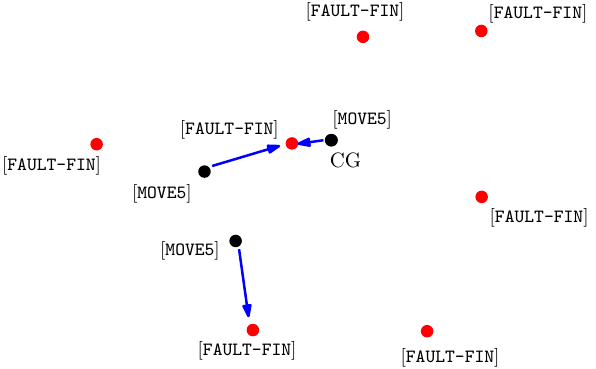}
\caption{Non-faulty robots reach to CG. All \texttt{MOVE5} robots move to \texttt{FAULT-FINISH} robots}
\label{fig:move5-to-faultfinish-case2-2}
\end{minipage}\hfill
\begin{minipage}[b]{0.32\linewidth}
\centering
\includegraphics[width=\linewidth]{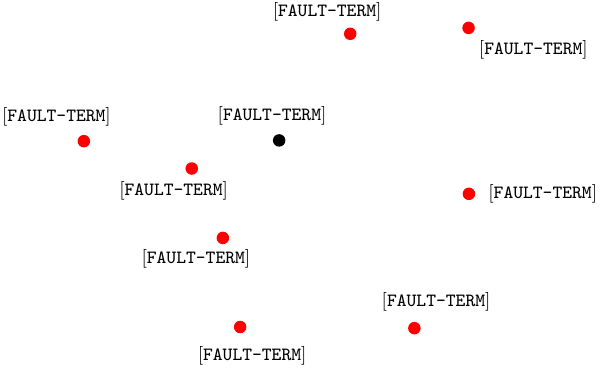}
\caption{All non-faulty robots terminate on a point having faulty robots}
\label{fig:final-faulterminate-case2-2}
\end{minipage}
\Description{A technical figure}
\end{figure}

\paragraph{\underline{\textbf{Stage} \textmd{\textsc{Move-To-GatheringPoint\&Terminate}}:}} Due to asynchronous scheduling, a non-faulty robot may have already established a gathering point before $r$. If this point is visible to $r$, it gives priority to the gathering step over any other actions discussed above.
\begin{itemize}[leftmargin=*]
    \item \textbf{$r.color \neq $ \texttt{FAULT-FINISH} and it finds a \texttt{GATHER}-colored robot:} If $r$ is already collocated with a \texttt{GATHER}-colored robot, it changes its color to \texttt{GATHER}. 
    Otherwise, $r$ moves to the position of the \texttt{GATHER}-colored robot after changing the current color to \texttt{MoveTo-GATHER} (refer to Fig. \ref{fig:gather-at-corner-2}).

    \item \textbf{$r.color =$ \texttt{MoveTo-GATHER}:} If $r$ becomes collocated with a \texttt{GATHER}-colored robot, it changes its color to \texttt{GATHER}. Otherwise, it switches to \texttt{FAULT-FINISH}, indicating that it has become faulty.




    \item \textbf{$r.color = $ \texttt{FAULT-FINISH}:} In this case, $r$ first computes the local layers $\mathcal{CH}_r^0, \mathcal{CH}_r^1,\cdots, \mathcal{CH}_r^{l'} $ and determines its own layer $\mathcal{CH}_r^i$, where it belongs. It selects the two neighbours $r_{Nbr_1}$ and $r_{Nbr_2}$ on $\mathcal{CH}_r^i$ and computes the two half-planes delimited by the line $\overleftrightarrow{r_{Nbr_1}r_{Nbr_2}}$. If one of the half-planes is either empty or contains the robots with the color \texttt{FAULT-TERMINATE}, it changes its current color to \texttt{FAULT-TERMINATE} and moves to the position of a \texttt{FAULT-FINISH} or \texttt{FAULT-TERMINATE} robot other than its own (if exists).
    This movement is to ensure that if the non-faulty robots are gathered at a point and colored \texttt{FAULT-FINISH}, they finally gather with a faulty robot (see Fig. \ref{fig:final-faulterminate-case2-2}). 

    \item \textbf{$r.color =$ \texttt{GATHER}:} If all the visible robots other than their own position have the color \texttt{FAULT-TERMINATE}, $r$ switches to \texttt{TERMINATE} and moves to nearest \texttt{FAULT-TERMINATE} robot (see Fig. \ref{fig:terminate-to-fault-2}).

    \item \textbf{Termination Condition:} If $r.color$ is either set to \texttt{TERMINATE} or \texttt{FAULT-TERMINATE}, the robot $r$ terminates. 
\end{itemize}

\noindent \textbf{Case 2: (There are no \texttt{INTERIOR}-colored robots)}
This specific configuration may either be the initial configuration of the robots or arise in cases where none of the interior robots become faulty.
In this case, each \texttt{CORNER}-colored robot that has a \texttt{BOUNDARY}-colored neighbor moves to its respective exterior visible area.
This movement causes the \texttt{BOUNDARY}-colored robots to become interior robots. In the subsequent epoch, they update their color to \texttt{INTERIOR}. 
We now proceed with the following steps. 
\begin{itemize}[leftmargin=*]
    \item \textbf{$r.color =$ \texttt{CORNER}:} The robot $r$ first checks whether it has any \texttt{BOUNDARY}-colored neighbour. If it does, it moves to a point in $ExtVisibleArea(r,\mathcal{CH}_r)$  (as in Fig. \ref{fig:visiblearea}) with the color \texttt{EXPANDING} in order to expand the convex hull and convert the \texttt{BOUNDARY}-colored robots into interior ones.
    Otherwise, if there is a \texttt{BOUNDARY}-colored robot in $\mathcal{V}_r$, $r$ maintains the status quo.
    In case of $r$ not observing any \texttt{BOUNDARY}-colored robot, it changes its color to  \texttt{INTERMEDIATE} without any movement.

    \item \textbf{$r.color =$ \texttt{EXPANDING}:} If both the neighbours of the robot $r$ do not have the color \texttt{BOUNDARY}, then it changes its color to \texttt{CORNER}. 
    Otherwise, it changes its color to \texttt{FAULT1} to indicate that it has become faulty. 
    Any robot that observes a \texttt{FAULT1}-colored robot will always disregard it in all future computations.

    \item $r.color =$ \texttt{INTERMEDIATE}: $r$ changes its color to \texttt{CORNER}, if it finds any robot with the color in $\{$ \texttt{MOVE4, MOVE-ENDED}, \texttt{MOVE5, MoveTo-CORNER}$\}$.
    Otherwise, it moves to the CG of the convex hull formed by all \texttt{CORNER} and \texttt{INTERMEDIATE}-colored robots after changing its color to \texttt{MOVE4}. When $r$ becomes a robot with color \texttt{MOVE4}, it follows the steps mentioned in Case 1.

   \item \textbf{$r.color =$ \texttt{BOUNDARY}:} If at least one of $r$’s neighbors has the color \texttt{EXPANDING}, then $r$ takes no action. Otherwise, it computes the convex hull $\mathcal{CH}_r^*$ formed by all visible robots excluding those with color \texttt{FAULT1}. 
   $r$ switches its color to \texttt{CORNER} if it lies on a corner, to \texttt{INTERIOR} if it is an interior robot, and retains the color \texttt{BOUNDARY} if it remains a boundary robot on $\mathcal{CH}_r^*$.
\end{itemize}

In the above steps, whenever a \texttt{CORNER}-colored robot has a neighbouring robot with the color \texttt{BOUNDARY}, it moves in such a way that the neighbouring boundary robot becomes an interior robot. If the \texttt{CORNER}-colored robot is faulty and unable to perform this movement, it eventually changes its color to \texttt{FAULT1}. This change signals that the neighbouring \texttt{BOUNDARY}-colored robot has now become a corner robot on the convex hull formed by the robots that do not have the color \texttt{FAULT1}.
Such \texttt{BOUNDARY}-colored robots then change  to \texttt{CORNER} and, if necessary, perform similar movements to convert their own neighbouring boundary robots into interior robots.

\medskip
\noindent \textbf{Case 3: (The linear configuration)} 
This case arises when every robot sees at most two other robots.
We call a robot a \emph{terminal robot} if it has exactly one neighbour not colored \texttt{FAULT1}, while the rest are \emph{non-terminal}.
For a robot $r$, let $\mathcal{CH}_r^*$ be the convex hull excluding the \texttt{FAULT1}-colored robots. 
When $\mathcal{CH}_r^*$ is a line, we denote it as $Linear\_\mathcal{CH}_r^*$.
Else, we use $\mathcal{CH}_r^*$. 
The goal is to transform the linear configuration to non-linear (Case 1 or 2).

In this process, a terminal robot first changes its color to \texttt{CORNER}, while a non-terminal one to \texttt{NON-CORNER}. 
A terminal \texttt{CORNER}-colored robot $r$ checks the color of its neighbour $r_{Nbr_1}$. 
When the system has only two non-faulty robots or a configuration in which robots are distributed in two distinct points on the plane, a robot $r$ finds $r_{Nbr_1}$ as \texttt{CORNER}.
In this case, the algorithm continues as in Section \ref{sec:impossibility} using \texttt{MOVE1} and \texttt{MOVE-ENDED} in place of \texttt{MOVE} and \texttt{END}, respectively.
When the robots are distributed at two distinct points, and each point has multiple robots, it is possible that some of the (non-faulty) robots move to the midpoint of the two-point configuration, leaving the faulty ones behind. 
In such a case, a robot that finds two neighbors in the configuration switches its color to \texttt{BOUNDARY}, the terminal ones switch back to \texttt{CORNER}.

Otherwise, when \texttt{CORNER}-colored terminal robot $r$ finds $r_{Nbr_1}$ as \texttt{BOUNDARY}, it moves perpendicularly away from $Linear\_\mathcal{CH}_r^*$ with color \texttt{EXPANDING-L}.
After the movement, $r$ might still find itself on $Linear\_\mathcal{CH}_r^*$, either because of the fault or due to a symmetric movement of the two terminal robots (it arises when there are three robots). 
Using a color transition, the neighbour $r_{Nbr_1}$ helps it to figure out the cause of such a configuration.
When both the neighbours $r_{Nbr_1}$ have the color \texttt{EXPANDING-L}, $r_{Nbr_1}$ changes its color to \texttt{GATHER}, so that $r$ can move to its position for gathering.
When $r_{Nbr_1}$ finds only one neighbour with color \texttt{EXPANDING-L}, it switches to \texttt{BOUNDARY} to signal $r$ that the movement of $r$ was unsuccessful and hence $r$ changes to \texttt{FAULT1}.
In that case, $r_{Nbr_1}$ finds itself as a terminal on the linear convex hull and eventually updates the color to \texttt{CORNER}.
This process repeats until the configuration becomes non-linear or only 2 or 3 non-faulty robots remain. 
Finally, if all other robots are \texttt{FAULT1}, a robot switches to \texttt{TERMINATE} and moves to the nearest \texttt{FAULT1} robot to collocate with it. 

\subsection{Analysis of 26-$\mathcal{LUMI}$} To prove Theorem \ref{thm:l-f-gather}, we need the following $7$ lemmas.

\begin{restatable}{lem}{robotcategorisation}\label{lemma:categorization}
    In stage \textsc{Robot-Categorisation}, all corner, boundary, and interior robots on $\mathcal{CH}$ correctly color themselves with \texttt{CORNER}, \texttt{BOUNDARY}, and \texttt{INTERIOR} respectively in $O(1)$ epochs.
\end{restatable}

\begin{proof}

During this stage, no robot has moved, so
$\CH$ is the initial convex hull. We show that the three (exhaustive and
mutually exclusive) global cases yield the corresponding local
outcomes.

\emph{$r$ is a corner of $\CH$.} There exist two distinct lines
$L^1_r,L^2_r$ through $r$\footnote{The lines supporting $\CH$ along the
two hull edges incident to $r$; if the configuration is linear and $r$
is terminal, take the line through the configuration and any other line
through $r$.} such that all robots, hence all visible robots, lie in
$\overline{\mathcal{H}}{}^{i}_{L^1_r}\cap
\overline{\mathcal{H}}{}^{j}_{L^2_r}$ for some $i,j\in\{1,2\}$. So $r$
is a corner robot on $\CHr$.

\emph{$r$ is a boundary robot of $\CH$.} Let $L_r$ be the line through
the convex hull edge of $\CH$ containing $r$.
All visible robots lie in one
closed half-plane of $L_r$, so $r$ is a boundary or corner robot on
$\CHr$. The corner case can be easily excluded by the following argument: the two endpoints of that edge are
robot positions on $L_r$, and the robot
nearest to $r$ in each direction along $L_r$ is visible to $r$.
If two lines $L^1_r,L^2_r$ through $r$ intersecting at an angle $< \pi$ exist with all visible
robots in $S = \overline{\mathcal{H}}{}^{i}_{L^1_r}\cap
\overline{\mathcal{H}}{}^{j}_{L^2_r}$, the
region $S$ cannot contain the visible robots on $L_r$.
Hence $r$ is a boundary robot on $\CHr$.

\emph{$r$ is an interior robot of $\CH$.} Then for every line $L$
through $r$, some robot lies strictly on each side of $L$,\footnote{Otherwise
$L$ would support $\CH$ at $r$, placing $r$ on the boundary of $\CH$.}
and some \emph{visible} robot does too. So no closed
half-plane delimited by a line through $r$ contains all visible robots.
Hence, $r$ is
an interior robot on $\CHr$.

Each robot resolves its classification at its first activation, so all
colors are set within just  $1$ epoch.
\end{proof}

\begin{restatable}{lem}{corollaryfirstlayercategorisarion}
    \label{cor:ch_1-categorization}
    After stage \textsc{Robot-Categorisation}, an \texttt{INTERIOR} colored robot $r$ correctly detects whether it is a part of $\mathcal{CH}^1$. 
\end{restatable}
\begin{proof}
    By Lemma~\ref{lemma:categorization}, the \texttt{INTERIOR}-colored
    robots are exactly the robots lying inside $\mathcal{CH}$, and
    $\mathcal{CH}^1$ is the boundary of their convex hull. The robot $r$
    checks its position on $\mathcal{CH}_r^*$, the convex hull of the
    \texttt{INTERIOR}-colored robots that $r$ can see. If $r$ is on the
    boundary of $\mathcal{CH}_r^*$ (as a corner or boundary robot), it
    reports that it is on $\mathcal{CH}^1$, and otherwise that it is not.

    We show this test is correct, that is, $r$ is on the boundary of
    $\mathcal{CH}_r^*$ if and only if it is on $\mathcal{CH}^1$. If $r$
    is on $\mathcal{CH}^1$, then some line $L$ through $r$ keeps all
    \texttt{INTERIOR}-colored robots on one closed half-plane delimited by $L$; in particular
    all such robots visible to $r$ stay on that half plane, so $r$ is on the
    boundary of $\mathcal{CH}_r^*$. Conversely, suppose $r$ is not on
    $\mathcal{CH}^1$, so $r$ is strictly inside the hull of all
    \texttt{INTERIOR}-colored robots. Then every line $L$ through $r$
    has an \texttt{INTERIOR}-colored robot strictly on each half plane delimited by $L$. Hence no line
    through $r$ keeps the visible \texttt{INTERIOR}-colored robots on
    one side, so $r$ is strictly inside $\mathcal{CH}_r^*$ and reports
    correctly.
\end{proof}

\begin{restatable}{lem}{lemmaBoundarytointerior}
    \label{lemma:boundary-to-interior}
    In the stage \textsc{Boundary-To-Interior}, all non-faulty boundary robots on $\mathcal{CH}^0$ become collocated with the interior robots of $\mathcal{CH}$ in $O(1)$ epochs.
\end{restatable}
\begin{proof}
    A \texttt{BOUNDARY}-colored robot moves to the position of the nearest \texttt{INTERIOR}-colored robot with the color \texttt{MoveTo-INTERIOR}.
    We argue that such a robot does not experience any delay in initiating its movement due to obstruction by other robots.
    Assume for contradiction that a \texttt{BOUNDARY}-colored robot $r$ is unable to see its closest \texttt{INTERIOR}-colored robot because another robot $r'$ is blocking its line of sight. 
    If both of these robots move toward the same robot $r_1$, the robot $r'$ cannot lie on the line segment $\overline{rr_1}$, as they started from different locations and hence the lines $\overleftrightarrow{rr_1}$ and $\overleftrightarrow{r'r_1}$ do not intersect except at $r_1$. 
    This implies that $r$ and $r'$ must have different destinations $r_1$ and $r_2$, respectively. 
    Let $r'$ start from a previous position $r'^{prev}$ towards $r_2$.
    Since $r$ moves toward $r_1$, we have $d(r,r_1) \leq d(r, r_2)$ and also $d(r'^{prev}, r_2) \leq d(r'^{prev}, r_1)$. These two inequalities give  $d(r, r_1)+d(r'^{prev}, r_2) \leq d(r, r_2)+ d(r'^{prev}, r_1)$, which contradicts the geometric fact that the sum of the lengths of two opposite sides of a quadrilateral $rr'^{prev} r_1 r_2$ is lesser than the sum of the lengths of its two diagonals.
    Hence, no such obstruction occurs and all the non-faulty boundary robots move to the position of interior robots and switch their color to \texttt{INTERIOR} within $2$ epochs.
\end{proof}

\begin{restatable}{rem}{RemarkTwoELigibleWaitFree}
    \label{rem:two-eligible-wait-free}
    Any two eligible robots from the same layer do not wait for one another to move to the position of \texttt{CORNER}-colored robots.
\end{restatable}

\begin{restatable}{lem}{lemmaInteriorToCorners}\label{lemma:interior-to-corners}
    In the stage $\textsc{Interior-To-Corners}$, all the non-faulty\\ \texttt{INTERIOR}-colored robots move to the corners of $\mathcal{CH}^0$ in $O(\ell)$ epochs.
\end{restatable}
\begin{proof}
    We begin with the following claim.
    
    \noindent \textit{Claim: In the stage $\textsc{Interior-To-Corners}$, all the non-faulty \texttt{INTERIOR}-colored robots on $\mathcal{CH}^1$ move to the corners of $\mathcal{CH}^0$ in $O(1)$ epochs.} 

    \noindent \textit{Proof:} After the stage \textsc{Boundary-To-Interior} is over, all robots lying on $\mathcal{CH}^1$ correctly determine their eligibility within a single epoch as ensured by Corollary \ref{cor:ch_1-categorization}: eligible ones adopt the color \texttt{NEXT}, whereas the non-eligible ones (i.e., those lying in the interior of $\mathcal{CH}^1$, where by interior of a layer, we mean the open region enclosed by the layer) update to \texttt{NOT-NEXT}.
    Subsequently, all the \texttt{NEXT}-colored robots transition to \texttt{ELIGIBLE}, after which all \texttt{NOT-NEXT}-colored robots change their color to \texttt{INTERIOR}. 
    Once this color updation is complete, an eligible robot $r$ on $\mathcal{CH}^1$ computes $BoundaryVisibleArea(r, \mathcal{CH}_r^*)$, changes its color to \texttt{MOVE1} and moves to a point within this area.
    By Lemma \ref{lemma:visiblearea}, the robot with color \texttt{MOVE1} sees all the stationary \texttt{CORNER}-colored robots lying on the open half-plane $H_{L_r}$, unless blocked by a simultaneously moving \texttt{MOVE2}-colored robot. 
    This potential obstruction arises when another \texttt{MOVE2} robot $r'$ is moving toward its designated \texttt{CORNER}-colored target. Nevertheless, we prove that each \texttt{MOVE1}-colored robot $r$ on $\mathcal{CH}^1$ always finds a valid corner to move to (i.e., all \texttt{CORNER}-colored robots cannot be blocked). To prove this, we analyse two cases depending on whether $r$ is a boundary or corner robot on $\mathcal{CH}_r^*$. 

    Let $r$ be the boundary robot on $\mathcal{CH}_r^*$ and $r_1$ denote the nearest \texttt{CORNER}-colored robot to $r$ in the open half-plane $\mathcal{H}_{L_r}$.
    Suppose, for contradiction, that $r_1$ is not visible to $r$, implying a \texttt{MOVE2}-colored robot $r_2$ lies between them and blocks the view.
    Let $r_2^{prev}$ be the previous position of $r_2$ before its movement.
    If $r_1$ were the target for $r_2^{prev}$, then $r_2^{prev}$ would lie on the line $\overleftrightarrow{rr_1}$ and hence $r$ would lie on $\overleftrightarrow{r_1 r_2^{prev}}$, which is not possible by Lemma \ref{lemma:visiblearea}.
    This means that $r_2^{prev}$ has a different \texttt{CORNER}-colored robot as its target other than $r_1$, say $r_3$. 
    Since $r_1$ is the nearest \texttt{CORNER}-colored robot to $r$ and $r_3$ is the nearest to $r_2^{prev}$, we get $d(r,r_1) \leq d(r, r_3)$ and $d(r_2^{prev}, r_3) \leq d(r_2^{prev}, r_1)$, implying $d(r, r_1)+ d(r_2^{prev}, r_3) \leq d(r, r_3) + d(r_2^{prev}, r_1)$, which is a contradiction to the fact that the sum of two opposite sides of the quadrilateral $rr_2^{prev}r_1r_3$ is lesser than the sum of the two diagonals. 
    Additionally, suppose there exists another robot $r'$ that shares the same target $r_1$ as $r$. Since robots move directly toward their target, the paths $\overline{rr_1}$ and $\overline{r'r_1}$ can only intersect at $r_1$, not before. 
    Therefore, $r_1$ is visible to the robot $r$.

    Now consider the case where the \texttt{MOVE1}-colored robot $r$ is a corner robot on $\mathcal{CH}_r^*$. 
    It is possible that $r$ was a boundary robot on $\mathcal{CH}_r^*$ in its previous LCM cycle, but became a corner robot on $\mathcal{CH}_r^*$ in the current cycle due to movements of its neighbouring robots, potentially altering the structure of $\mathcal{CH}_r^*$ (see the Fig. \ref{fig:eligible-to-corner-2}). 
    Consequently, the currently computed open half-plane $\mathcal{H}_{L_r}$ differs from its previous counterpart $H_{L_r}^{prev}$.
    Observe that $\mathcal{H}_{L_r}^{prev} \subset \mathcal{H}_{L_r}' \cup \mathcal{H}_{L_r}''$.
    We claim that one of the half-planes $\mathcal{H}_{L_r}'$ and $\mathcal{H}_{L_r}''$ contains a visible \texttt{CORNER}-colored robot. 
    W.l.o.g, let $\mathcal{H}_{L_r}'$ contain exactly one such robot, say $r_1$.
    If $r_1$ is visible to $r$, the claim is established. 
    Otherwise, let $r_2$ be a \texttt{MOVE2}-colored robot  obstructing the view between $r$ and $r_1$.
    Then, its previous position $r_2^{prev}$ must lie on the line segment $\overline{rr_1}$, implying that $r$ and $r_2^{prev}$ were neighbours on $\mathcal{CH}_r^*$ in the previous cycle and $r$ was a boundary robot then.
    In that case, $r$ considers the other half-plane $\mathcal{H}_{L_r}''$ as $\mathcal{H}_{L_r}$. By applying the same reasoning as in the last paragraph, namely the geometric constraints and the inequalities, we can conclude that $\mathcal{H}_{L_r}''$ (and hence $\mathcal{H}_{L_r}$) has a \texttt{CORNER}-colored robot.

    Hence, an eligible robot $r$ colored \texttt{MOVE1} on $\mathcal{CH}^1$ always identifies a \texttt{CORNER}-colored robot as its valid target to move to and proceeds without waiting for another eligible robot, regardless of whether it is a boundary or a corner of $\mathcal{CH}_r^*$.
    Therefore, in the stage \textsc{Interior-To-Corners}, any non-faulty \texttt{INTERIOR} robot on $\mathcal{CH}^1$ moves to the corners of $\mathcal{CH}^0$ in $6$ epochs following a sequence of color transitions: \texttt{INTERIOR} $\rightarrow$ \texttt{NEXT} $\rightarrow$ \texttt{ELIGIBLE} $\rightarrow$ \texttt{MOVE1} $\rightarrow$ \texttt{MOVE2} $\rightarrow$ \texttt{CORNER} with one additional epoch potentially required when an \texttt{ELIGIBLE} robot waits for visible \texttt{NOT-NEXT} robots to update their color. This concludes the proof of claim. 
    
    
    We now extend this claim to the robots on $\mathcal{CH}^2$. 
    We prove that once all non-faulty robots on $\mathcal{CH}^1$ settle on the corner of $\mathcal{CH}^0$, any non-faulty \texttt{INTERIOR}-colored robot on $\mathcal{CH}^2$ also reaches $\mathcal{CH}^0$ in $O(1)$ epochs.
    Such a robot $r$, when it becomes eligible, undergoes a fixed sequence of color transitions (mentioned in above claim) and updates its color to \texttt{MOVE1} after $4$ epochs. 
    When activated with the color \texttt{MOVE1}, $r$ searches for a \texttt{CORNER}-colored robot on $\mathcal{H}_{L_r}$. 
    If such a robot is visible, it reaches $\mathcal{CH}^0$ within $2$ additional epochs and changes its color to \texttt{CORNER}. 
    Therefore, in this case, $r$, starting from $\mathcal{CH}^2$, settles on $\mathcal{CH}^0$ within $6$ epochs.
    Now consider the case when $r$ finds no \texttt{CORNER}-colored robot in $\mathcal{H}_{L_r}$. By construction of the half planes, $\mathcal{H}_{L_r}$ must contain a \texttt{CORNER} robot, say $r_1$. 
    This implies that there is another robot $r_2$ lying between $r$ and $r_1$. 
    By Lemma \ref{lemma:visiblearea}, $r_2$ cannot be a \texttt{FAULT}-colored (stationary) robot, hence $r_2.color =$ \texttt{MOVE2}. 
    Using a similar argument as presented in the proof of the above claim and by Remark \ref{rem:two-eligible-wait-free}, $r_2^{prev}$, the previous position of $r_2$ can not lie on $\mathcal{CH}^2$, and must therefore belong to $\mathcal{CH}^1$. 
    This contradicts our assumption that all non-faulty robots from $\mathcal{CH}^1$ have already settled on the corners of $\mathcal{CH}^0$.
    Therefore, if a \texttt{MOVE1} robot $r$ does not see any \texttt{CORNER}-robot and reverts its color to \texttt{ELIGIBLE}, our claim still holds since all such color transitions of $r$ hide behind the $6$ epochs that was taken by the robots of $\mathcal{CH}^1$ to reach $\mathcal{CH}^0$.

    By extending the above arguments, we conclude that every layer $\mathcal{CH}^i$ (for $1\leq i \leq l-1$) requires $O(1)$ epochs to move to the corners of $\mathcal{CH}^0$ once they become eligible and all the non-faulty robots on $\bigcup_{j=1}^{i-1}\mathcal{CH}^j$ have already settled on $\mathcal{CH}^0$. Hence, the entire stage \textsc{Interior-To-Corners} takes $O(l)$ epochs, as there are $(l-1)$ many layers made by the \texttt{INTERIOR}-colored robots at the beginning of the stage. 
\end{proof}

\begin{restatable}{lem}{LemmaStageConfirmationToCorners}\label{lemma:stage_confirmation_to_corners}
    In Stage \textsc{Confirmation-Signal-To-Corners}, all \texttt{FAULT}-colored robots change their color to \texttt{FAULT-FINISH} in $O(f)$ epochs. 
\end{restatable}
\begin{proof}
    There could be at most $f$ many \texttt{FAULT}-colored robots located within the interior region of $\mathcal{CH}$. All such robots remain stationary throughout. 
    In the worst case, these robots can form $O(f)$ many layers. 
    Among these layers, the robots on the innermost layer first change their color to \texttt{FAULT-FINISH} in just 1 epoch. 
    Subsequently, the robots in the next layer update their color to \texttt{FAULT-FINISH} in the following epoch.
    This layer-by-layer propagation continues outward until all \texttt{FAULT}-colored robots have transitioned to \texttt{FAULT-FINISH}, resulting in $O(f)$ epochs in total.
\end{proof}

\begin{restatable}{rem}{RemarkCornerSeeFaultFinish}
    \label{rem:corner-see-fault-finish}
    A \texttt{CORNER}-colored robot $r$ correctly determines that no\\ \texttt{INTERIOR}-colored robots left to move onto the corners of $\mathcal{CH}^0$ by checking the interior region of $\mathcal{CH}_r$ empty or occupied by \texttt{FAULT-FINISH}-colored robots.
\end{restatable}

\begin{restatable}{lem}{LemmaUniqueGatherPointFAULTFINISH}
    \label{lemma:unique-gather-point-fault-finish}
    All non-faulty robots gather at a unique point and terminate in $O(f)$ epochs when there are \texttt{FAULT-FINISH}-colored robots after the stage \textsc{Confirmation-Signal-To-Corners}. 
\end{restatable}
\begin{proof}
    After Stage \textsc{Confirmation-Signal-To-Corners}, a \texttt{CORNER} robot $r$ changes its color to \texttt{INTERMEDIATE} and then computes $ExtVisibleArea(r, \mathcal{CH}_r)$ and moves within it with the color \texttt{MOVE3}. By Lemma \ref{lemma:visiblearea}, the \texttt{MOVE3}-colored robot $r$ can see all the \texttt{FAULT-FINISH}-colored robots, as they are stationary. 
    If there is another \texttt{MOVE3}-colored non-faulty robot $r'$, all the \texttt{FAULT-FINISH} robots are also visible to it.
    By Remark \ref{rem:corner-see-fault-finish}, when $r$ and $r'$ are both with the color \texttt{CORNER} and execute the stage \textsc{Corner-To-CG}, all the interior robots of $\mathcal{CH}$ are of color \texttt{FAULT} or \texttt{FAULT-FINISH}, indicating that they are stationary.
    So, two \texttt{MOVE3}-colored non-faulty robots $r$ and $r'$ have the same set of visible \texttt{FAULT-FINISH} robots.
    It implies that any two simultaneously activated non-faulty robots compute the same CG and move to it with \texttt{MOVE4} from their respective visible area.
    Indeed, since the \texttt{FAULT-FINISH}-colored robots are stationary and visible to both, the two robots take the CG of the same set of points, and hence obtain the same target.
    This concludes that if all \texttt{CORNER}-colored robots are non-faulty and activated simultaneously, they gather at a single point.
    After the non-faulty robots change their color to \texttt{MOVE-ENDED}, they move to a point in $BoundaryVisibleArea()$ with \texttt{MOVE5} and do not find any \texttt{CORNER} robot.
    In that case, they update their color to \texttt{MoveTo-CORNER} and move to a position containing \texttt{FAULT-FINISH}-colored robot. Finally, they terminate with the color \texttt{TERMINATE}. 
    It is clear after observing the color transitions of a non-faulty \texttt{CORNER} robot that it requires $7$ epochs to terminate. 

    Consider a situation where a subset of \texttt{MOVE3} robots, which includes a non-faulty robot $r$, simultaneously move to the CG while others remain stationary due to asynchronous activation.
    Using a similar argument as above, $r$ successfully moves to the CG with all other non-faulty robots (if they exist). 
    After $2$ epochs, when such a non-faulty robot updates its color to \texttt{MOVE5} and moves to a point in $BoundaryVisibleArea()$, it finds a \texttt{CORNER}-colored robot, say $r_1$ on $\mathcal{CH}^0$, which is still stationary from the time $r$ started its movement to the CG.
    Such a $r_1$ exists because not all \texttt{CORNER}-colored robots were activated, and a \texttt{CORNER}-colored robot does not move while it observes a robot colored \texttt{MOVE3}, \texttt{MOVE4}, \texttt{MOVE-ENDED}, \texttt{MOVE5}, or \texttt{MoveTo-CORNER}.
    Finally, $r$ along with all other non-faulty robots moves to $r_1$ with the color \texttt{GATHER}.  
    This position becomes the gathering point for other non-faulty \texttt{CORNER} robots. 
    Meanwhile, the \texttt{FAULT-FINISH}-colored robots change their color to \texttt{FAULT-TERMINATE} in $O(f)$ epochs, by an argument similar to Lemma \ref{lemma:stage_confirmation_to_corners}.
    We stress that a \texttt{FAULT-TERMINATE}-colored robot need not be faulty: by the \texttt{FAULT-FINISH} rule of the stage \textsc{Move-To-GatheringPoint\&Terminate}, a group of non-faulty robots that has gathered at the CG also turns \texttt{FAULT-TERMINATE}. The gathering point is nonetheless the location of a faulty robot, for the following reason. A non-faulty \texttt{FAULT-TERMINATE} robot always finds another \texttt{FAULT-FINISH} or \texttt{FAULT-TERMINATE} robot and moves onto its position, whereas a faulty robot that turns \texttt{FAULT-TERMINATE} cannot move and stays at its location. This case has at least one faulty robot, since its confirmation produced the \texttt{FAULT-FINISH} robots that define the CG in the first place; hence a stationary \texttt{FAULT-TERMINATE} anchor exists, and the half-plane condition of the \texttt{FAULT-FINISH} rule drives all \texttt{FAULT-FINISH} and \texttt{FAULT-TERMINATE} robots onto a single such anchor.
    Finally, each \texttt{GATHER}-colored robot, once every other visible robot is colored \texttt{FAULT-TERMINATE}, moves to that anchor and terminates. Hence all non-faulty robots terminate at the position of a faulty robot, as required.
    In total, this requires $O(f)$ epochs in the worst case.

    If all the simultaneously activated robots moving towards the CG are faulty, they all set their color to \texttt{FAULT-FINISH} in $6$ epochs after they get activated with \texttt{MoveTo-CORNER} yet collocated with neither \texttt{CORNER} nor \texttt{FAULT-FINISH} robots.
    Afterwards, the remaining \texttt{CORNER}-colored robots will get a chance to move to the CG.
    In the worst case, a constant number of faulty robots get activated to move to the CG, leading to the termination in $O(f)$ epochs. 
\end{proof}

\begin{restatable}{lem}{LemmaUniqueGatherPointNoFaultFINISH}
    \label{lemma:unique-gather-point-no-fault-finish}
    All non-faulty robots gather at a point and terminate in $O(f)$ epochs, when there is no \texttt{FAULT-FINISH}-colored robot after the stage \textsc{Confirmation-Signal-To-Corners}. 
\end{restatable}
\begin{proof}
    In such a situation, all the robots become \texttt{CORNER}-colored robots. 
    If all of them are activated simultaneously, such a robot considers the CG of the convex hull formed by all \texttt{CORNER} and \texttt{INTERMEDIATE} robots. 
    The non-faulty robots gather at the CG, and all robots update their color to \texttt{FAULT-FINISH}. 
    If there is any faulty robot, there are at least two distinct positions occupied by \texttt{FAULT-FINISH} robots.
    In that case, the non-faulty robots move to another position of \texttt{FAULT-FINISH} robots and terminate with \texttt{FAULT-TERMINATE}.
    This requires $7$ epochs in total.

    If all the robots are not activated simultaneously, a robot that moves to the CG with \texttt{MOVE4} and gets activated with \texttt{MOVE5} after $2$ epochs, finds at least one \texttt{CORNER} or \texttt{FAULT-FINISH}-colored robot, which is a similar case as argued in Lemma \ref{lemma:unique-gather-point-fault-finish}, requiring $O(f)$ epochs. 
\end{proof}

\begin{restatable}{theo}{theoremLFgather}\label{thm:l-f-gather}
    The above algorithm solves indulgent gathering in an $(N,f)$-mobility system, $f<N$, $f,N$ unknown, under {\async} in 26-$\mathcal{LUMI}$, has runtime $O(\max\{\ell,f\})$ epochs, and works under obstructed visibility. 
\end{restatable}

\begin{proof}
Let $\mathcal{CH}$ be the convex hull of the robots. We show that the non-faulty
robots terminate and gather within $O(\max\{\ell,f\})$ epochs, treating the
non-linear and linear configurations separately.

First, suppose $\mathcal{CH}$ is non-linear. If $\mathcal{CH}$ has an interior
robot, the robots are driven from the boundary inward to the corners and then to a
single gathering point; this is established by Lemmas~\ref{lemma:categorization} to
\ref{lemma:unique-gather-point-no-fault-finish}, and takes $O(\max\{\ell,f\})$
epochs. If $\mathcal{CH}$ has no interior robot, the \texttt{CORNER}-colored robots
first move to their respective $ExtVisibleArea()$ so as to turn the boundary robots
into interior ones. If some of these moves succeed without a fault, the neighbouring
boundary robots become interior and we are back to the previous situation. If
instead all corners are faulty (and hence eventually colored \texttt{FAULT1}), their
neighbouring boundary robots become the corners of the convex hull formed by the
visible robots other than the \texttt{FAULT1} ones, and these move to their exterior
visible area to interiorise the current boundary. This may repeat---each new corner
possibly becoming faulty and passing the role to its neighbour---until no boundary
robot remains or the boundary robots have become interior, which takes $O(f)$
epochs. In the former case the newly created interior robots form a single layer, so
by the lemmas above termination follows in $O(f)$ further epochs; in the latter case
termination follows in $O(f)$ epochs, as in
Lemma~\ref{lemma:unique-gather-point-no-fault-finish}. Thus the non-linear case
costs $O(\max\{\ell,f\})$ epochs.

Now, suppose $\mathcal{CH}$ is linear. The two terminal robots take color
\texttt{CORNER} and the rest \texttt{NON-CORNER}, and we distinguish $N=2,3$,
and $N>3$.
When $N=2$, a terminal robot $r$ (\texttt{CORNER}) sees its only neighbour also
\texttt{CORNER}; it then runs the two-robot procedure of
Section~\ref{sec:impossibility} with \texttt{MOVE1} and \texttt{MOVE-ENDED} in the
roles of \texttt{MOVE} and \texttt{END}, and
Theorem~\ref{thm:gathering2robots} gives gathering.

When $N=3$, the terminal \texttt{CORNER} robot $r$ finds its neighbour $r_{Nbr_1}$
colored \texttt{NON-CORNER} and moves perpendicularly away from
$\overline{r\,r_{Nbr_1}}$ with color \texttt{EXPANDING-L}. If the other terminal
$r'$ stays idle while $r$ moves, the configuration becomes non-linear with at most
two layers, and by the non-linear case above the robots gather in $O(1)$ epochs
(here $\ell\le 2$ and $f<N=3$). If $r$ and $r'$ move simultaneously, they may, in the
worst case, become collinear with $r_{Nbr_1}$ again (due to their co-ordinate axis, which is a reflection of the other along $r_{Nbr_1}$ or both being faulty); then $r_{Nbr_1}$, seeing both neighbours
\texttt{EXPANDING-L}, switches to \texttt{GATHER}, marking its own position as the
gathering point, and after the non-faulty terminals collocate there they move on to
a faulty robot's position (if one exists) to terminate. If instead $r$ becomes
faulty while $r'$ is idle, $r$ fails to leave the line; within $3$ epochs it turns
\texttt{FAULT1} and $r_{Nbr_1}$ (\texttt{CORNER}) becomes terminal, which reduces to
the case $N=2$ and gathers in $O(1)$ epochs.

When $N>3$, a terminal \texttt{CORNER} robot $r$ moves perpendicularly away from
$\overleftrightarrow{r\,r_{Nbr_1}}$ with color \texttt{EXPANDING-L}. If $r$ is
non-faulty, its move turns the convex hull non-linear with at most two layers, and
the non-linear case applies. If $r$ is faulty, $\mathcal{CH}$ may stay linear; $r$
turns \texttt{FAULT1} in $2$ epochs and $r_{Nbr_1}$ turns \texttt{CORNER} one epoch
later, after which $r_{Nbr_1}$, now terminal, repeats the same steps. This continues
until the configuration becomes non-linear or reduces to the case $N=2$ or $N=3$. In the worst case, starting from the terminal robots, $f/2$ robots sequentially become \texttt{FAULT1}, leading us to $O(f)$ epochs to get a non-linear configuration and hence to gather.

In either of the cases, the algorithm takes $O(\max\{\ell,f\})$ epochs using only the $26$ colors, which proves the
theorem.
\end{proof}

\section{Conclusion}
\label{sec:discussion}
In this paper, we have considered the indulgent version of gathering with {\async} opaque robots, which are prone to mobility faults. 
The same algorithm also solves the fault-tolerant traditional gathering, where we only gather the non-faulty robots at a single point.
As future work, it would be interesting to improve the time and/or color complexities, as well as to establish time-complexity lower bounds for $(N,f)$-mobility systems.  In another direction, the whole spectrum of fault-tolerant pattern formation is open. 


\bibliographystyle{plainurl}
\bibliography{bibliography}

\end{document}